\documentclass[11pt]{article}
\usepackage[utf8]{inputenc}
\usepackage{mathtools,amsmath,amssymb}
\usepackage{graphicx}
\usepackage{lmodern}
\usepackage[T1]{fontenc}
\usepackage{tikz}
\usepackage[pagebackref=false]{hyperref}
\usepackage{bbm}
\usepackage{bm}
\usepackage{appendix}
\usepackage{comment}
\usepackage{mathtools}
\usepackage{array}
\usepackage{amsthm}
\usepackage[mathscr]{euscript}
\usepackage[left=2cm,right=2cm,top=2cm,bottom=2cm]{geometry}
\usepackage{ulem}

\numberwithin{equation}{section}

\newtheorem{theorem}{Theorem}[section]
\newtheorem{proposition}[theorem]{Proposition}
\newtheorem{lemma}[theorem]{Lemma}
\newtheorem{remark}[theorem]{Remark}
\newtheorem{corollary}[theorem]{Corollary}
\newtheorem{definition}[theorem]{Definition}

\newcommand{\oa}{\mathbf{a}}
\newcommand{\oad}{\mathbf{\bar{a}}}
\newcommand{\of}{\mathbf{a}}
\newcommand{\ofd}{\mathbf{\bar{a}}}

\newcommand{\OA}{\mathbf{A}}
\newcommand{\OAD}{\mathbf{\bar{A}}}

\newcommand{\OF}{\mathbf{A}}
\newcommand{\OFD}{\mathbf{\bar{A}}}

\newcommand{\mmu}{\mu}

\newcommand{\M}{M}
\newcommand{\N}{N}

\newcommand{\OZ}{\bm{\partial}_{\theta}}
\newcommand{\OZD}{\bm{\Theta}}

\newcommand{\OT}{\bm{\partial}_z}
\newcommand{\OTD}{\bm{Z}}

\newcommand{\hmH}{\mathcal{H}}

\newcommand{\ID}{\mathrm{I}}

\newcommand{\eR}{\mathcal{R}}

\newcommand{\abrOSCR}{\mathrm{R}}
\newcommand{\ecR}{\bm{\mathscr{R}}}
\newcommand{\ecdR}{\mathtt{R}}
\newcommand{\fundR}{R}
\newcommand{\eH}{\mathcal{H}}

\newcommand{\abrOSCH}{\mathrm{H}}
\newcommand{\ecH}{\bm{\mathscr{H}}}
\newcommand{\ecdH}{\mathtt{H}}
\newcommand{\eL}{\mathcal{L}}

\newcommand{\ecL}{\bm{\mathscr{L}}}
\newcommand{\ecdL}{\mathtt{L}}
\newcommand{\Khat}{\mathcal{K}^{0}}
\newcommand{\K}{\widetilde{\mathcal{K}}^{0}}
\newcommand{\KhatFull}{\mathcal{K}}
\newcommand{\Kfull}{\widetilde{\mathcal{K}}}
\newcommand{\Khatfund}{K^{0}}
\newcommand{\Kfund}{\widetilde{K}^{0}}
\newcommand{\KhatFundFull}{K}
\newcommand{\KFundFull}{\widetilde{K}}
\newcommand{\eJ}{J}

\newcommand{\hidden}{\theta}
\newcommand{\heat}{z}
\newcommand{\Hidden}{\Theta}
\newcommand{\Heat}{Z}
 \DeclareMathOperator{\spa}{span}

\newcommand{\rep}{\tau}
\newcommand{\dualrep}{\overline{\tau}}

\DeclareMathOperator{\tr}{\text{tr}}

\usepackage{xspace}

\usepackage{graphicx}

\usepackage{subcaption} 

\allowdisplaybreaks
\usepackage{tikz}
\usetikzlibrary{decorations.pathreplacing}
\title{(Integrability of the multispecies harmonic process}
\author{Francesco Casini, Rouven Frassek, Cristian Giardinà}

\begin{document}

\begingroup
\begin{center}

\vspace{5em}
 \begingroup\LARGE
 \bf The boundary-driven multispecies harmonic process
\par\endgroup
 \vspace{3.5em}
 \begingroup\large \bf
Francesco Casini$^{\,a}$, Rouven Frassek$^{\,b}$, Cristian Giardinà$^{\,b}$
 \par\endgroup
\vspace{2em}
\begin{center}
$^{\,a}$ {\normalsize Department of Physics and Astronomy, KU Leuven, Belgium.}
\\[0.1in]
$^{\,b}$ University of Modena and Reggio Emilia, FIM,\\
Via G. Campi 213/b, 41125 Modena, Italy
\\[0.6in]
 \end{center}

\vspace{-2em}

\end{center}
 \begin{abstract}
We introduce the multispecies version of the harmonic process on a one-dimensional chain in contact with boundary reservoirs. This process is a continuous-time Markov chain where each site can host an unbounded number of colored particles. The symmetric bulk dynamics is put in contact with reservoirs, which inject and remove particles driving the system out-of-equilibrium. The Markov generator of the process is identified with the integrable Hamiltonian of an open rational Heisenberg spin chain of higher rank. We derive the underlying R- and K-matrices in operator form and construct the double-row transfer matrix following Sklyanin. Similar to the monospecies case the R-matrix factorises into two R-operators, each factor corresponding to left and right moving particles.
We further define three dual models: an absorbing particle model, a hidden parameter model and a heat conduction model.

\end{abstract}

\addtocontents{toc}{\protect\setcounter{tocdepth}{3}}
\tableofcontents

\section{Introduction and main results} 

\subsection{Integrable stochastic processes and duality}
The Simple Symmetric Exclusion Process (SSEP) plays a central role in modern probability theory and statistical physics \cite{SPITZER1970246,liggett1985interacting}. Its underlying integrable structure, originally unveiled for the asymmetric version ASEP \cite{Gwa-Spohn}, allows to compute microscopic properties exactly, even when the process is driven out of equilibrium by boundary reservoirs. 
For example this led to the stationary measure, density profiles, large deviations and transition probabilities, see e.g. \cite{schutz2000exactly,Derrida,blythe2007nonequilibrium,Mallick} for reviews. It was soon noticed that for a certain choice of transition rates, exact solvability extends also to multispecies generalisations of the  exclusion processes \cite{Alcaraz:1992zc,2009JPhA...42p5004P,crampe2016matrix,crampe2016integrable,vanicat2017exact,casini2024duality}.

In the exclusion process each site can be void or occupied by at most one particle.
Integrable particle processes that also allow for higher-occupancy were formalized in \cite{borodin2016stochastic,2016arXiv160105770B,2016CMaPh.343..651C}  as the stochastic higher-spin six-vertex model. For the case with asymmetry, integrability and duality have been exploited in the
study of scaling limits of these  probabilistic systems, which have been shown to belong to the Kardar-Parisi-Zhang (KPZ) universality class. 
Lattice weights of this type were first constructed by Povolotsky \cite{povolotsky2013integrability} and considered as a model with jumps in both directions in \cite{barraquand2016q}. The process is now referred to as the  q-Hahn process, see also  \cite{8203134}. Here we also point out the work of Bytsko \cite{Bytsko:2001uh} who gave a concise expression of the R-matrix that contains the non-stochastic lattice weights based on the representation worked out for the rational case \cite{Faddeev:1996iy}. The non-stochastic lattice weights were  worked out explicitly by Mangazeev \cite{Mangazeev:2014gwa}. It is by now well understood that those lattice weights can be rendered stochastic by a suitable choice of normalisation and an additional diagonal Drinfeld twist like in the ASEP \cite{deGier:2005zz}, see also \cite{Kuniba:2016fpi} for the discussion of higher spin. For negative spin values one can go to continuous time, and extract the Markov generator from the lattice weights. 
For a particular spin value, the resulting process reduces  to the Multiparticle Asymmetric Diffusion Model (MADM) that was proposed  earlier in \cite{sasamoto1998one}. The multispecies generalisation also known as the coloured stochastic vertex model has been obtained in \cite{Kuniba:2016fpi,2018arXiv180801866B}.  In analogy to the case of ASEP/SSEP the rates of the higher spin stochastic vertex model with negative spin are related to integrable Hamiltonians of non-compact spin chains \cite{Frassek:2019vjt,Frassek:2019isa} that allows to construct commuting transfer matrices following the Quantum Inverse Scattering Method \cite{Faddeev:1996iy} where the stochastic Hamiltonian (Markov generator) of the continuous-time process arises as the logarithmic derivative of the (stochastic) R-matrix. This implementation has been crucial to equip the non-compact continuous-time models as above with integrable boundary reservoirs, which is the main focus of this paper.

In the case of boundary-driven systems, the construction of the stochastic Hamiltonian is based on the work of Sklyanin \cite{Sklyanin:1988yz} who provided a generalisation of the Quantum Inverse Scattering Method (QISM) to 1d quantum integrable systems with boundaries. Combined with stochasticity of the R-matrix and K-matrix it provides a systematic construction of integrable stochastic particle processes on a line with boundary reservoirs in continuous time, see~\cite{crampe2014integrable}. As discussed in detail below, for the non-compact vertex models discussed previously, the classification of stochastic integrable boundaries in terms of K-matrices is still an open problem and only partial results are available \cite{2014JPhA...47C5202L,2019NuPhB.94514665M,2024JPhA...57x5201K,Frassek:2019vjt,Frassek:2022fjs}.

The  analysis of integrable stochastic particle systems often relies on Markov duality relations \cite{liggett1985interacting, schutz1997duality, demasi2006mathematical, dualityBook}. Stochastic duality connects two Markov processes through a duality function,
which serves as ``gate''
to transfer information 
from one process to another.
This allows to express observables of the original system, such as correlation functions, as expectations of a  simpler dual process, often involving finitely many particles. As a result, duality yields concise and exact formulas for the expectations of several relevant observables. A standard example is the computation of the $n$-th moment of an observable, such as the current, or of an $n$-point correlation function in terms of $n$ particles evolving in the dual system. The latter dynamics can then be solved by exploiting integrability  \cite{borodin2014duality}. For boundary-driven systems, the construction of the dual processes requires enlarging the state space by introducing additional sites \cite{Spohn_1983,kipnis1982heat}, see also \cite{2023JPhA...56A4001S}.
Lie-algebraic methods provide a powerful tool for deriving new duality relations, see  \cite{dualityBook} for an introduction, and the existence of a dual process is closely related to the presence of symmetries in the original dynamics.
The formulation of a stochastic process within the algebraic framework of the QISM makes the underlying symmetries apparent and simplifies the search for dualities.

\subsection{The boundary-driven multispecies harmonic process}

In this paper we introduce 
the boundary-driven multispecies version of the harmonic process. The monospecies case has been shown to be integrable in \cite{Frassek:2019vjt}.  
Here we describe the process
in words; the mathematical definition will be given in Section \ref{sub-harmonic}.

We first recall the monospecies case where the bulk dynamics is as follows. If a given site is occupied with $m$ particles then $k$ of them jump symmetrically to its neighbor sites with a rate proportional to an improper Beta-Binomial random variable with parameter $(m;0,2s)$. This means that $k$ particle are selected at random with an $m$ binomial sampling in which the success probability is $P$,
an improper Beta$(0,2s)$. For instance, in the simplest case of spin $s=1/2$, one moves $k$ particles (out of available $m$ ones) at rate $1/k$.  

In the multispecies case one has $M$ possible species at each site.
If a given site is occupied with $m_i$ particles of type $i$, where $i=1,\ldots,M$, then the transitions to a new configuration are dictated by a vector $(k_1,k_2,\ldots,k_M)$ 
which specifies that $k_i$ particles of type $i$ are moved to neighboring sites. A natural choice for
the rates of the multispecies
process would be to take them proportional to an improper
Dirichlet-Multinomial random vector with parameter $(m_1,m_2,\ldots,m_M;0,2s,\ldots,2s)$. This would mean that the vector $(k_1,k_2,\ldots,k_M)$ is selected at random with a $(m_1,m_2,\ldots,m_M)$ multinomial sampling in which the success probability is the vector $(P_1,P_2,\ldots,P_M)$
having an improper Dirichlet $(0,2s,\ldots,2s)$ distribution. However, with this choice of the rates, integrability of the process would be lost.

The choice for the rates which preserve integrability of the harmonic process 
is the one in which 
the vector $(k_1,k_2,\ldots,k_M)$
is obtained by performing 
$M$ binomial samplings,
the first from a sample of size $m_1$, the second from a sample of size $m_2$, etc.. These sampling are however not independent since they all use the same success probability $P$, provided by
an improper Beta$(0,2s)$
random variable as in the monospecies case.
As we shall prove, this choice for the bulk rates follows from  the solution of the Yang-Baxter equation.
This choice of transition rates in the bulk coincides  with the $q\to1$ limit of the integrable rates studied in \cite{Kuniba:2016fpi,2018arXiv180801866B}.
Furthermore, to have a boundary-driven process, we introduce integrable reservoirs, i.e.
processes at the boundaries
which inject and remove particles and fix different average particle densities 
at the boundaries. 
The corresponding intensities follow from the $K$-matrix
and the boundary Yang-Baxter equation. See Section  \ref{sub-harmonic} for the expression of those rates.

\subsection{Integrability}

We establish the Yang-Baxter integrability of the boundary driven multispecies harmonic processes by constructing the corresponding transfer matrix within the framework of the quantum inverse scattering method \cite{Faddeev:1996iy,Sklyanin:1988yz}. The Markov generator of the process is identified with the Hamiltonian of a non-compact open $gl(M+1)$-invariant Heisenberg spin chain with appropriately chosen open boundary condition. The construction of the transfer matrix relies on the rational limit of the stochastic R-matrix \cite{Kuniba:2016fpi,Bosnjak:2016oze} and employs novel solutions to the boundary Yang-Baxter equation, generalising the singlespecies  case \cite{Frassek:2019vjt} with $M=1$.

\paragraph{R-matrix}
The  R-matrix relevant for the construction of the transfer matrix has been known explicitly in terms of its eigenvalues for quite some time, see e.g.~\cite{NJMacKay_1991}. This representation, however, is not convenient when studying stochastic processes as the action on the tensor product of two sites describes the dynamics of the particles. It is well known that not all R-matrices are stochastic and the choice of the basis is important in order to define a stochastic matrix. As discussed above, such form of the R-matrix of non-compact symmetric representations of $gl(M+1)$ was only given recently in components for the $q$-deformed case \cite{Kuniba:2016fpi,Bosnjak:2016oze}. The rational limit yields the R-matrix governing the multispecies harmonic process studied here.  It has been pointed out in \cite{Bosnjak:2016oze} that the R-matrix factorises and that the factorisation in the case of $M=1$ coincides with the algebraic factorisation in terms of Heisenberg pairs (oscillators) studied in Derkachov \cite{Derkachov:2005hw} for the rational case. The case of arbitrary $M$ has been discussed in \cite{Bosnjak} where a mapping from the factorised R-matrix of \cite{Derkachov:2006fw} has been proposed. However, while the stochastic R-matrix factorises into two terms for all $M$, the R-matrix of \cite{Derkachov:2006fw} factorises into the product of $M+1$ terms. The reason for this mismatch is that the latter R-matrix is constructed for arbitrary representations of $gl(M+1)$ using  $\tfrac{1}{2}M(M+1)$ oscillator pairs   while the stochastic R-matrix arises from the symmetric tensor representations that can be realised in terms of only $M$ oscillator pairs. We derive such reduced R-matrix from the Yang-Baxter relation following  \cite{Derkachov:2005hw} that is realised in terms of $M$ pairs of oscillators and factorises into two operators, as in the $gl(2)$ case. To the best of our knowledge this R-matrix has not appeared before in this form in the literature and can be found in Theorem~\ref{Theorem-main}. We further show that the single terms in the factorisation then correspond exactly to the terms in the stochastic R-matrix that describe left and right moving particles. The R-matrix allows us to extract the algebraic form of the stochastic Hamiltonian density.

\paragraph{K-matrix} The construction of K-matrices for a class of non-compact representations, including the symmetric tensor representations relevant here, has been discussed in \cite{Tsuboi:2018gfd} for diagonal boundaries. The boundary reservoirs of the harmonic process however are non-diagonal which allows for insertion and extraction of particles at the boundaries. To our knowledge a K-matrix of this type is not known, we also refer the reader to \cite{2024JPhA...57x5201K} where triangular boundaries are studied for the q-deformed case as well as \cite{2014JPhA...47C5202L}.  We derive the relevant K-matrix from the boundary Yang-Baxter relation following \cite{Frassek:2019vjt}, see also \cite{Antonenko:2024bgw}. Together with the dual K-matrix this allows to construct the transfer matrix. We then show that the logarithmic derivative of the transfer matrix at the permutation point yields the stochastic Hamiltonian of the boundary-driven multispecies harmonic process in Theorem~\eqref{Thm-transfer-Ham}, thus proving its integrability.

\subsection{Dual processes}

The expression
for the Hamiltonian corresponding
to the boundary-driven multispecies harmonic process discussed above
is given in terms of the 
$gl(M+1)$ Lie algebra.
Based on this result,
we then establish three stochastic dualities.
They are obtained by
considering the same 
abstract Hamiltonian in different representations.
In particular, the novel processes
include stochastic energy exchange models with a continuum state space
and an integrable energy redistribution rule.
We here describe informally these dual processes, see Section \ref{subsection-relatedModels} for their
mathematical definition.

\paragraph{Absorbing dual process.}
The main result
is in Theorem \ref{Thm-duality}
which proves the duality of the boundary-driven multispecies harmonic process with a Markov process with identical bulk dynamics and absorbing boundaries following \cite{Spohn_1983,kipnis1982heat}. This follows by
exploiting the $gl(M+1)$
symmetry of the bulk generator
and using a change of representation for the boundary
part. 
This duality implies that
the non-equilibrium steady state
can be studied in terms of
the absorption probabilities
of the dual particles.

\paragraph{Hidden parameter process.}
By considering representations of $gl(M+1)$ on polynomial spaces, we introduce two additional processes.
One of them is 
the hidden parameter process, which can be viewed as a multispecies analogue of the processes defined in~\cite{2025JSP...192...21G}. This process is 
shown to be dual to the multispecies harmonic process with absorbing boundaries, suggesting a probabilistic characterization of the stationary state in terms of mixtures (see \cite{2024JSP...191..150D,2025JSP...192...21G} for the mixture representation in the monospecies case). We note that representations with integral operators of the type studied here have also appeared in the context of QCD \cite{Lipatov:1993yb,Faddeev:1994zg,Braun:1999te}.

\paragraph{Integrable heat conduction model.}

The other process which arises
from representations of $gl(M+1)$ on polynomial spaces,
is the multispecies boundary-driven integrable heat conduction model, which provides a multispecies extension of energy transport models in the spirit of KMP~\cite{kipnis1982heat}.
This process is 
also dual to the multispecies harmonic process with absorbing boundaries.

\subsection{Organization of the paper} Section~\ref{sec:model} introduces the boundary-driven multispecies harmonic process. We further define the dual absorbing process, the boundary-driven integrable heat conduction process and the hidden parameter model, which are related via stochastic duality.
In Section~\ref{sec:facR}, we solve the Yang-Baxter equation associated with $gl(M+1)$ without fixing a priori a representation. Appropriate choices of the representation then yield three $R$-matrices, whose logarithmic derivatives provide the corresponding bulk Hamiltonians of the original and dual models.
Section~\ref{sec:openchain} focuses on the Fock representation, where we solve both the boundary Yang-Baxter equation and its dual, obtaining off-diagonal K-matrices. These allow for the construction of the transfer matrix, whose logarithmic derivative reproduces the Hamiltonian of the boundary-driven multispecies harmonic process.
Finally, Section~\ref{section-duality} establishes several duality relations. Using the $gl(M+1)$ symmetry inherited from the R-matrix, we show duality with a process that has absorbing boundaries. Intertwining relations between Fock and integral representations imply additional dualities.

\paragraph{Acknowledgemets}
RF thank Matthias Staudacher, Jan de Gier, Vladimir Mangazeev, Vladimir Kazakov, Gwena\"{e}l Ferrando and De-Liang Zhong for discussions on factorised R-matrices.
CG thanks Frank Redig for  useful discussions; he also acknowledges financial support from INdAM through INdAM Project 2024 CUP-E53C23001740001.
The work was further  supported by the FAR
UNIMORE project CUP-E93C23002040005, and by the PRIN project “2022ABPBEY” CUP-E53D23002220006. RF was supported by the INFN Bologna (GAST). We thank the mathematical research institute MATRIX in Australia where part of this research was performed. This research was supported in part by the ICTS in India for participating in the program ICTS/DISDECAP2024/10.

\newpage
  
\section{Models}\label{sec:model}

In  Section \ref{sub-harmonic} we introduce the $M$-species harmonic process with boundary reservoirs which is the main process studied in this paper. As shown in Section~\ref{sec:facR} and~\ref{sec:openchain}, its Markov generator can be related directly with the Hamiltonian of the non-compact $gl(M+1)$ integrable XXX Heisenberg spin chain with a certain choice of basis and boundaries.  We further introduce in Section \ref{subsection-relatedModels} three related models:
the multispecies harmonic model with absorbing boundaries, the multispecies integrable heat-conduction model with boundary reservoirs and the multispecies hidden parameter model with conditions. These models are related by duality relations that will be discussed in Section \ref{section-duality}.

\subsection{Boundary-driven multispecies harmonic process}
\label{sub-harmonic}

 \subsubsection{Markov generator}\label{sec:markov}

We start by introducing the notation that is used throughout the paper.
 We consider the geometry of a chain of length $\N\in \mathbb{N}$. The sites of the chain are denoted by $\ell\in \{1,\ldots, \N\}$. To each site $\ell$ we assign an $M$-tuple $m^{\ell }=(m_{1}^{\ell },\ldots,m_{M}^{\ell })$
which encodes a local configuration. Then, the local state space is
\begin{align}\label{single-site-space}
	 \Omega_{\ell}=\left\{m^{\ell }=(m_{1}^{\ell },\ldots,m_{M}^{\ell })\,:\, m_{a}^{\ell }\in \mathbb{N}_{0},\; \forall a=1,\ldots,\M\right\}\,.
    	\end{align}
Here $m_a^\ell$ denotes the number of particle of type $a$ (also refered to as species) that occupy site $\ell$. The number $\M\in \mathbb{N}$ denotes the total number of species   and we use the notation $\mathbb{N}_{0} = \mathbb{N} \cup\{0\}$.
The state space of the model is given by the $N$-fold Cartesian product
\begin{align}\label{cartpro}
	{\Omega}=\bigtimes_{\ell =1}^{\N}\Omega_\ell\,.
\end{align}
We denote the process configurations $\bm{m}\in\Omega$ by the $N$-tuple $\bm{m}=(m^{1},\ldots,m^{\N})$.
To alleviate formulas, we introduce the notation
\begin{align}
	\bm{m}\pm k \bm{\delta}_{\ell }= (m^{1},\ldots,m^{\ell -1},m^\ell \pm k,m^{\ell +1},\ldots,m^{\N})\,,
\end{align}
where  $k\in\mathbb{N}_0^\M$ such that
\begin{align}
	m^{\ell }\pm k=(m_{1}^{\ell }\pm k_{1},\ldots,m_{\M}^{\ell }\pm k_{\M})\,.
\end{align}
Thus, $\bm{m}\pm k \bm{\delta}_{\ell }$ denotes the configuration  obtained from $\bm{m}$ by adding or removing particles at site $\ell$.
The number of species that are moved is specified by $k$.
We use the notation
\begin{equation}
 |k|=k_{1}+\ldots+k_{\M}\,,
\end{equation}
 to indicate the total number of  particles.

 The boundary-driven multispecies harmonic process is a family of processes labeled by a parameter $s>0$.
 The  dynamics of the continuous-time Markov  process is described by its infinitesimal generator of nearest-neighbor type
\begin{align}\label{boundary-driven-generator}
	\mathcal{L}=\mathcal{L}_{\text{left}}+\sum_{\ell =1}^{\N-1}\mathcal{L}_{\ell ,\ell +1}+\mathcal{L}_{\text{right}}.
\end{align}
Here we have separated the bulk part and the boundary part, which we now describe.

	\paragraph{The bulk dynamics.} We consider functions $f:\Omega\to \mathbb{R}$ and a generic bond of two neighboring sites $(\ell ,\ell +1)$. The bulk generator $\mathcal{L}_{\ell ,\ell +1}$ modifies the occupation numbers on the two neighboring sites $(\ell ,\ell +1)$. Furthermore it can be decomposed into a part  associated to right jumps $\mathcal{L}^{-}_{\ell ,\ell +1}$ and into a part associated to left jumps $\mathcal{L}^{+}_{\ell ,\ell +1}$. It acts as
	\begin{equation}\label{bulk-generatorDensity}
		\eL_{\ell ,\ell +1}f(\bm{m})=\eL^{-}_{\ell ,\ell +1}f(\bm{m})+\eL^{+}_{\ell ,\ell +1}f(\bm{m})
	\end{equation}
	where
	\begin{equation}
		\eL^{-}_{\ell ,\ell +1}f(\bm{m})=\sum_{k_{1}=0}^{m_{1}^{\ell }}\ldots\sum_{k_{\M}=0}^{m_{\M}^{\ell }}\varphi_{s}(k,m^{\ell })\left[f(\bm{m}-k\bm{\delta}_{\ell }+k\bm{\delta}_{\ell +1})-f(\bm{m})\right]\,,
	\end{equation}
	and   
		\begin{equation}
		\eL^{+}_{\ell ,\ell +1}f(\bm{m})=\sum_{k_{1}=0}^{m_{1}^{\ell +1}}\ldots\sum_{k_{\M}=0}^{m_{\M}^{\ell +1}}\varphi_{s}(k,m^{\ell +1})\left[f(\bm{m}-k\bm{\delta}_{\ell +1}+k\bm{\delta}_{\ell })-f(\bm{m})\right]\,.
	\end{equation}
	Here, we have introduced the jump rates
\begin{align}\label{rates}
	\varphi_{s}({k},{m})=\frac{\Gamma(|k|)\Gamma(2s+|m|-|k|)}{\Gamma(2s+|m|)}\prod_{a=1}^{\M}\frac{\Gamma(m_{a}+1)}{\Gamma(k_{a}+1)\Gamma(m_{a}-k_{a}+1)}\mathbbm{1}_{\{|k|>0\}}\,,
\end{align}
that depend on the particle numbers of the departing site, the number of particles that move, and the parameter $s$. We further note that the jump rates are symmetric, so that particles move at the same rates to the left and right of the chain. The rates generalize those of the mono-species case, as introduced in \cite{Frassek:2019vjt}.

\begin{remark}\label{Remark-MqHahn}[Multispecies q-Hahn process]
    Considering the right and left transition rates written in \cite[Eq. (51),(59)]{Kuniba:2016fpi}, one can verify that by choosing $\mu=q^{2s}$, multiplying by $(1-q)$ and taking the rational limit $q\to 1$, the rates of the multispecies q-Hahn process  converge to the transition rate~\eqref{rates} of the multispecies harmonic process.
\end{remark}
\paragraph{The boundary dynamics.}
Next, we define the boundary dynamics, which can be decomposed into two parts: the out-part, which removes particles from the chain, using the same rate as the bulk generator; the in-part, that instead injects particles into the chain, with a rate that depends on the boundary parameters. More precisely,
the left boundary generator changes only occupation numbers at site $1$. It acts
on functions $f:\Omega \to \mathbb{R}$
as
\begin{align}\label{left-boundary-generator}
	\eL_{\text{left}}f(\bm{m})=\eL_{\text{left}}^{\text{out}}f(\bm{m})+\eL_{\text{left}}^{\text{in}}f(\bm{m})\,,
\end{align}
with
\begin{align}\label{boundary-out-generator}
	\eL_{\text{left}}^{\text{out}}f(\bm{m})=\sum_{k_1=0}^{m_{1}^{1}}\ldots\sum_{k_{\M}=0}^{m_{\M}^{1}}\varphi_{s}(k,m^{1})\left[f(\bm{m}-\bm{\delta}_{1}k)-f(\bm{m})\right]\,,
\end{align}
and 
\begin{align}\label{boundary-in-generator}
\eL_{\text{left}}^{\text{in}}f(\bm{m})=\sum_{k_{1}=0}^{\infty}\ldots\sum_{k_{\M}=0}^{\infty}\mathbbm{1}_{\{|k|>0\}}\Gamma(|k|)\left(\prod_{a=1}^{\M}\frac{\beta_{l,a}^{k_{a}}}{\Gamma(k_{a}+1)}\right)\,\left[f(\bm{m}+\bm{\delta}_{1}k)-f(\bm{m})\right]\,.
\end{align}
Here, $(\beta_{l,a})_{a\in \{1,\ldots,\M\}}$ are the boundary parameters of the left reservoir, taking values in the interval $(0,1)$ and satisfying $\sum_{a=1}^{\M}\beta_{l,a}<1$. 

Similarly, the right boundary generator acts as
\begin{align}\label{right-boundary-generator}
	\eL_{\text{right}}f(\bm{m})=\eL_{\text{right}}^{\text{out}}f(\bm{m})+\eL_{\text{right}}^{\text{in}}f(\bm{m})\,,
\end{align}
where
\begin{align}
\eL_{\text{right}}^{\text{out}}f(\bm{m})=\sum_{k_{1=0}}^{m_{1}^{\N}}\ldots\sum_{k_{\M}=0}^{m_{\M}^{\N}}\varphi_{s}(k,m^{\N})\left[f(\bm{m}-\bm{\delta}_{\N}k)-f(\bm{m})\right]\,,
\end{align}
and
\begin{align}\label{boundary-right-in}
	\eL_{\text{right}}^{\text{in}}f(\bm{m})=\sum_{k_{1}=0}^{\infty}\ldots\sum_{k_{\M}=0}^{\infty}\mathbbm{1}_{\{|k|>0\}}\Gamma(|k|)\left(\prod_{a=1}^{\M}\frac{\beta_{r,a}^{k_{a}}}{\prod_{a=1}^{\M}\Gamma(k_{a}+1)}\right)\left[f(\bm{m}+\bm{\delta}_{\N}k)-f(\bm{m})\right]\,.
\end{align}
Here, $(\beta_{r,a})_{a\in \{1,\ldots,\M\}}$ are the boundary parameters of the right reservoir,  taking values in the interval $(0,1)$ and satisfying $\sum_{a=}^{\M}\beta_{r,a}<1$.

\begin{remark}
    If the left and right boundary parameters are equal for all species, namely $\beta_{l,a}=\beta_{r,a}$ for all $a\in \{1,\ldots, \M\}$, then the process has a reversible measure given by a product of Negative Multinomials, see Lemma \ref{lemma-reversibleMeasure-M-harmonic-boundaries}. 
    In the following we denote by 
\begin{equation}\label{rhos}
 \rho_{r,a}=\frac{\beta_{r,a}}{1-\sum_{a=1}^{\M}\beta_{r,a}}\,,\qquad 
 \rho_{l,a}=\frac{\beta_{l,a}}{1-\sum_{a=1}^{\M}\beta_{l,a}}\,.
\end{equation}
\end{remark}

One of the main results of this paper will be the relation that
is established between
the Markov generator \eqref{boundary-driven-generator} and the
Hamiltonian of the integrable open spin chain with $gl(M+1)$
symmetry. Remarkably, the Hamiltonian becomes a stochastic matrix in a well-chosen basis,
as we discuss in the next paragraph.

\subsubsection{The stochastic Hamiltonian.}
In order to relate the $M$-species harmonic process to the stochastic Hamiltonian of the non-compact  $gl(M+1)$ spin chain we identify the Hilbert space of the latter with the configuration space of the process.

For $\ell\in \{1,\ldots,N\}$, we introduce the  vector space $V_\ell$ over $\mathbb{C}$ whose basis is labeled by the configuration space $\Omega_\ell$ in~\eqref{single-site-space} such that
\begin{equation}\label{basis}
    V_\ell =\spa\left(\left\{|m_1^\ell,\ldots, m_M^\ell\rangle :m^\ell\in\Omega_\ell\right\}\right)\,,
\end{equation}
where $|{m_1,\ldots, m_M}\rangle$ denotes the tensor product of $M$ elementary infinite-dimensional unit vectors $e_{m_1}\otimes\ldots\otimes e_{m_M}$ with $(e_m)_{n}=\delta_{mn}$ and $n,m\in \mathbb{N}_0$. The so-called quantum or physical space of the spin chain then takes the form
\begin{equation}
 V_{\N}=\bigotimes_{\ell=1}^N V_{\ell}\,,
\end{equation}
cf.~\eqref{cartpro}. Its basis vectors will be denoted by $|\bm{m}\rangle=|m^1,\ldots,m^N\rangle$. 

We introduce the Hamiltonian of the open non-compact $gl(M+1)$ XXX spin chain as
\begin{equation}\label{Hamiltonian-MHarmonic}
 \eH=\eH_{\text{left}}+\sum_{\ell=1}^{N-1}\eH_{\ell,\ell+1}+\eH_{\text{right}}\,.
\end{equation}
The bulk Hamiltonian decomposes as
\begin{equation}\label{eq:dens}
    \eH_{\ell,\ell+1}=\eH_{\ell,\ell+1}^{-}+\eH_{\ell,\ell+1}^{+}
\end{equation}
and,  for all $\ell\in \{1,\ldots,\N-1\}$, they act on the two-site tensor product space as
\begin{equation}\label{hbulk}
\begin{split}
&\eH_{\ell,\ell+1}^{-}|m^{\ell}\rangle \otimes |m^{\ell+1}\rangle=h_s(|m^{\ell}|)|m^{\ell}\rangle \otimes |m^{\ell+1}\rangle-\sum_{k_{1}=0}^{m_{1}^{\ell}}\cdots \sum_{k_{N}=0}^{m_{N}^{\ell}}\varphi_{s}(k,m^{\ell}) |m^{\ell}-k\rangle \otimes |m^{\ell+1}+k\rangle\,,\\
&\eH_{\ell,\ell+1}^{+}|m^{\ell}\rangle \otimes |m^{\ell+1}\rangle=h_s(|m^{\ell+1}|)|m^{\ell}\rangle \otimes |m^{\ell+1}\rangle-\sum_{k_{1}=0}^{m_{1}^{\ell+1}}\cdots \sum_{k_{N}=0}^{m_{N}^{\ell+1}}\varphi_{s}(k,m^{\ell+1}) |m^{\ell}+k\rangle \otimes |m^{\ell+1}-k\rangle\,,
 \end{split}
\end{equation}
with $|m\pm k\rangle =|m_1\pm k_1,\ldots,m_M\pm k_M\rangle$ and
  \begin{align}
 	h_{s}(|m|)=\sum_{k=1}^{|m|}\frac{1}{2s+k-1}. \end{align}
 At the boundaries we have \begin{equation}\label{hleft}
\begin{split}
 \eH_{\text{left}}|m\rangle  =(h_s(|m|)-\log(1-|\beta_l|))|m\rangle &-\sum_{k_{1}=0}^{m_{1}}\cdots \sum_{k_{N}=0}^{m_{N}}\varphi_{s}(k,m) |m-k\rangle\\
  &-\sum_{k_{1},\ldots, k_{\M}=0\,:\,|k|\neq 0}^{\infty}\Gamma(|k|)\left(\prod_{a=1}^{\M}\frac{\beta_{l,a}^{k_{a}}}{\Gamma(k_{a}+1)}\right) |m+k\rangle
 \end{split}
\end{equation}
and
\begin{equation}\label{hright}
\begin{split}
 \eH_{\text{right}}|m\rangle =(h_s(|m|)-\log(1-|\beta_r|))|m\rangle &-\sum_{k_{1}=0}^{m_{1}}\cdots \sum_{k_{N}=0}^{m_{N}}\varphi_{s}(k,m) |m-k\rangle\\
  &-\sum_{k_{1},\ldots, k_{\M}=0\,:\,|k|\neq 0}^{\infty}\Gamma(|k|)\left(\prod_{a=1}^{\M}\frac{\beta_{r,a}^{k_{a}}}{\Gamma(k_{a}+1)}\right) |m+k\rangle
 \end{split}
\end{equation}
\begin{remark}[Diagonal terms]
 The diagonal terms in the Hamiltonian density~\eqref{hbulk} are obtained when summing over all possible transition rates.
More precisely, we have that
 \begin{align}\label{digamma-ratesBulk}
\begin{split}
 	\sum_{k_{1}=0}^{m_{1}}\cdots\sum_{k_{\M}=0}^{m_{\M}}\varphi_{s}(k,m)
 	&=\psi(|m|+2s)-\psi(2s)
    = h_{s}(|m|)\,.
 	\end{split}
 \end{align}
This relation can be obtained by writing the rates in terms of the Beta function
 \begin{align}\label{rates2}
	\varphi_{s}(k,m)=B(2s+|m|-|k|,|k|)\prod_{a=1}^{\M}\binom{m_a}{k_a}\,,
\end{align}
and subsequently using the integral representation of the Beta and Digamma function
\begin{equation}\label{eq:intrp}
 B(x,y)=\int_0^1 t^{x-1}(1-t)^{y-1}dt\,,\qquad  \psi(x+2s)-\psi(2s)=\int_0^1 \frac{t^{2s-1}}{1-t}(1-t^x) dt\,.
\end{equation}

The diagonal terms in the boundary terms~\eqref{hleft} and~\eqref{hright} arise similarly noting that  for all $0\leq \beta_{a}<1$ such that $|\beta|=\sum_{a=1}^{\M}\beta_{a} <1$ we have that
\begin{align}\label{log-rateBoundary}
	\sum_{k_{1},\ldots, k_{\M}=0\,:\,|k|\neq 0}^{\infty}\Gamma(|k|)\left(\prod_{a=1}^{\M}\frac{\beta_{a}^{k_{a}}}{\Gamma(k_{a}+1)}\right)=-\log\left(1-|\beta|\right)\,.
\end{align}
\end{remark}
As we will prove in Section~\ref{sec:openchain}, the Hamiltonian $\eH$ can be identified with the integrable open higher-rank XXX spin chain with non-compact highest-weight representations of $gl(M+1)$ labeled by the Dynkin labels $\mmu=(\mmu_1,\mmu_2,\ldots,\mmu_{2})$ with $\mmu_1<\mmu_2$. The parameter $s$ in the rates defined above in Section~\ref{sec:markov}  is related to the Dynkin labels via $2s=\mmu_2-\mmu_1$ such that $s>0$.

If we further define the dual basis $\langle \bm{n}|=\langle n^1,\ldots,n^N|$ with the orthogonality condition
\begin{equation}
\langle n^1,\ldots,n^N|m^1,\ldots,m^N\rangle=\prod_{\ell=1}^N\prod_{a=1}^M\delta_{n^\ell_a,m^\ell_a}\,,
\end{equation}
we can spell out the relation of the Markov generator $\mathcal{L}$ and the stochastic Hamiltonian $\eH$ explicitly.
Indeed, by a direct computation, one shows that
\begin{equation}\label{gen-ham1}
    \mathcal{L}f(\bm{m})=-\langle f|\eH^t|\bm{m}\rangle \qquad\forall\bm{m}\in \Omega\,,
\end{equation}
where
\begin{equation}\label{gen-ham2}
    \langle f|=\sum_{\bm{n}\in \Omega} f(\bm{n})\langle \bm{n}|\,.
\end{equation}

\subsection{Dual models}\label{subsection-relatedModels}
We introduce three Markov processes by means of their Markov generator: the absorbing dual harmonic model, the hidden parameter model and the boundary-driven multispecies integrable heat conduction model. These processes
will be discussed in Section \ref{section-duality}.
The first is useful to 
compute the $n$-th moment
of the stationary state 
using $n$ dual particles.
The other two arise by considering a representation of the $gl(M+1)$ algebra in terms of differential operators.
The hidden process describes
the evolution of the parameters of a mixture of
products Negative Multinomials (see \cite{2025JSP...192...21G} for the mono-species case).
The integrable heat conduction is a stochastic exchange model where energy
is redistributed among particles (see \cite{Franceschini:2022vmr} for the mono-species case).

\subsubsection{The multispecies absorbing dual harmonic model} We denote the multispecies harmonic process with absorbing boundaries (or absorbing dual process) by $(\bm{\xi}(t))_{t\geq 0}$ and we define it on the enlarged chain with sites denoted by $\ell\in\{0,1,\ldots,\N,\N+1\}$, namely the chain with two extra sites $0$ and $\N+1$ attached to the end sites $1$ and $\N$, respectively. The state space for the dual process is given by
\begin{equation}
	\widetilde{\Omega}=\bigtimes_{\ell =0}^{\N+1}\Omega_{\ell }\,.
\end{equation}
with $\Omega_{\ell}$ defined in~\eqref{single-site-space}. We write a dual configuration $\bm{\xi}\in \widetilde{\Omega}$ as 
\begin{equation}
	\bm{\xi}=(\xi^0,\xi^{1},\ldots,\xi^{\N},\xi^{\N+1})\,,
\end{equation}
where $\xi^{\ell}=(\xi_1^{\ell},\ldots,\xi_{\M}^{\ell})$ and the component $\xi_{a}^{\ell}$ denotes the number of dual particles of species $a$ at site $\ell$. The generator of $(\bm{\xi}(t))_{t\geq 0}$ reads 
\begin{equation}\label{abs-Harmonic-Generator}
		\widetilde{\mathcal{L}}=\widetilde{\mathcal{L}}_{\text{left}}+\sum_{\ell=1}^{N-1}\mathcal{L}_{\ell,\ell+1}+\widetilde{\mathcal{L}}_{\text{right}}
	\end{equation}
    where $\mathcal{L}_{\ell,\ell+1}$ is  defined in~\eqref{bulk-generatorDensity} and where the boundary generators act on functions $f:\widetilde{\Omega}\to \mathbb{R}$ as
	\begin{equation}
		\widetilde{\mathcal{L}}_{\text{left}}f(\bm{\xi})=\sum_{k_{1}=0}^{\xi^{1}_{1}}\cdots \sum_{k_{\M}=0}^{\xi^{1}_{\M}}\varphi_{s}(k,\xi^{1})\left[f(\bm{\xi}-\bm{\delta}_{1}k+\bm{\delta}_{0}k)-f(\bm{\xi})\right]
	\end{equation}
    and
    \begin{equation}
		\widetilde{\mathcal{L}}_{\text{right}}f(\bm{\xi})=\sum_{k_{1}=0}^{\xi^{\N}_{1}}\cdots \sum_{k_{\M}=0}^{\xi^{\N1}_{\M}}\varphi_{s}(k,\xi^{\N})\left[f(\bm{\xi}-\bm{\delta}_{\N}k+\bm{\delta}_{\N+1}k)-f(\bm{\xi})\right]\,.
	\end{equation} 
\subsubsection{Multispecies hidden parameter model}\label{sec:hidden}
The process is defined on a chain of length $\N\in\mathbb{N}$. A configuration is denoted by $(\hidden^{\ell})_{\ell\in\{1,\ldots,\N\}}$, where for each site $\ell$ we set $\hidden^{\ell}=(\hidden_{1}^{\ell},\ldots,\hidden_{\M}^{\ell})$ with $\hidden_{a}^{\ell}\in \mathbb{R}_{+}$. 
The infinitesimal generator reads 
\begin{equation}\label{generator-Hidden} 
	\ecL=\ecL_{\text{left}}+\sum_{\ell=1}^{\N-1}\ecL_{\ell,\ell+1}+\ecL_{\text{right}}\,,
\end{equation}
where the bulk generator can be decomposed on each bond as
\begin{equation}\label{generator-bulk-hidden}
	\ecL_{\ell,\ell+1}=\ecL^-_{\ell,\ell+1}+\ecL^+_{\ell,\ell+1}\,.
\end{equation}
The linear operators above act on functions $f:\mathbb{R}_{+}^{\M}\times \mathbb{R}_{+}^{\M}\to \mathbb{R}$ as 
\begin{equation}\label{generatorMinus-heat}
	\begin{split}
		\ecL^{-}_{\ell,\ell+1}f(\hidden^{\ell},\hidden^{\ell+1})&= \int_0^1d\alpha\,\frac{\alpha ^{2s-1}}{1-\alpha}\left[f(\alpha \hidden^{\ell}+(1-\alpha)\hidden^{\ell+1},\hidden^{\ell+1})-f(\hidden^{\ell},\hidden^{\ell+1})\right]\,.
	\end{split}
\end{equation}
and
\begin{equation}\label{generatorPlus-heat}
	\begin{split}
		\ecL^{+}_{\ell,\ell+1} f(\hidden^{\ell},\hidden^{\ell+1})&= \int_0^1d\alpha\,\frac{\alpha ^{2s-1}}{1-\alpha }\left[f(\hidden^{\ell},
        (1-\alpha)\hidden^{\ell} +
        \alpha \hidden^{\ell+1})-f(\hidden^{\ell},\hidden^{\ell+1})\right]\,,
	\end{split}
\end{equation}
The left boundary generator depends on the $\M$ dimensional vector of left boundary parameters $\rho_{l}=(\rho_{l,1},\ldots,\rho_{l,\M})$ and acts on functions $f:\mathbb{R}_{+}^{\M}\to \mathbb{R}$ as
\begin{equation}\label{redR2sl2int4w}
	\begin{split}
		\ecL_{\text{left}} f(\hidden^{1})&= \int_0^1d\alpha\,\frac{\alpha ^{2s-1}}{1-\alpha }\left[f(\alpha \hidden^{1}+(1-\alpha)\rho_{l})-f(\hidden^{1})\right]\,.
	\end{split}
\end{equation}
The right boundary generator depends on the $\M$ dimensional vector of right boundary parameters $\rho_{r}=(\rho_{r,1},\ldots,\rho_{r,\M})$ and acts on functions $f:\mathbb{R}_{+}^{\M}\to \mathbb{R}$ as  
\begin{equation}	\begin{split}
		\ecL_{\text{right}} f(\hidden^{\N})&= \int_0^1d\alpha\,\frac{\alpha ^{2s-1}}{1-\alpha }\left[f(\alpha \hidden^{\N}+(1-\alpha)\rho_{l})-f(\hidden^{\N})\right]\,.
	\end{split}
\end{equation} 
This model may be seen as the multispecies analogue of the   model introduced in \cite{Derkachov:1999pz,2025JSP...192...21G}.
\begin{remark}
    For the remainder of this paper, we assume that the infinitesimal generator~\eqref{generator-Hidden} acts on the space of polynomial functions. 
The characterization of the domain of this generator goes beyond the scope of this paper.
\end{remark}
\subsubsection{Boundary-driven multispecies integrable heat conduction model}\label{sec:heat}
The process is defined on a chain of length $\N\in\mathbb{N}$. A configuration is denoted by $(\heat^{\ell})_{\ell\in \{1,\ldots,\N\}}$, where for each site $\ell$ we set  $\heat^{\ell}=(\heat_{1}^{\ell},\ldots,\heat_{\M}^{\ell})$, with $\heat_{a}^{\ell}\in \mathbb{R}_{+}$. The infinitesimal generator reads 
\begin{equation}\label{heatGenerato}
	\ecdL=\ecdL_{\text{left}}+\sum_{\ell=1}^{\N-1}\ecdL_{\ell,\ell+1}+\ecdL_{\text{right}}\,.
\end{equation}
Here the bulk generator can be decomposed on each bond as
\begin{equation}\label{bulk-Heat}
	\ecdL_{\ell,\ell+1}=\ecdL^{+}_{\ell,\ell+1}+\ecdL^{-}_{\ell,\ell+1}\,.
\end{equation}
The linear operators above act on functions $f:\mathbb{R}_{+}^{\M}\times \mathbb{R}_{+}^{\M}\to \mathbb{R}$ as 
\begin{equation}\label{redR2sl2int4b}
	\begin{split}
		\ecdL^{-}_{\ell,\ell+1}f(\heat^{\ell},\heat^{\ell+1})&= \int_0^1d\alpha\,\frac{\alpha ^{2s-1}}{1-\alpha}\left[f(\alpha \heat^{\ell},\heat^{\ell+1}+(1-\alpha)\heat^{\ell})-f(\heat^{\ell},\heat^{\ell+1})\right]
	\end{split}
\end{equation}
and
\begin{equation}\label{redR2sl2int24b}
	\begin{split}
		\ecdL^{+}_{\ell,\ell+1}f(\heat^{\ell},\heat^{\ell+1})&= \int_0^1d\alpha\,\frac{\alpha ^{2s-1}}{1-\alpha}\left[f(\heat^{\ell}+(1-\alpha)\heat^{\ell+1},\alpha \heat^{\ell+1})-f(\heat^{\ell},\heat^{\ell+1})\right]\,.
	\end{split}
\end{equation}
The left boundary generator depends on  the left boundary parameters $\rho_{l}=(\rho_{l,1},\ldots,\rho_{l,\M})$, with $\rho_{l,a}>0$, and can be decomposed into
\begin{equation}
	\ecdL_{\text{left}}=\ecdL_{\text{left}}^{\text{in}}+\ecdL_{\text{left}}^{\text{out}}
\end{equation}
where the action of the two linear operators above  on functions $f:\mathbb{R}_{+}^{\M}\to \mathbb{R}$ reads
\begin{equation}
	\ecdL_{\text{left}}^{\text{out}}f(\heat^{1})=\int_0^1 \frac{\alpha^{2s-1}}{1-\alpha}\left[f(\alpha \heat^{1})-f(\heat^{1}) \right]d\alpha
\end{equation}
and
\begin{equation}
	\ecdL_{\text{left}}^{\text{in}}f(\heat^{1})=\int_{0}^{\infty}\frac{e^{-\alpha}}{\alpha}\left[f(\heat^{1}+\alpha \rho_{l})- f(\heat^{1})\right]d \alpha\,.
\end{equation}
Similarly, the right boundary generator depend on the right boundary parameters $\rho_{r}=(\rho_{r,1},\ldots,\rho_{r,\M})$, with $\rho_{r,a}>0$, and can be decomposed into
\begin{equation}
	\ecdL_{\text{right}}=\ecdL_{\text{right}}^{\text{in}}+\ecdL_{\text{right}}^{\text{out}}
\end{equation}
where the action of the two linear operators above  on functions $f:\mathbb{R}_{+}^{\M}\to \mathbb{R}$ reads
\begin{equation}
	\ecdL_{\text{right}}^{\text{out}}f(\heat^{\N})=\int_0^1 \frac{\alpha^{2s-1}}{1-\alpha}\left[f(\alpha \heat^{\N})-f(\heat^{\N}) \right]d\alpha
\end{equation}
and
\begin{equation}
	\ecdL_{\text{right}}^{\text{in}}f(\heat^{\N})=\int_{0}^\infty\frac{e^{-\alpha}}{\alpha}\left[f(\heat^{\N}+\alpha \rho_{r})- f(\heat^{\N})\right]d \alpha\,.
\end{equation}
This model may be seen as the multispecies analogue of  model introduced in \cite{Frassek:2019vjt,Franceschini:2022vmr}.

\section{The factorised R-matrix}\label{sec:facR}
In this section we derive the factorised R-matrix for a certain class of representations of ${gl(\M+1)}$ following the construction of Derkachov \cite{Derkachov:2005hw}. The R-matrices obtained in this section are new but can be seen as a non-trivial degeneration of the generic construction in \cite{Derkachov:2008aq}. In contrast to loc.~cit., the R-matrix presented in the following has $M$ conjugate Heisenberg pairs and factorises into two parts only. We further show that the matrix elements of this R-matrix coincide with the stochasic R-matrix obtained in \cite{Kuniba:2016fpi,Bosnjak:2016oze} in the rational limit. We  derive the Hamiltonian density and observe that the two factors in the R-matrix yield the left and right hopping terms of the stochastic Hamiltonian density, cf.~\eqref{eq:dens}. In Appendix~\ref{section-integral-reps},  we also present two integral representation of the R-matrix following \cite{Derkachov:2005hw} that are immediately related to the dual processes of Section~\ref{subsection-relatedModels}.

	\subsection{Lax matrices}
    
    We begin with the R-matrix of $gl(\M+1)$ in fundamental representation that is  explicitly given by
	\begin{equation}\label{fundamental-R}
		\fundR(x)=x+\sum_{A,B=0}^{\M} e_{AB}\otimes e_{BA}
	\end{equation}
	with $(e_{AB})_{CD}=\delta_{AC}\delta_{DB}$ and the spectral parameter $x\in\mathbb{C}$ such that  $\fundR(x):\mathbb{C}^{\M+1}\otimes \mathbb{C}^{\M+1}\to \mathbb{C}^{\M+1}\otimes \mathbb{C}^{\M+1}$. To alleviate notation we suppressed the $(\M+1)^2\times(\M+1)^2$ identity matrix multiplying the spectral parameter. It is well known that this R-matrix is a solution of the Yang-Baxter equation
	\begin{equation}
	 \fundR_{12}(x-y)\fundR_{13}(x-z)\fundR_{23}(y-z)=\fundR_{23}(y-z)\fundR_{13}(x-z)\fundR_{12}(x-y)
	\end{equation}
	acting on the tensor product of three spaces $V_1\otimes V_2\otimes V_3$ where $V_i=\mathbb{C}^{\M+1}$. The subscripts indicate on which spaces the R-matrices act non-trivially, e.g.
\begin{equation}
 \fundR_{12}(x)=\fundR(x)\otimes \ID\,,
\end{equation}
see \cite{Faddeev:1996iy} for an introduction.

Furthermore, we introduce the Lax matrix
	\begin{equation}\label{eq:lax}
		\mathrm{L}(x)=x+\sum_{A,B=0}^{\M} e_{AB} \mathrm{J}_{BA}
	\end{equation}
	which is an $(\M+1)\times(\M+1)$ matrix whose entries contain the generators $\mathrm{J}_{AB}$ of the  $gl(\M+1)$ Lie algebra. They obey the commutation relations
	\begin{equation}\label{gl-commutators}
[\mathrm{J}_{AB},\mathrm{J}_{CD}]=\delta_{BC}\mathrm{J}_{AD}-\delta_{AD}\mathrm{J}_{CB}\,.
	\end{equation}
    These commutation relations guarantee that the RLL equation
	\begin{equation}\label{RLL1}
		\fundR(x-y)(\mathrm{L}(x)\otimes\ID)(\ID\otimes \mathrm{L}(y))=(\ID\otimes \mathrm{L}(y))(\mathrm{L}(x)\otimes\ID)\fundR(x-y)
	\end{equation}
is satisfied. The Lax matrix~\eqref{eq:lax} is also referred to as the evaluation map of the Yangian $Y(gl(M+1))$, which is generated by the RLL equation~\eqref{RLL1}, into the Lie algebra $ev:Y(gl(M+1))\to gl(M+1)$, see e.g.~\cite{molev2003yangians}.

We are interested in a certain class realisations of $gl(\M+1)$ that are obtained using $\M$ copies of the Heisenberg (oscillator) algebra satisfying the commutation relation
  \begin{equation}\label{eq:osc}
   [\oa_{ a},\oad_{ b}]=\delta_{a,b}\qquad \text{with}\qquad a, b=1,\ldots,\M\,.
  \end{equation}
  Denoting the universal enveloping  algebra of the Heisenberg algebra by $\mathcal{A}_M$ we can define an algebra homomorphism $\phi:gl(M+1)\to \mathcal{A}_M$. For this purpose it is convenient to introduce the vectors
\begin{equation}
		\OAD=(\oad_{1},\ldots,\oad_{\M})\,,\qquad \OA=(\oa_1,\ldots,\oa_{\M})^t,
	\end{equation}
	such that
	\begin{equation}\label{GL-oscRealization}
		J_{BA}=\phi(\mathrm{J}_{BA})=\left(\begin{array}{cc}
			\mmu_1-\OAD\OA&\OAD\left((\mmu_1-\mmu_2+1)\ID-\OA\OAD\right)\\
			\OA&(\mmu_2-1)\ID+\OA\OAD
		\end{array}\right) _{AB}
	\end{equation}
	where $\ID$ denotes the $\M\times\M$ identity matrix and where $\mu_{1},\mu_{2}\in \mathbb{R}$. Depending on convenience, we may also use the notation
    \begin{equation}\label{CandS}
        c:=\frac{\mu_{1}+\mu_{2}}{2}\quad \text{and}\quad s:=\frac{\mu_{2}-\mu_{1}}{2}\,.
    \end{equation}
The corresponding Lax matrix
\begin{equation}\label{lax2}
    L(x)=x+\sum_{A,B=0}^{\M} e_{AB}J_{BA}
    \end{equation} 
is obtained when inserting the oscillator realisation~\eqref{GL-oscRealization} in~\eqref{lax2}. It can be written in the factorised form
	\begin{equation}\label{eq:laxfac}
		L(x)=\left(\begin{array}{cc}
			1&-\OAD\\
			0&\ID
		\end{array}\right)\left(\begin{array}{cc}
			x+\mmu_1&0\\
			\OA&(x+\mmu_2-1)\ID
		\end{array}\right)\left(\begin{array}{cc}
			1&\OAD\\
			0&\ID
		\end{array}\right)\,,
	\end{equation}
	see also \cite{Derkachov:2006fw}.
\begin{remark}
A q-deformation of this Lax matrix appeared in the context of Macdonald polynomials in
\cite{2017CMaPh.352..773G} which are relevant for the study of Macdonald processes \cite{2011arXiv1111.4408B}  and   the stochastic vertex model \cite{2016arXiv160105770B}.
\end{remark}

\begin{remark}
The algebra $gl(M+1)$ has a non-trivial center $Z(gl(M+1))$. The generators of $sl(M+1)\simeq gl(M+1)/Z(gl(M+1))$
can be defined in the standard way \cite{Humphreys1972} as
\begin{equation}
 \tilde J_{AB}=J_{AB}-\frac{\delta_{AB}}{M+1}\sum_{C=0}^{M}J_{CC}\,.
\end{equation}
We see that the term proportional to the central element $\sum_{C}J_{CC}$  can be absorbed by a shift in the spectral parameter when using the parametrisation~\eqref{CandS} and can the Lax matrix above can equally well be expressed in terms of the $sl(M+1)$ generators.
\end{remark}

	\subsection{Derkachov's construction for $gl(M+1)$}

		In this section we will derive the R-operator $\abrOSCR(x)=\sum_k \abrOSCR_k x^k $ with $\abrOSCR_k\in \mathcal{A}_M\otimes \mathcal{A}_M$ that solves the Yang-Baxter  equation
		\begin{equation}\label{RLL}
			\abrOSCR_{12}(x-y) L_1(x)L_2(y)= L_2(y) L_1(x)\abrOSCR_{12}(x-y)\,.
		\end{equation}
        Here the subscripts $1$ and $2$ denote two copies of the algebra $\mathcal{A}_M$ whose elements we denote by $(\OA^{[i]},\OAD^{[i]})$ with $i=1,2$,  respectively.
		        Defining $x_i=x+\mmu_i^{[1]}$ and $y_i=y+\mmu_i^{[2]}$ we write the Yang-Baxter equation~\eqref{RLL} as
		\begin{equation}
			\abrOSCR_{12}(x-y)L_{1}(x_{1},x_{2})L_{2}(y_{1},y_{2})=L_{2}(y_{1},y_{2})L_{1}(x_{1},x_{2})\abrOSCR_{12}(x-y)
		\end{equation}
where we now made explicit the dependence on the parameters $x_i$ and $y_i$. More precisely we have
		\begin{equation}\label{L1-u1u2}
			L_{1}(x_{1},x_{2})= \left(\begin{array}{cc}
			x_{1}-\OAD^{[1]}\OA^{[1]}&	\OAD^{[1]}\left((x_{1}-x_{2}+1)\ID-\OA^{[1]}\OAD^{[1]}\right)\\
				\OA^{[1]}&(x_{2}-1)\ID+\OA^{[1]}\OAD^{[1]}
			\end{array}\right) \,,
		\end{equation}
				and
		\begin{equation}\label{L2-u1u2}
			L_{2}(y_{1},y_{2})= \left(\begin{array}{cc}
			y_{1}-\OAD^{[2]}\OA^{[2]}&	\OAD^{[2]}\left((y_{1}-y_{2}+1)\ID-\OA^{[2]}\OAD^{[2]}\right)\\
				\OA^{[2]}&(y_{2}-1)\ID+\OA^{[2]}\OAD^{[2]}
			\end{array}\right) \,,
		\end{equation}
			It is further convenient to introduce the permuted R-operator
    \begin{equation}
\check{\abrOSCR}=P\abrOSCR,\end{equation} where $P$ is the permutation operator. The Yang-Baxter relation~\eqref{RLL} then turns into the intertwining relation
		\begin{equation}
			\check{\abrOSCR}_{12}(x-y)L_{1}(x_{1},x_{2})L_{2}(y_{1},y_{2})=L_{1}(y_{1},y_{2})L_{2}(x_{1},x_{2})\check{\abrOSCR}_{12}(x-y)\,.
		\end{equation}
																																							 		 In the following theorem, which is our first main result, we determine the R-operator $\abrOSCR(x)=P\check{\abrOSCR}(x)$.
		\begin{theorem}[Factorized form of the R-matrix]\label{Theorem-main}
	The R-operator $\abrOSCR(x-y)$ that solves the Yang-Baxter equation~\eqref{RLL} factorizes as
		 \begin{equation}\label{Facrtorized-R}
		 	{\abrOSCR}(x-y) =P\,\check {\abrOSCR}(x-y)=P \,{\abrOSCR}^+(x_2|x_1,y_2){\abrOSCR}^-(x_1,x_2|y_1)=P\, {\abrOSCR}^-(x_1,y_2|y_1){\abrOSCR}^+(x_2|y_1,y_2)\,.
		 \end{equation}
The single factors of the R-operator solve the equations
		\begin{equation}\label{eq-Rp}
			{\abrOSCR}^+_{12}(x_2|y_1,y_2) L_1(x_1,x_2)L_2(y_1,y_2)=L_1(x_1,y_2)L_2(y_1,x_2){\abrOSCR}^+_{12}(x_2|y_1,y_2)
		\end{equation} 
		and 
		\begin{equation}\label{eq-Rm}
			{\abrOSCR}^-_{12}(x_1,x_2|y_1) L_1(x_1,x_2)L_2(y_1,y_2)=L_1(y_1,x_2)L_2(x_1,y_2){\abrOSCR}^-_{12}(x_1,x_2|y_1)\,,
		\end{equation} 
        where the Lax matrices $L_1(x_1,x_2)$ and $L_2(x_1,x_2)$ are  given in~\eqref{L1-u1u2} and~\eqref{L2-u1u2}.
		 A solution  to~\eqref{eq-Rp} and~\eqref{eq-Rm}  in terms of oscillators is given by
		\begin{equation}\label{Rp-solution}
			{\abrOSCR}_{12}^+(x_2|y_1,y_2) =\kappa_+(x_2|y_1,y_2)e^{-\OAD^{[1]}\OA^{[2]}}\frac{\Gamma(\OAD^{[2]}\OA^{[2]}-y_1+x_2)}{\Gamma(\OAD^{[2]}\OA^{[2]}-y_1+y_2)}e^{\OAD^{[1]}\OA^{[2]}}
		\end{equation}
		and
		\begin{equation}\label{Rm-solution}
			{\abrOSCR}_{12}^-(x_1,x_2|y_1)=\kappa_-(x_1,x_2|y_1)e^{-\OAD^{[2]}\OA^{[1]}}\frac{\Gamma(\OAD^{[1]}\OA^{[1]}+x_2-y_1)}{\Gamma(\OAD^{[1]}\OA^{[1]}+x_2-x_1)}e^{\OAD^{[2]}\OA^{[1]}}\,.
		\end{equation}
		{Here $\kappa_{\pm}$ denotes an overall normalization that is not fixed by the Yang-Baxter equation~\eqref{RLL}. }
		 		\end{theorem}

\subsubsection{Proof of Theorem~\ref{Theorem-main}.}
The proof of Theorem~\ref{Theorem-main} is constructive and closely follows the case of $M=1$ worked out in \cite{Derkachov:2005hw}. We first show that the two equations~\eqref{eq-Rp} and~\eqref{eq-Rm} that determine $\abrOSCR^\pm$ imply the  Yang-Baxter equation~\eqref{RLL} for the R-matrix in the factorised form~\eqref{Facrtorized-R}, see Lemma~\ref{lem1}. Next, focussing on the proof of~\eqref{eq-Rm} (the one of~\eqref{eq-Rp} is analogous), we derive in Lemma~\ref{Lemma-equivalence} a set of equations that is equivalent to~\eqref{eq-Rm}. Finally, Lemma~\ref{Lemma-solution-def} states that $\abrOSCR^-$ solves the equations obtained by Lemma~\ref{Lemma-equivalence}. This concludes the derivation of formula \eqref{Rm-solution} for $\abrOSCR^-$. The derivation of \eqref{Rp-solution} for $\abrOSCR^+$ works similarly and is briefly discussed in Remark~\ref{Remark-differnce-Rp-Rm-derivation}.

\begin{lemma}\label{lem1}
Given two solutions  $\abrOSCR^\pm$ of the equations~\eqref{eq-Rp} and~\eqref{eq-Rm}, a solution $\abrOSCR$ to the Yang-Baxter equation~\eqref{RLL} can be obtained by taking the product~\eqref{Facrtorized-R}.
\end{lemma}

\begin{proof}
Let us first consider the factorisation $\abrOSCR(x-y)=P\abrOSCR^{+}(x_{2}|x_{1},y_{2})\abrOSCR^{-}(x_{1},x_{2}|y_{1})$.  Starting from~\eqref{eq-Rm}
and multiplying both sides by $\abrOSCR^{+}(x_{2}|x_{1},y_{2})$ we have that
\begin{equation}\label{here-R}
\abrOSCR^{+}_{12}(x_{2}|x_{1},y_{2})\abrOSCR^-_{12}(x_1,x_2|y_1) L_1(x_1,x_2)L_2(y_1,y_2)=\abrOSCR^{+}_{12}(x_{2}|x_{1},y_{2})L_1(y_1,x_2)L_2(x_1,y_2)\abrOSCR^-_{12}(x_1,x_2|y_1)\,.
\end{equation}
On right hand side, the operator $\abrOSCR^+$  can then be exchanged with the two L-operators using
\eqref{eq-Rp}:
\begin{equation}
\abrOSCR^{+}_{12}(x_{2}|x_{1},y_{2})\abrOSCR^-_{12}(x_1,x_2|y_1) L_1(x_1,x_2)L_2(y_1,y_2)=L_1(y_1,y_2)L_2(x_1,x_2)\abrOSCR^{+}_{12}(x_{2}|x_{1},y_{2})\abrOSCR^-_{12}(x_1,x_2|y_1)\,,
\end{equation}
Multiplying with the permutation operator $P$ than yields  the Yang-Baxter relation~\eqref{RLL}. A similar argument can be used for the case where the R-matrix factorises as $\abrOSCR(x-y)= P\abrOSCR^-(x_1,y_2|y_1)\abrOSCR^+(x_2|y_1,y_2)$.

\end{proof}

\begin{lemma}\label{Lemma-equivalence}
	For a shift-invariant operator $\abrOSCR^{-}(x_1,x_2|y_1)=\abrOSCR^{-}(x_1+\lambda,x_2+\lambda|y_1+\lambda)$, the equation in~\eqref{eq-Rm} is equivalent to
	\begin{equation}\label{goal-2a}
		\begin{split}
			&\abrOSCR_{12}^{-}(x_1,x_2|y_1)\left[L_{1}(x_{1},x_{2})+L_{2}(y_{1},y_{2})\right]=\left[L_{1}(y_{1},x_{2})+L_{2}(x_{1},y_{2})\right]\abrOSCR_{12}^{-}(x_1,x_2|y_1),\\
			\end{split}
	\end{equation}
	and
	\begin{equation}\label{goal-2b}
		\begin{split}
			&\abrOSCR_{12}^{-}(x_1,x_2|y_1)\oad_{a}^{[2]}=\oad^{[2]}_a\abrOSCR_{12}^{-}(x_1,x_2|y_1)\qquad \forall a\in \{1,2,\ldots,\M\}\,.
		\end{split}
	\end{equation}
\end{lemma}
\begin{proof}
    
 First, we show that~\eqref{eq-Rm} implies~\eqref{goal-2a} and~\eqref{goal-2b}.
In the following we often suppress the arguments of $\abrOSCR^{-}(x_{1},x_{2}|y_{1})$ and just write $\abrOSCR^{-}$ to shorten the presentation.
Let us consider the following shift of spectral parameters in~\eqref{eq-Rm}:
\begin{equation}\label{variable-shift}
	\begin{split}
		x_{1}\to x_{1}+\lambda,\quad x_{2}\to x_{2}+\lambda,\quad y_{1}\to y_{1}+\lambda \quad \text{and}\quad y_{2}\to y_{2}+\mu\,.
	\end{split}
\end{equation}
Under the shift the Lax matrices  transform as
\begin{equation}
	\begin{split}
		&L(x_{1}+\lambda,x_{2}+\lambda)= L(x_{1},x_{2})+\lambda\ID,\;\quad\qquad\qquad\qquad L(y_{1}+\lambda,x_{2}+\lambda)= L(y_{1},x_{2})+\lambda \ID\\
		&L(y_{1}+\lambda,y_{2}+\mu)=L(y_{1},y_{2})+\lambda\ID+(\mu-\lambda)F,\qquad 	L(x_{1}+\lambda,y_{2}+\mu)= L(x_{1},y_{2})+\lambda\ID+(\mu-\lambda)F
	\end{split}
\end{equation}
where
						\begin{eqnarray}
	F=\begin{pmatrix}
		0&-\OAD\\
		0&\ID
	\end{pmatrix}\,.
\end{eqnarray}
Therefore, using the assumption that $\abrOSCR^{-}$ is shift invariant, the relation~\eqref{eq-Rm} becomes
\begin{equation}\label{consequence-shift-1}
	\begin{split}
		&\abrOSCR_{12}^{-}(x_1,x_2|y_1)\left[\left(L_{1}(x_{1},x_{2})+\lambda\ID\right)\left(L_{2}(y_{1},y_{2})+\lambda\ID+(\mu-\lambda)F^{[2]}\right)\right]\\
		=&
		\left[\left(L_{1}(y_{1},x_{2})+\lambda\ID\right)\left(L_{2}(x_{1},y_{2})+\lambda\ID+(\mu-\lambda)F^{[2]}\right)\right]\abrOSCR_{12}^{-}(x_1,x_2|y_1)\,.
	\end{split}
\end{equation}
By equating the powers of $\lambda$ and $\mu$ one obtains the three conditions
	\begin{align}
			&\abrOSCR_{12}^{-}(x_1,x_2|y_1)\left[L_{1}(x_{1},x_{2})+L_{2}(y_{1},y_{2})\right]=\left[L_{1}(y_{1},x_{2})+L_{2}(x_{1},y_{2})\right]\abrOSCR_{12}^{-}(x_1,x_2|y_1)\label{list-1-first}\\
			&\abrOSCR_{12}^{-}(x_1,x_2|y_1)F^{[2]}=F^{[2]}\abrOSCR_{12}^{-}(x_1,x_2|y_1)\label{list-1-second}\\
			&\abrOSCR_{12}^{-}(x_1,x_2|y_1)L_{1}(x_{1},x_{2})F^{[2]}=L_{1}(y_{1},x_{2})F^{[2]}\abrOSCR_{12}^{-}(x_1,x_2|y_1)\label{list-1-third}\,.
		\end{align}
  Notice that~\eqref{list-1-third} is implied by~\eqref{list-1-first} and~\eqref{list-1-second}. This can be seen as follows.
   First use~\eqref{list-1-second} to bring~\eqref{list-1-third} to the form
   \begin{equation}
       \left[\abrOSCR^{-}(x_1,x_2|y_1)L_{1}(x_{1},x_{2})-L_{1}(y_{1},x_{2})\abrOSCR^{-}(x_1,x_2|y_1)\right]F^{[2]}=0\,.
   \end{equation}
   Next apply~\eqref{list-1-first} and again~\eqref{list-1-second} to obtain
 \begin{equation}
     \abrOSCR^{-}(x_1,x_2|y_1)L_{2}(x_{1},y_{2})F^{[2]}=L_{2}(y_{1},y_{2})F^{[2]}\abrOSCR^{-}(x_1,x_2|y_1)\,.
   \end{equation}
Finally  noting that
\begin{equation}
    L_{2}(y_{1},y_{2})F^{[2]}=(y_{2}-1)F^{[2]}\,,
\end{equation}
we see that~\eqref{list-1-third} follows from~\eqref{list-1-first} and~\eqref{list-1-second}.

Now, we show that~\eqref{goal-2a} and~\eqref{goal-2b} imply~\eqref{eq-Rm}.  The logic of this proof consists of three steps. First, one derives an equation equivalent to~\eqref{eq-Rm}. Second,  using the same shift of spectral parameters performed in~\eqref{variable-shift}, one obtains a set of equations that is equivalent to~\eqref{goal-2a} and~\eqref{goal-2b}. Third, one shows that this last set of equations implies~\eqref{eq-Rm}.

\paragraph{Step (i).} One shows that~\eqref{eq-Rm} is equivalent to
\begin{equation}\label{equivalent-def-2}
	\begin{split}
		\mathbf{r}L_{1}(x_{1},x_{2})\left(\begin{array}{cc}
			y_1&0\\
			\OA^{[2]}-\OA^{[1]}&\ID
		\end{array}\right)=
		L_{1}(y_{1},x_{2})\left(\begin{array}{cc}
			x_1&0\\
			\OA^{[2]}-\OA^{[1]}&\ID
		\end{array}\right)\mathbf{r}\,,
	\end{split}
\end{equation}
where
\begin{equation}\label{r-definition}
	\mathbf{r}:=\exp{\left(\OAD^{[2]}\OA^{[1]}\right)}\abrOSCR^{-}\exp{\left(-\OAD^{[2]}\OA^{[1]}\right)}\,.
\end{equation}
This is obtained as follows. The Lax operator is $gl(\M+1)$ invariant, i.e.
\begin{equation}\label{glinv}
	\left(e_{A,B}+ J_{A,B}\right)L(x)=L(x)(e_{A,B}+ J_{A,B})\qquad \forall A,B\in \{0,1,\ldots, \M\}\,.
\end{equation}
Here, $e_{A,B}$ denotes the algebra generators of the fundamental representation with $(e_{AB})_{CD}=\delta_{AC}\delta_{BD}$ and $J_{A,B}$ the algebra generators of the oscillator representation~\eqref{GL-oscRealization}. Consider $\mathbf{t}=(t_{1},\ldots,t_{M})$ such that $\mathbf{t}\OA=\sum_{a=1}^{\M}t_{a}\mathbf{a}_{a}$. By fixing $A=0$ in~\eqref{glinv} such that $J_{0,b}=\mathbf{a}_{b}$, $b=1,\ldots,\M$ and multiplying by $\mathbf{t}$ one finds that
\begin{equation}
	\begin{split}
		\left[\left(\sum_{b=1}^{\M}t_{b}e_{0,b}\right)+ \mathbf{t}\OA\right]L(x)=L(x)\left[\left(\sum_{b=1}^{\M}t_{b}e_{0,b}\right)+ \mathbf{t}\OA\right]\,.
	\end{split}
\end{equation}
This implies that
\begin{equation}
	\begin{split}G^{-1}(\mathbf{t})\exp{\left(\mathbf{t}\OA\right)}L(u)=L(u)G^{-1}(\mathbf{t})\exp{\left(\mathbf{t}\OA\right)}\,.
	\end{split}
\end{equation}
where
\begin{equation}
	G^{-1}(\mathbf{t})=\exp{\left[\sum_{b=1}^{\M}t_{b}e_{0,b}\right]}=\ID+\sum_{b=1}^{\M}t_{b}e_{0,b}
	 \,.
\end{equation}
Then, it follows that
\begin{align}\label{useful-relation}
	G^{-1}(\mathbf{t})L(u)G(\mathbf{t})=\exp{\left(-\mathbf{t}\OA\right)} L(u)\exp{\left(\mathbf{t}\OA\right)}\,.
\end{align}
Equation~\eqref{eq-Rm} is rewritten using~\eqref{eq:laxfac} and~\eqref{useful-relation}, where $t_{a}={\oad}_{a}^{[2]}$ for all $a\in \{1,2,\ldots,\M\}$. One obtains
\begin{equation}
	\begin{split}
		&\abrOSCR^{-}L_{1}(x_{1},x_{2})G(\OAD^{[2]})\left(\begin{array}{cc}
			y_{1}&0\\
			\OA^{[2]}&(y_{2}-1)\ID
		\end{array}\right)G^{-1}(\OAD^{[2]})
		\\=&
		L_{1}(y_{1},x_{2})G(\OAD^{[2]})\left(\begin{array}{cc}
			x_1&0\\
			\OA^{[2]}&(y_{2}-1)\ID
		\end{array}\right)G^{-1}(\OAD^{[2]})\abrOSCR^{-}\,,
	\end{split}
\end{equation}
which is equivalent to
\begin{equation}
	\begin{split}
		&G^{-1}(\OAD^{[2]})\abrOSCR^{-}L_{1}(x_{1},x_{2})G(\OAD^{[2]})\left(\begin{array}{cc}
			y_{1}&0\\
			\OA^{[2]}&(y_{2}-1)\ID
		\end{array}\right)
		\\=&
		G^{-1}(\OAD^{[2]})L_{1}(y_{1},x_{2})G(\OAD^{[2]})\left(\begin{array}{cc}
			x_1&0\\
			\OA^{[2]}&(y_{2}-1)\ID
		\end{array}\right)G^{-1}(\OAD^{[2]})\abrOSCR^{-}G(\OAD^{[2]})\,.
	\end{split}
\end{equation}
Since it has been shown that $\abrOSCR^{-}{\oad}_{a}^{[2]}={\oad}_{a}^{[2]}\abrOSCR^{-}$  for any solution of~\eqref{eq-Rm} one obtains
\begin{equation}
	\begin{split}
		&\abrOSCR^{-}G^{-1}(\OAD^{[2]})L_{1}(x_{1},x_{2})G(\OAD^{[2]})\left(\begin{array}{cc}
			y_{1}&0\\
			\OA^{[2]}&(y_{2}-1)\ID
		\end{array}\right)
		\\=&
		G^{-1}(\OAD^{[2]})L_{1}(x_{1},y_{2})G(\OAD^{[2]})\left(\begin{array}{cc}
			x_1&0\\
			\OA^{[2]}&(y_{2}-1)\ID
		\end{array}\right)\abrOSCR^{-}\,.
	\end{split}
\end{equation}
Next, using~\eqref{useful-relation} one finds
\begin{equation}
	\begin{split}
		&\abrOSCR^{-}\exp{\left(-\OAD^{[2]}\OA^{[1]}\right)}L_{1}(x_{1},x_{2})\exp{\left(\OAD^{[2]}\OA^{[1]}\right)}\left(\begin{array}{cc}
			y_{1}&0\\
			\OA^{[2]}&(y_{2}-1)\ID
		\end{array}\right)
		\\=&
		\exp{\left(-\OAD^{[2]} \OA^{[1]}\right)}L_{1}(y_{1},x_{2})\exp{\left(\OAD^{[2]}\OA^{[1]}\right)}\left(\begin{array}{cc}
			x_1&0\\
			\OA^{[2]}&(y_{2}-1)\ID
		\end{array}\right)\abrOSCR^{-}\,.
	\end{split}
\end{equation}
By computing the commutator
\begin{equation}
	\left[\OAD^{[2]}\OA^{[1]},\mathbf{a}_{b}^{[2]}\right]=-\mathbf{a}_{b}^{[1]}
\end{equation}
 and it follows, by the Hadamard formula,
 that
\begin{equation}
	\exp{\left(\OAD^{[2]}\OA^{[1]}\right)}\mathbf{a}_{a}^{[2]}=\left(\mathbf{a}_{a}^{[2]}-\mathbf{a}_{a}^{[1]}\right)	\exp{\left(\OAD^{[2]}\OA^{[1]}\right)} \qquad \forall a\in \{1,2,\ldots,\M\}\,.
\end{equation}
Therefore, one obtains that
\begin{equation}
	\begin{split}
		&\abrOSCR^{-}\exp{\left(-\OAD^{[2]}\OA^{[1]}\right)}L_{1}(x_{1},x_{2})\left(\begin{array}{cc}
			y_1&0\\
			\OA^{[2]}-\OA^{[1]}&(y_{2}-1)\ID
		\end{array}\right)\exp{\left(\OAD^{[2]}\OA^{[1]}\right)}
		\\=&
		\exp{\left(-\OAD^{[2]}\OA^{[1]}\right)}L_{1}(y_{1},x_{2})\left(\begin{array}{cc}
			x_1&0\\
			\OA^{[2]}-\OA^{[1]}&(y_{2}-1)\ID
		\end{array}\right)\exp{\left(\OAD^{[2]}\OA^{[1]}\right)}\abrOSCR^{-}\,,
	\end{split}
\end{equation}
which can be rewritten as
\begin{equation}
	\begin{split}
		&\exp{\left(\OAD^{[2]}\OA^{[1]}\right)}\abrOSCR^{-}\exp{\left(-\OAD^{[2]}\OA^{[1]}\right)}L_{1}(x_{1},x_{2})\left(\begin{array}{cc}
			y_1&0\\
			\OA^{[2]}-\OA^{[1]}&(y_{2}-1)\ID
		\end{array}\right)
		\\=&
		L_{1}(y_{1},x_{2})\left(\begin{array}{cc}
			x_1&0\\
			\OA^{[2]}-\OA^{[1]}&(y_{2}-1)\ID
		\end{array}\right)\exp{\left(\OAD^{[2]}\OA^{[1]}\right)}\abrOSCR^{-}\exp{\left(-\OAD^{[2]}\OA^{[1]}\right)}\,.
	\end{split}
\end{equation}
By recalling~\eqref{r-definition}
one writes~\eqref{eq-Rm} as
\begin{equation}\label{equivalent-def}
	\begin{split}
		\mathbf{r}L_{1}(x_{1},x_{2})\left(\begin{array}{cc}
			y_1&0\\
			\OA^{[2]}-\OA^{[1]}&(y_{2}-1)\ID
		\end{array}\right)=
		L_{1}(y_{1},x_{2})\left(\begin{array}{cc}
			x_1&0\\
			\OA^{[2]}-\OA^{[1]}&(y_{2}-1)\ID
		\end{array}\right)\mathbf{r}\,.
	\end{split}
\end{equation}
Finally,  observing that
\begin{equation}
	\left(\begin{array}{cc}
			y_1&0\\
			\OA^{[2]}-\OA^{[1]}&(y_{2}-1)\ID
		\end{array}\right)=\left(\begin{array}{cc}
		y_1&0\\
		\OA^{[2]}-\OA^{[1]}&\ID
	\end{array}\right)\left(\begin{array}{cc}
		1&0\\
		0&(y_{2}-1)\ID
	\end{array}\right)\,,
\end{equation}
\eqref{equivalent-def-2} is obtained.

\paragraph{Step (ii).} Next, a set of equations equivalent to~\eqref{goal-2a} and~\eqref{goal-2b} is derived. In particular, it  will be shown that~\eqref{goal-2a} and~\eqref{goal-2b} are equivalent to the following system of equations:
\begin{align}
				&\mathbf{r}{\oad}_{a}^{[1]}\mathbf{a}_{b}^{[1]}={\oad}_{a}^{[1]}\mathbf{a}_{b}^{[1]}\mathbf{r}\qquad \forall a,b \in \{1,\ldots,\M\}
    \label{set1} \\
		&\mathbf{r}\mathbf{a}_{a}^{[2]}=\mathbf{a}_{a}^{[2]}\mathbf{r}\qquad \forall a \in \{1,\ldots,\M\}
    \label{set2} \\
		&\mathbf{r}{\oad}_{a}^{[1]}(x_{2}-x_{1}+\OAD^{[1]}\OA^{[1]})={\oad}_{a}^{[1]}(x_{2}-y_{1}+\OAD^{[1]}\OA^{[1]})\mathbf{r}\qquad \forall a\in \{1,\ldots,\M\}\,.
    \label{set3}
        \end{align}

First, the same shift of parameters~\eqref{variable-shift} in~\eqref{equivalent-def} is applied, obtaining that \begin{equation}
	\begin{split}
		&\mathbf{r}\left(L_{1}(x_{1},x_{2})+\lambda \ID\right) \left[\left(\begin{array}{cc}
			y_1&0\\
			\OA^{[2]}-\OA^{[1]}&(y_{2}-1)\ID
		\end{array}\right)+\lambda\ID+(\mu-\lambda)\begin{pmatrix}
			0&0\\
			0&\ID
		\end{pmatrix}\right]\\
		=&
		\left(L_{1}(y_{1},x_{2})+\lambda \ID\right) \left[\left(\begin{array}{cc}
			x_1&0\\
			\OA^{[2]}-\OA^{[1]}&(y_{2}-1)\ID
		\end{array}\right)+\lambda\ID+(\mu-\lambda)\begin{pmatrix}
			0&0\\
			0&\ID
		\end{pmatrix}\right]\mathbf{r}\,.
	\end{split}
\end{equation}
Then, by equating the powers of $\lambda$ and $\mu$ one obtains
\begin{equation}\label{eqqq}
	\begin{split}
		&\mathbf{r}\left[L_{1}(x_{1},x_{2})+\left(\begin{array}{cc}
			y_1&0\\
			\OA^{[2]}-\OA^{[1]}&(y_{2}-1)\ID
		\end{array}\right)\right]=\left[L_{1}(y_{1},x_{2})+\left(\begin{array}{cc}
			x_1&0\\
			\OA^{[2]}-\OA^{[1]}&(y_{2}-1)\ID
		\end{array}\right)\right]\mathbf{r}
	\end{split}
\end{equation}
and
\begin{equation}\label{eq:rLI}
	\begin{split}
		\mathbf{r}L_{1}(x_{1},x_{2})\begin{pmatrix}
			0&0\\
			0&\ID
		\end{pmatrix}=L_{1}(y_{1},x_{2})\begin{pmatrix}
			0&0\\
			0&\ID
		\end{pmatrix}\mathbf{r}\,.
	\end{split}
\end{equation}
The second equation above is implied by the first one.
By inserting the explicit form of $L_{1}(x_{1},x_{2})$ and $L_{1}(y_{1},x_{2})$, cf.~\eqref{L1-u1u2}, into~\eqref{eqqq} the following equation is obtained
\begin{equation}\label{intermediate}
	\begin{split}
		&\mathbf{r} \left(\begin{array}{cc}
			x_{1}+y_{1}-\OAD^{[1]}\OA^{[1]}&-\OAD^{[1]}\left((x_{2}-x_{1}-1)\ID+\OA^{[1]}\OAD^{[1]}\right)\\
			\OA^{[2]}&(x_{2}+y_{2}-2)\ID+\OA^{[1]}\OAD^{[1]}
		\end{array}\right)
		\\=&
		\left(\begin{array}{cc}
			x_{1}+y_{1}-\OAD^{[1]}\OA^{[1]}&-\OAD^{[1]}\left((x_{2}-y_{1}-1)\ID+\OA^{[1]}\OAD^{[1]}\right)\\
			\OA^{[2]}&(x_{2}+y_{2}-2)\ID+\OA^{[1]}\OAD^{[1]}
		\end{array}\right)\mathbf{r}\,,
	\end{split}
\end{equation}
The elements of the matrix yield the relations in~\eqref{set1}-\eqref{set3}.
Thus, one concludes that they are equivalent to~\eqref{goal-2a} and~\eqref{goal-2b}. 
\paragraph{Step (iii).}  Finally, it is shown that the relations in~\eqref{set1}-~\eqref{set3} imply~\eqref{equivalent-def-2}.
                                  To achieve this, by using~\eqref{eq:rLI} and~\eqref{set2}, equation~\eqref{equivalent-def-2} is reduced to
\begin{equation}\label{eq:relRLv}
	\mathbf{r}L_{1}(x_{1},x_{2})\begin{pmatrix}
		y_{1}\\
		-\OA^{[1]}
	\end{pmatrix}=L_{1}(y_{1},x_{2})\begin{pmatrix}
		x_{1}\\
		-\OA^{[1]}
	\end{pmatrix}\mathbf{r}\,.
\end{equation}
By substituting the matrices $L_{1}(x_{1},x_{2})$ and $L_{1}(y_{1},x_{2})$ one finds
\begin{equation}\label{quasi-final}
	\begin{split}
		&\mathbf{r}
		\begin{pmatrix}
			\OAD^{[1]}\OA^{[1]}\left(x_{2}-y_{1}-x_{1}-1+\OAD^{[1]}\OA^{[1]}\right)\\
				-\OA^{[1]}\left(x_{2}-1-y_{1}+\OAD^{[1]}\OA^{[1]}\right)
            		\end{pmatrix}=
				\begin{pmatrix}
			\OAD^{[1]}\OA^{[1]}\left(x_{2}-y_{1}-x_{1}-1+\OAD^{[1]}\OA^{[1]}\right)\\
			-\OA^{[1]}\left(x_{2}-1-x_{1}+\OAD^{[1]}\OA^{[1]}\right)
            		\end{pmatrix}\mathbf{r}
	\end{split}
\end{equation}
The relation in the upper
component of~\eqref{quasi-final} follows directly from~\eqref{set1}.

Now, it will be shown that the first equation of~\eqref{quasi-final} follows from~\eqref{set3} and~\eqref{set1}. From~\eqref{set3} it follows that
\begin{equation}
    	\mathbf{r}{\OAD}^{[1]}(x_{2}-x_{1}+\OAD^{[1]}\OA^{[1]})\OA^{[1]}={\OAD}^{[1]}(x_{2}-y_{1}+\OAD^{[1]}\OA^{[1]})\mathbf{r}\OA^{[1]}\,,
\end{equation}
by using~\eqref{set1} one obtains
\begin{equation}
    	\mathbf{r}{\OAD}^{[1]}(x_{2}-x_{1}+\OAD^{[1]}\OA^{[1]})\OA^{[1]}={\OAD}^{[1]}\mathbf{r}(x_{2}-y_{1}+\OAD^{[1]}\OA^{[1]})\OA^{[1]}\,.
\end{equation}
And thus
\begin{equation}
    	\mathbf{r}{\OAD}^{[1]}\OA^{[1]}(x_{2}-x_{1}-1+\OAD^{[1]}\OA^{[1]})={\OAD}^{[1]}\mathbf{r}\OA^{[1]}(x_{2}-y_{1}-1+\OAD^{[1]}\OA^{[1]})\,.
\end{equation}
Then, using~\eqref{set1} again, one obtains
\begin{equation}
    	{\OAD}^{[1]}\left\{\OA^{[1]}(x_{2}-x_{1}-1+\OAD^{[1]}\OA^{[1]})\mathbf{r}-\mathbf{r}\OA^{[1]}(x_{2}-y_{1}-1+\OAD^{[1]}\OA^{[1]})\right\}=0\,,
\end{equation}
that implies the first equation of~\eqref{quasi-final}.
\end{proof}

\begin{lemma}\label{Lemma-solution-def}
	The following operator
	\begin{equation}
		\abrOSCR^-(x_1,x_2|y_1)=\kappa_-(x_1,x_2|y_1)\,e^{-\OAD^{[2]}\OA^{[1]}}\frac{\Gamma(\OAD^{[1]}\OA^{[1]}+x_2-y_1)}{\Gamma(\OAD^{[1]}\OA^{[1]}+x_2-x_1)}e^{\OAD^{[2]}\OA^{[1]}}
	\end{equation}
	solves equations~\eqref{goal-2a} and~\eqref{goal-2b}.
\end{lemma}

\begin{proof}
Recall that~\eqref{goal-2a} and~\eqref{goal-2b} are equivalent to~\eqref{set1}-~\eqref{set3} then, it is enough to construct a solution of the latter. The first two conditions in~\eqref{set1}-\eqref{set3} are satisfied by a function of the color-blind number operator $\OAD^{[1]}\OA^{[1]}$, therefore one writes
\begin{align}
	\mathbf{r}=\mathbf{r}(\OAD^{[1]}\OA^{[1]})\,.
\end{align}
Then, by substituting the above $\mathbf{r}$ in~\eqref{set3}  one obtains
\begin{equation}
	\mathbf{r}(\OAD^{[1]}\OA^{[1]}){\oad}_{a}^{[1]}(x_{2}-x_{1}+\OAD^{[1]}\OA^{[1]})={\oad}_{a}^{[1]}(x_{2}-y_{1}+\OAD^{[1]}\OA^{[1]})\mathbf{r}(\OAD^{[1]}\OA^{[1]})\,.
\end{equation}
Using the exchange relations
\begin{equation}
	\mathbf{r}(\OAD^{[1]}\OA^{[1]})\oad_{a}=\oad_{a}\mathbf{r}(\OAD^{[1]}\OA^{[1]}+1)\,,
\end{equation}
if follows that
\begin{equation}
	\mathbf{r}(\OAD^{[1]}\OA^{[1]}+1)(x_{2}-x_{1}+\OAD^{[1]}\OA^{[1]})=\mathbf{r}(\OAD^{[1]}\OA^{[1]})(x_{2}-y_{1}+\OAD^{[1]}\OA^{[1]})\,.
\end{equation}
The above recurrence relation is solved by
\begin{align}
	\mathbf{r}(x_{1},x_{2}|y_{1})=\kappa_-(x_{1},x_{2}|y_{1})\,\frac{\Gamma(\OAD^{[1]}\OA^{[1]}+x_{2}-y_{1})}{\Gamma(\OAD^{[1]}\OA^{[1]}+x_{2}-x_{1})}\,,
\end{align}
where $\kappa_-(x_{1},x_{2}|y_{1})$ denotes a normalisation factor.
Then, using the definition~\eqref{r-definition}, one has that
\begin{equation}
	\abrOSCR^{-}(x_{1},x_{2}|y_{1})=\exp{\left(-\OAD^{[2]}\OA^{[1]}\right)}\mathbf{r}(x_{1},x_{2}|y_{1})\exp{\left(\OAD^{[2]}\OA^{[1]}\right)}\,.
\end{equation}
\end{proof}

\begin{remark}\label{Remark-differnce-Rp-Rm-derivation}
    The proof for $\abrOSCR^{+}(x_{2}|y_{1},y_{2})$ is similar. The main difference is that one has to do the following shift of the spectral parameters:
    \begin{equation}
        x_{1}\to x_{1}+\mu,\quad x_{2}\to x_{2}+\lambda,\quad y_{1}\to y_{1}+\lambda,\quad y_{2}\to y_{2}+\lambda\,.
    \end{equation}
    After this shift, the logic and the steps of the proof are analogous. 
 A more elegant way  is provided by Proposition~\ref{permsym} presented in the following section.
\end{remark}

\subsection{Symmetries}
In this section, we establish some symmetry properties of the Lax operators \eqref{eq:laxfac} and of the R-matrix. The inversion symmetry presented in Section~\ref{sec:inL} is essential for solving the boundary Yang–Baxter equation, see Section \ref{Section-BYBE}, whereas the permutation symmetry presented in Section~\ref{sec:persym} yields an elegant construction of $\abrOSCR^{+}$ from $\abrOSCR^{-}$ of Theorem \ref{Theorem-main}.

\subsubsection{Unitarity relations of the Lax matrix}\label{sec:inL}
We prove here two propositions of unitarity and crossing-unitarity of the Lax matrices that will be used in the construction of the transfer matrix in relation to the boundary Yang-Baxter equation in Section~\ref{sec:openchain}.
\begin{proposition}\label{prop:Lbar}
The Lax matrix
	\begin{equation}\label{L-bar-definition}
		\bar L(x):=L(x-\mu_1-\mu_2+1)
	\end{equation}
satisfies the so-called unitarity relation 
	\begin{equation}
		L(x)\bar L(-x)=-(x+\mu_1)(x+\mu_2-1)\ID
	\end{equation}
    that relies on the characteristic identity 
    \begin{equation}\label{eq:genrec}
       \sum_{C=0}^{\M} J_{CA}J_{BC}-(\mu_1+\mu_2-1)J_{AB}+\mu_1(\mu_2-1)\delta_{AB}=0\,.
    \end{equation}
\end{proposition}

\begin{proof} It is convenient to write $L(x)=x+J^t$. We then find
\begin{equation}
    L(x)\bar L(-x)=-x(x+\mu_1+\mu_2-1)-(\mu_1+\mu_2-1)J^t+(J^t)^2
\end{equation}
Next we use the factorised form of the Lax matrix to obtain
\begin{equation}
\begin{split}
		(J^t)^2&=\left(\begin{array}{cc}
			1&-\OAD\\
			0&\ID
		\end{array}\right)\left(\begin{array}{cc}
			\mu_1&0\\
			\OA&(\mu_2-1)\ID
		\end{array}\right)\left(\begin{array}{cc}
			\mu_1&0\\
			\OA&(\mu_2-1)\ID
		\end{array}\right)\left(\begin{array}{cc}
			1&\OAD\\
			0&\ID
		\end{array}\right)\\
        &=\left(\begin{array}{cc}
			1&-\OAD\\
			0&\ID
		\end{array}\right)\left(\begin{array}{cc}
			\mu_1^2&0\\
			(\mu_1+\mu_2-1)\OA&(\mu_2-1)^2\ID
		\end{array}\right)\left(\begin{array}{cc}
			1&\OAD\\
			0&\ID
		\end{array}\right)\\
        &=\left(\begin{array}{cc}
			1&-\OAD\\
			0&\ID
		\end{array}\right)\left[(\mu_1+\mu_2-1)\left(\begin{array}{cc}
			\mu_1&0\\
			\OA&(\mu_2-1)\ID
		\end{array}\right)-\mu_1(\mu_2-1)\left(\begin{array}{cc}
			1&0\\
			0&\ID
		\end{array}\right)\right]\left(\begin{array}{cc}
			1&\OAD\\
			0&\ID
		\end{array}\right)\\
        &=(\mu_1+\mu_2-1)J^t-\mu_1(\mu_2-1)
        \end{split}
	\end{equation}
    which concludes the proof.
\end{proof} 
\begin{remark}
   Noting that
    \begin{equation}
        \bar L(-x)=L(1-x_2,1-x_1)
    \end{equation}
    it follows 
    \begin{equation}
    \label{L-symmetry-eq}
		L(x_1,x_2)L(1-x_2,1-x_1)=-x_1(x_2-1)\ID
	\end{equation}
\end{remark}
We further note that
\begin{proposition}\label{crossL}
The Lax matrix satisfies the crossing-unitarity relation
\begin{equation}
 L^t(x)L^t(-x-\mu_1-\mu_2-M)=L^t(x)\bar L^t(-x-(M+1))=-(x+\mu_1+M)(x+\mu_2)\ID\,.
\end{equation}
\end{proposition}

\begin{proof}
 As above we write $L(x)=x+J^t$ such that
 \begin{equation}
  L^t(x)\bar L^t(-x-(M+1))=-x(x+\mu_1+\mu_2+M)-(\mu_1+\mu_2+M)J +J^2
 \end{equation}
 Further noting that
\begin{equation}
		J=\left(\begin{array}{cc}
			1&\OAD\\
			0&\ID
		\end{array}\right)^t\left(\begin{array}{cc}
			\mmu_1+M&0\\
			\OA&\mmu_2\ID
		\end{array}\right)^t\left(\begin{array}{cc}
			1&-\OAD\\
			0&\ID
		\end{array}\right)^t\,,
	\end{equation}
	we proceed as above and find
	\begin{equation}
\begin{split}
		J^2&=\left(\begin{array}{cc}
			1&\OAD\\
			0&\ID
		\end{array}\right)^t\left(\begin{array}{cc}
			\mmu_1+M&0\\
			\OA&\mmu_2\ID
		\end{array}\right)^t\left(\begin{array}{cc}
			\mmu_1+M&0\\
			\OA&\mmu_2\ID
		\end{array}\right)^t\left(\begin{array}{cc}
			1&-\OAD\\
			0&\ID
		\end{array}\right)^t\\
        &=\left(\begin{array}{cc}
			1&\OAD\\
			0&\ID
		\end{array}\right)^t\left(\begin{array}{cc}
			(\mu_1+M)^2&0\\
			(\mu_1+\mu_2+M)\OA&\mu_2^2\ID
		\end{array}\right)^t\left(\begin{array}{cc}
			1&-\OAD\\
			0&\ID
		\end{array}\right)^t\\
        &=\left(\begin{array}{cc}
			1&\OAD\\
			0&\ID
		\end{array}\right)^t\left[(\mu_1+\mu_2+M)\left(\begin{array}{cc}
			\mu_1+M&0\\
			\OA&\mu_2\ID
		\end{array}\right)-(M+\mu_1)\mu_2\left(\begin{array}{cc}
			1&0\\
			0&\ID
		\end{array}\right)\right]\left(\begin{array}{cc}
			1&-\OAD\\
			0&\ID
		\end{array}\right)^t\\
        &=(\mu_1+\mu_2+M)J-(M+\mu_1)\mu_2
        \end{split}
	\end{equation}
	which concludes the proof.
\end{proof}

  \subsubsection{Permutation symmetry} \label{sec:persym}

\begin{proposition}\label{permsym}
    Given a solution $\abrOSCR^-$ of~\eqref{eq-Rm} with normalization $\kappa^{-}$, we can construct a solution  $\abrOSCR^+$ to~\eqref{eq-Rp} with normalization $\kappa^{+}$ via 
    \begin{equation}
      \abrOSCR_{12}^{+}(x_{2}|y_{1},y_{2})=f(x_2|y_1,y_2) P\abrOSCR^-_{12}(1-y_2,1-y_1|1-x_2)P\,,
    \end{equation}
where     \begin{equation}
        f(x_2|y_1,y_2) =\frac{\kappa_+ (x_2|y_1,y_2) }{\kappa_- (1-y_2,1-y_1|1-x_2)}\,.
    \end{equation}
\end{proposition}
\begin{proof}
We start from the relation 
		\begin{equation}
			\abrOSCR^-_{12}(x_1,x_2|y_1) L_1(x_1,x_2)L_2(y_1,y_2)=L_1(y_1,x_2)L_2(x_1,y_2)\abrOSCR^-_{12}(x_1,x_2|y_1)\,,
		\end{equation}   
which is equivalent, by~\eqref{L-symmetry-eq}, to 
\begin{equation}
			P\abrOSCR^-_{12}(x_1,x_2|y_1)P L_1(1-y_2,1-y_1)L_2(1-x_2,1-x_1)=L_1(1-y_2,1-x_1)L_2(1-x_2,1-y_1)P\abrOSCR^-_{12}(x_1,x_2|y_1)P\,.
		\end{equation}
By performin a shift of spectral paramters, we obtain
\begin{equation}
			P\abrOSCR^-_{12}(1-y_2,1-y_1|1-x_2)P L_1(x_1,x_2)L_2(y_1,y_2)=L_1(x_1,y_2)L_2(y_1,x_2)P\abrOSCR^-_{12}(1-y_2,1-y_1|1-x_2)P\,.
\end{equation}
Finally, after multiplying by the function $ f(x_2|y_1,y_2) $, we find that this is the same equation stated in~\eqref{eq-Rp} when identifying 
 \begin{equation}
       \abrOSCR^{+}(x_{2}|y_{1},y_{2})= f(x_2|y_1,y_2) P\abrOSCR^-(1-y_2,1-y_1|1-x_2)P\,.
    \end{equation}
    This concludes the proof.
\end{proof}

\subsection{Hamiltonian density}

In this section, we obtain the Hamiltonian density associated to the R-matrix of Theorem~\eqref{Theorem-main} by taking the logarithmic derivative at the permutation point, see e.g. \cite{Faddeev:1996iy}.

\begin{proposition}\label{proposition-Hdensity}
Consider the factorised R-matrix with equal representation labels in both spaces $\mu_1^{[i]}=\mu_{1}$ and $\mu_{2}^{[i]}=\mu_2$ for $i\in\{1,2\}$ 
 \begin{equation}
 \abrOSCR(x)=P\,\abrOSCR^-(x+\mu_1,\mu_2|\mu_1)\abrOSCR^+(x+\mu_2|\mu_1,\mu_2)\,.
 \end{equation}
 Assume that  $\abrOSCR^-(\mu_1,\mu_2|\mu_1)=\abrOSCR^+(\mu_2|\mu_1,\mu_2)=\ID$,
  and define the Hamiltonian density as
 \begin{equation}
  \abrOSCH:=\frac{\partial}{\partial x}\log \abrOSCR(x)\bigg|_{x=0}\,.
 \end{equation}
 Then, it decomposes into 
 \begin{equation}
     \abrOSCH=\abrOSCH^++\abrOSCH^-\,,
 \end{equation}
where, in the oscillator realization \eqref{eq:osc}, 
\begin{equation}\label{Hp-from-Rp}
\begin{split}
			\abrOSCH^+=&\left.                     \frac{\partial}{\partial x}\abrOSCR^+(x+\mu_2|\mu_1,\mu_2)\right|_{x=0}
            = {\color{black}\frac{\partial}{\partial x} \kappa_{+}(x+\mu_2|\mu_1,\mu_2)\bigg|_{x=0}
                                  +e^{-\OAD^{[1]}\OA^{[2]}}\psi(\OAD^{[2]}\OAD^{[2]}+2s)e^{\OAD^{[1]}\OA^{[2]}},}\\
                                        \end{split}
		\end{equation}
        and
        \begin{equation}\label{Hm-from-Rm}
            \begin{split}
                \abrOSCH^- =&                \left.\frac{\partial}{\partial x}\abrOSCR^-(x+\mu_1,\mu_2|\mu_1)\right|_{x=0}
                =
                                {\color{black}\frac{\partial}{\partial x}\kappa_{-}(x+\mu_1,\mu_2|\mu_1)\bigg|_{x=0}
                                +e^{-\OAD^{[2]}\OA^{[1]}}\psi(\OAD^{[1]}\OAD^{[1]}+2s)e^{\OAD^{[2]}\OA^{[1]}}}\,.
            \end{split}
        \end{equation}
        Above, we set $2s=\mu_2-\mu_1$.
\end{proposition}
\begin{proof}
   First we observe that assuming that $\abrOSCR^{\pm}(0)=\ID$ implies that $\kappa_{-}(\mu_{1},\mu_{2}|\mu_{1})=\kappa_{+}(\mu_{2}|\mu_{1},\mu_{2})=1$. A direct computation shows that
\begin{equation}
\begin{split}
  \abrOSCH=&\bigg(P\,\abrOSCR^-(\mu_1,\mu_2|\mu_1)\abrOSCR^+(\mu_2|\mu_1,\mu_2)\bigg)^{-1}\frac{\partial}{\partial x}P\,\abrOSCR^-(x+\mu_1,\mu_2|\mu_1)\abrOSCR^+(x+\mu_2|\mu_1,\mu_2)\bigg|_{x=0}
      \\=&  \left.\frac{\partial}{\partial x}\abrOSCR^-(x+\mu_1,\mu_2|\mu_1)\right|_{x=0}\abrOSCR^+(\mu_2|\mu_1,\mu_2)
 +\left.\abrOSCR^-(\mu_1,\mu_2|\mu_1)\frac{\partial}{\partial x}\abrOSCR^+(x+\mu_2|\mu_1,\mu_2)\right|_{x=0}
 \\=& \left.\frac{\partial}{\partial x}\abrOSCR^-(x+\mu_1,\mu_2|\mu_1)\right|_{x=0}
 +\left.\frac{\partial}{\partial x}\abrOSCR^+(x+\mu_2|\mu_1,\mu_2)\right|_{x=0}\,,
 \end{split}
\end{equation}
 where in the second equality we used $\abrOSCR(0)=P$ and in the last equality we used $\abrOSCR^-(\mu_1,\mu_2|\mu_1)=\abrOSCR^+(\mu_2|\mu_1,\mu_2)=\ID$. Finally, by replacing formulas \eqref{Rp-solution} \eqref{Rm-solution} and by using the identity $\psi(x)=\partial_x\log \Gamma(x)$ we obtain the result. 
\end{proof}

We now specify the normalization of $\abrOSCR^{\pm}$ as
 \begin{equation}\label{normali}
		 				\kappa_+(x_2|y_1,y_2)=\frac{\Gamma(-y_1+y_2)}{\Gamma(-y_1+x_2)}\,,\qquad
		 				\kappa_-(x_1,x_2|y_1)=\frac{\Gamma(x_2-x_1)}{\Gamma(x_2-y_1)}\,.
		 			\end{equation}
Choosing this normalization, the R-operators $ \abrOSCR^\pm$ reduce to the identity and, as a consequence, the R-operator reduces to permutation at $x=0$, i.e. $ \abrOSCR(0)=P$.
Furthermore, the choice of normalization \eqref{normali} yields a stochastic Hamiltonian density for the representations specified in Section \ref{sec:stoch} and Appendix \ref{section-integral-reps}.
\begin{corollary}\label{Hamalg}
 Fix the normalization \eqref{normali}. Then, the Hamiltonian density in the oscillator realization~\eqref{eq:osc} decomposes as follows 
 \begin{equation}\label{bulk-density-Hamiltonian}
  \abrOSCH=\abrOSCH^++\abrOSCH^-
 \end{equation}
with
			\begin{equation}\label{Rp-solutionD}
					\abrOSCH^+ = e^{-\OAD^{[1]}\OA^{[2]}}\left[\psi(\OAD^{[2]}\OA^{[2]}+2s)-
                    \psi(2s)
                    \right]e^{\OAD^{[1]}\OA^{[2]}}
		\end{equation}
        and
\begin{equation}\label{Rm-solutionD}
			\abrOSCH^-=e^{-\OAD^{[2]}\OA^{[1]}}\left[\psi(\OAD^{[1]}\OA^{[1]}+2s)-
            \psi(2s)
            \right]e^{\OAD^{[2]}\OA^{[1]}}\,,
		\end{equation}
where $2s=\mu_2-\mu_1$.

\end{corollary}

\begin{proof}
Inserting the normalizations \eqref{normali} into the expression for Hamiltonian matrices $\abrOSCH^{\pm}$ \eqref{Hp-from-Rp} and \eqref{Hm-from-Rm} of Proposition \ref{proposition-Hdensity} and using the identity $\psi(x)=\partial_x\log \Gamma(x)$ we find the expressions above. 
		\end{proof}

\subsection{Fock representation}\label{sec:stoch}
In this section we represent the oscillator algebra $\mathcal{A}_M$ on the Fock space
\begin{equation}\label{fock-space}
 \mathcal{F}=\spa \{|m_1,\ldots,m_M\rangle\}_{m_1,\ldots,m_N=0}^\infty
\end{equation}
such that $\pi:\mathcal{A}_M\to End(\mathcal{F})$. 
For convenience, we use the notation
    \begin{equation}
\pi(\oa_{a})={\of_{a}}\,,\qquad \pi(\oad_{a})=\ofd_{a}\,.
	\end{equation}
	such that
	\begin{equation}\label{osc-notation-descrete}
\pi(\OA)=\OF\,,\qquad \pi(\OAD)=\OFD\,.
	\end{equation}
	where  $\OF=(\of_1,\ldots,\of_M)^{t}$ and   $\OFD=(\ofd_1,\ldots,\ofd_M)$ .

The {orthonormal} basis is
obtained from the Fock vacuum $|0\rangle$, with $\of_a|0\rangle=0$ for $a=1,\ldots,M$,
via
\begin{equation}\label{basis-vectors-Fock}
|m_1,\ldots,m_M\rangle=\ofd_1^{m_1}\cdots \ofd_M^{m_M}|0\rangle
\end{equation}
and will be identified with the basis introduced in~\eqref{basis} for a given site $\ell$ such that $V_\ell\simeq \mathcal{F}$. It follows that the action of the creation and annihilation operators on the basis vectors~\eqref{basis-vectors-Fock} reads
\begin{equation}\label{discrete-representation}
\begin{split}
	&\oa_{a}|m\rangle=m_{a}|m_{1},\ldots,m_{a}-1,\ldots,m_{\M}\rangle,\quad
    \oad_a|m\rangle=|m_{1},\ldots,m_{a}+1,\ldots,m_{\M}\rangle
    \end{split}\,.
\end{equation}
This induces a representation of $gl(M+1)$ through the homomorphism~\eqref{GL-oscRealization}, which has the following generators
\begin{equation}\label{E-matrix}
		\pi(\mathrm{J}_{BA})=\left(\begin{array}{cc}
			\mmu_1-\OAD\OA&\OAD\left((\mmu_1-\mmu_2+1)\ID-\OA\OAD\right)\\
			\OA&(\mmu_2-1)\ID+\OA\OAD
		\end{array}\right) _{AB}\,,
	\end{equation}
    where we have used the oscillator notation of equation~\eqref{osc-notation-descrete}. 
The resulting representation is characterised as follows. The Fock vacuum is the highest-weight state, i.e.  $\pi(J_{AB})|0\rangle=0$ for $A<B$. The Dynkin labels of the resulting representation  are read of from the action of the Cartan elements on the Fock vacuum
    \begin{equation}\label{HW-discreteReps}
        \pi(J_{AA})|0\rangle=\begin{cases*}
                    \mu_1|0\rangle & if  $A =0$  \\
                     \mu_2|0\rangle & if $A=1,\ldots,M$
                 \end{cases*}\,,
    \end{equation}
For $\mu_1< \mu_2$, we have $2s=\mu_2-\mu_1\in \mathbb{R}_{>0}$ and the infinite-dimensional representation of $gl(M+1)$ is irreducible. For negative integer values $2s\in \mathbb{Z}_{<0}$ the representation becomes reducible and contains a finite-dimensional submodule. In the following we assume that $2s>0$ if not specified otherwise to avoid complications.

We further define the inner product via
\begin{equation}
 \langle n_1,\ldots,n_M|m_1,\ldots,m_M\rangle =\prod_{a=1}^M\delta_{n_a,m_a}\,.
\end{equation}
The definitions above naturally extend to two sites and we can represent the R-operator on the tensor product of two Fock spaces
\begin{equation}
 \pi(\abrOSCR)\in End(\mathcal{F}\otimes\mathcal{F})\,.
\end{equation}
This allows to regard the R-operator $\pi(\abrOSCR)$ as an infinite-dimensional matrix with entries
\begin{equation}
 \langle m,n|  \pi(\abrOSCR(x)) |m',n'\rangle\,.
\end{equation}
In the following, the representation on the Fock space $\pi$ will be suppressed when using the bra-ket notation and the R-matrix in this representation will be denoted by $\eR= \pi(\abrOSCR)$ and $\eR^\pm= \pi(\abrOSCR^\pm)$.

		 		\begin{proposition}\label{Proposition-R-components}
		 			When fixing the normalisation~\eqref{normali},
		 			the factors of the R-matrix~\eqref{Rp-solution} and~\eqref{Rm-solution} have components
		 			\begin{equation}\label{element-M}
		 				\begin{split}
		 					&\langle m,n|  \eR^+(x_2|y_1,y_2) |m',n'\rangle =  \frac{(x_2-y_1)_{|n|} (y_2-x_2)_{|n'|-|n|}}{(y_2-y_1)_{|n'|} }\prod_{a=1}^{\M}\binom{n'_a}{n_a}\delta_{n'_a+m'_a,n_a+m_a}\,,
		 				\end{split}
		 			\end{equation}
		 			
		 			\begin{equation}\label{element-N}
		 				\begin{split}
		 					&\langle m,n|  \eR^-(x_1,y_2|y_1) |m',n'\rangle = \frac{(y_2-y_1)_{|m|} (y_1-x_1)_{|m'|-|m|}}{(y_2-x_1)_{|m'|} }\prod_{a=1}^{\M}\binom{m'_a}{m_a}\delta_{n'_a+m'_a,n_a+m_a}\,,
		 				\end{split}
		 			\end{equation}
		 			while the permutation yields
		 			\begin{equation}\label{permutation-elements}
		 				\begin{split}
		 					&\langle m,n|  P |m',n'\rangle =\prod_{a=1}^{\M}\delta_{m_a,n'_a}\delta_{n_a,m'_a}\,.
		 				\end{split}
		 			\end{equation}
                 Moreover, for $x,y\in \mathbb{R}$ such that $(\mu_{1}^{[2]}-\mu_{2}^{[1]})\leq (x-y)\leq \min\{(\mu_{2}^{[2]}-\mu_{2}^{[1]}),(\mu_{1}^{[2]}-\mu_{1}^{[1]}),(\mu_{2}^{[2]}-\mu_{1}^{[1]})\}$ the matrices $\eR^{+}(x_2|y_1,y_2)$ and $\eR^{-}(x_1,y_2|y_1)$ are stochastic.
		 		\end{proposition}
\begin{proof}
The details are only provided for~\eqref{element-N}, since the proof of equation~\eqref{element-M} is analogous. First, the following matrix components are evaluated
\begin{equation}
	\begin{split}
	&\langle m,n|{\eR}^{-}(x_1,y_{2}|y_{1})|m',n'\rangle
	=\frac{\Gamma(y_{2}-x_{1})}{\Gamma(y_{2}-y_{1})}
	\langle m, n|e^{-\OFD^{[2]}\OF^{[1]}}\frac{\Gamma(\OFD^{[1]}\OF^{[1]}+y_2-y_1)}{\Gamma(\OFD^{[1]}\OF^{[1]}+y_2-x_1)}e^{\OFD^{[2]}\OF^{[1]}}|m',n'\rangle\,.
	\end{split}
\end{equation}
Notice that
\begin{equation}
	\begin{split}
e^{\OFD^{[2]}\OF^{[1]}}|m',n'\rangle =\sum_{k_{1}=0}^{\infty}\cdots\sum_{k_{\M}=0}^{\infty}\left(\prod_{a=1}^{\M}\binom{m_{a}{'}}{k_{a}} \left(\oad_{a}^{[2]}\right)^{k_{a}}\right)|m'-k,n'\rangle \,.
	\end{split}
\end{equation}
Therefore, one has that
\begin{equation}
	\begin{split}
&\langle m,  n|{\eR}^{-}(x_{1},y_{2}|y_{1})|m',n'\rangle
\\=&
\delta_{m+n,m'+n'}\frac{\Gamma(y_{2}-x_{1})}{\Gamma(y_{2}-y_{1})}\sum_{j_{1}=0}^{\infty}\cdots\sum_{j_{\M}=0}^{\infty}(-1)^{|j|}\frac{\Gamma(|m|+|j|+y_{2}-y_{1})}{\Gamma(|m|+|j|+y_{2}-x_{1})}\prod_{a=1}^{\M}\binom{m_{a}{'}}{m_{a}{'}-m_{a}-j_{a}} \binom{m_{a}+j_{a}}{j_{a}}
\\=&
\delta_{m+n,m'+n'}\frac{\Gamma(y_{2}-x_{1})}{\Gamma(y_{2}-y_{1})}\left(\prod_{a=1}^{\M}\binom{m_{a}{'}}{m_{a}}\right)\frac{\Gamma(|m|+y_{2}-y_{1})}{\Gamma(|m|+y_{2}-x_{1})}
\\&\times\text{F}_{\text{D}}^{(\M)}(|m|+y_{2}-y_{1},m_{1}-m_{1}{'},\ldots,m_{\M}-m_{\M}{'};|m|+y_{2}-x_{1};1,\ldots,1)\,.
	\end{split}
\end{equation}
Here,  $\text{F}_{\text{D}}^{(\M)}(a,b_1,\ldots,b_\M;c;z_1,\ldots,z_\M)$ denotes the Lauricella's hypergeometric series of type D, cf. \cite{lauricella1893sulle}. Using the integral representation for the Lauricella's hypergeometric series, one may writes the following relation
\begin{equation}
	\text{F}_{\text{D}}^{(\M)}(a,b_{1},\ldots,b_{\M};c;1,\ldots,1)=\frac{\Gamma(c)\Gamma(c-|b|-a)}{\Gamma(c-a)\Gamma(c-|b|)}\,,
\end{equation}
where we have denoted $|b|=\sum_{a=1}^{\M}b_{a}$. In the specific case of this proof one chooses 
\begin{equation}
	a=y_2-y_1+|m|,\qquad c=y_{2}-x_{1}+|m|,\qquad |b|=|m|-|m'|\,
\end{equation}
then it follows that
\begin{equation}
	\begin{split}
		\langle m, n|{\eR}^{-}(x_{1},y_{2}|y_{1})|m',n'\rangle
		&=
		\delta_{m+n,m'+n'}\frac{\Gamma(y_{2}-x_{1})}{\Gamma(y_{2}-y_{1})}\frac{\Gamma(|m|+y_{2}-y_{1})\Gamma(y_{1}-x_{1}+|m'|-|m|)}{\Gamma(y_{1}-x_{1})\Gamma(y_{2}-x_{1}+|m'|)}\prod_{a=1}^{\M}\binom{m_{a}{'}}{m_{a}}
		\\&=
		\delta_{m+n,m'+n'}\frac{\left(y_{2}-y_{1}\right)_{|m|}\left(y_{1}-x_{1}\right)_{|m'|-|m|}}{\left(y_{2}-x_{1}\right)_{|m'|}}\prod_{a=1}^{\M}\binom{m_{a}{'}}{m_{a}}\,.
		\end{split}
	\end{equation}
    We now prove the stochasticity. First, it is straightforward to show that the matrix elements are non-negative for $x,y\in \mathbb{R}$ such that  $(\mu_{1}^{[2]}-\mu_{2}^{[1]})\leq (x-y)\leq \min\{(\mu_{2}^{[2]}-\mu_{2}^{[1]}),(\mu_{1}^{[2]}-\mu_{1}^{[1]}),(\mu_{2}^{[2]}-\mu_{1}^{[1]})\}$. Indeed, under this condition all the arguments of the Pochhammer symbols of the matrix elements of $\eR^{\pm}$ written in \eqref{element-M} and \eqref{element-N} are non-negative.

    Then, we show that the sum of the components of each column is $1$. Due to the presence of the binomial coefficients, we prove that
        \begin{equation}\label{stoch-RpRm}
    \sum_{n=0}^{n'}\sum_{m=0}^{m'}\langle m,n|  {\eR}^+(x_2|y_1,y_2) |m',n'\rangle=1 \qquad\text{and}\qquad \sum_{n=0}^{n'}\sum_{m=0}^{m'}\langle m,n|  {\eR}^-(x_1,y_2|y_1) |m',n'\rangle=1\,.
\end{equation} 
Above we used the notation $\sum_{n=0}^{n'}=\sum_{n_{1}=0}^{n'_{1}}\cdots \sum_{n_{\M}=0}^{n'_{\M}}$. We only show the first relation, since the second one can be proven similarly. First, we observe that it is enough to sum over $n$ only, because of the term $\prod_{a=1}^{\M}\delta_{m_{a}+n_{a},m'_{a}+n'_{a}}$. Moreover, by using the integral representation of the Beta function, we obtain 
\begin{equation}
\begin{split}
                     & \sum_{n=0}^{n'}\sum_{m=0}^{m'}\langle m,n| \eR^{+}(x_2| y_1, y_2) |m',n'\rangle\\
        &=\sum_{n=0}^{n'} \frac{(x_2-y_1)_{|n|} (y_2-x_2)_{|n'|-|n|}}{(y_2-y_1)_{|n'|} }\prod_{a=1}^{\M}\binom{n'_a}{n_a}
        \\&=
         \frac{\Gamma(y_{2}-y_{1})}{\Gamma(x_{2}-y_{1})\Gamma(y_{2}-x_{2})}\int_{0}^{1}t^{x_{2}-y_{1}-1}(1-t)^{y_{2}-x_{2}+|n'|-1}
         \sum_{n_{1}=0}^{n'_{1}}\cdots \sum_{n_{\M}=0}^{n'_{\M}}\prod_{a=1}^{\M}\binom{n'_{a}}{n_{a}}\left(\frac{t}{1-t}\right)^{n_{a}}dt\,.
        \end{split}
    \end{equation}
     By the Newton binomial we have that  
    \begin{equation}
         \sum_{n=0}^{n'}\prod_{a=1}^{\M}\binom{n'_{a}}{n_{a}}\left(\frac{t}{1-t}\right)^{n_{a}}=\prod_{a=1}^{\M}\left(1+\frac{t}{1-t}\right)^{n'_{a}}=(1-t)^{-|n'|}\,.
    \end{equation}    
    Therefore, we obtain that
    \begin{equation}
        \begin{split}
            &\sum_{n=0}^{n'}\sum_{m=0}^{m'}\langle m,n| \eR^{+}(x_2| y_1, y_2) |m',n'\rangle=
            \frac{\Gamma(y_{2}-y_{1})}{\Gamma(x_{2}-y_{1})\Gamma(y_{2}-x_{2})}\int_{0}^{1}t^{x_{2}-y_{1}-1}(1-t)^{y_{2}-x_{2}-1}dt=1\,.
                                 \end{split}
    \end{equation}
                                  	\end{proof}

\begin{corollary}\label{Corollary-R-components-total}
		 			The R-matrix in components takes the form
		 			\begin{equation}\label{eq:fullR}
		 				\begin{split}
		 					\langle m,n|  {\eR}(x-y) |m',n'\rangle& =\frac{(x_2-x_1)_{|m|} }{(x_2-x_1)_{|m'|}}\frac{(y_1-x_1)_{|m'|}(y_2-x_2)_{|n|}}{(y_2-x_1)_{|m|+|n|}} \prod_{a=1}^{\M}\binom{m_a+n_a}{m_a} \delta_{m_a+n_a,m'_a+n'_a} \\
		 					&\quad\times\sum_{m''=0}^{m+n}
		 					\frac{\left(x_2-y_1\right)_{|m''|}\left(x_1-y_2+1-|m|-|n|\right)_{|m''|}}{\left(x_2-y_2+1-|n|\right)_{|m''|}\left(x_1-y_1+1-|m'|\right)_{|m''|}}
		 					\prod_{a=1}^{\M}\frac{\left({-n_{a}}\right)_{m''_{a}}\left(-m'_{a}\right)_{m''_{a}}}{\left(-m_{a}-n_{a}\right)_{m''_{a}}\,m''_{a}!}
		 					\,,
		 				\end{split}
		 			\end{equation}
		 			where we used the notation $\sum_{m''=0}^{m+n}=\sum_{m''_1=0}^{m_1+n_1}\cdots \sum_{m''_\M=0}^{m_\M+n_\M}$. {\color{black}Changing the order of factors in the factorization}, cf.~\eqref{Facrtorized-R}, we may also write
		 			\begin{equation}\label{eq:fullR1}
		 				\begin{split}
		 					\langle m,n|  {\eR}(x-y) |m',n'\rangle& =\frac{(y_2-y_1)_{|n|}}{(y_2-y_1)_{|{n'}|}} \frac{\left( y_2-x_2\right)_{|n'|} \left(y_1-x_1\right)_{|{m}|}}{\left(y_2-x_1\right)_{|m|+|n|}} \prod_{a=1}^{\M}\binom{m_a+n_a}{n_a} \delta_{n'_a+m'_a,n_a+m_a} \\
		 					&\quad\times\sum_{n''=0}^{m+n}\frac{\left(x_2-y_1\right)_{|n''|}\left(x_1-y_2+1-|m|-|n|\right)_{|n''|}}{\left( x_2-y_2+1-|n'|\right)_{|n''|}\left(x_1-y_1+1-|m|\right)_{|n''|}}  \prod_{a=1}^{\M} \frac{({-m_a})_{n_a''}(-n_a')_{n_a''}}{(-m_a-n_a)_{n_a''}\,n_a''!}
		 					\,.
		 				\end{split}
		 			\end{equation} 
                     Moreover,  for $x,y\in \mathbb{R}$ such that $(\mu_{1}^{[2]}-\mu_{2}^{[1]})\leq (x-y)\leq \min\{(\mu_{2}^{[2]}-\mu_{2}^{[1]}),(\mu_{1}^{[2]}-\mu_{1}^{[1]}),(\mu_{2}^{[2]}-\mu_{1}^{[1]})\}$ the matrix $\eR(x-y)$ is stochastic.  
                     		 		\end{corollary}
                
\begin{proof}

	Here only~\eqref{eq:fullR} is proved, since the demonstration of~\eqref{eq:fullR1} is analogous. By a direct computation we find
	\begin{equation}\label{eq:evR}
		\begin{split}
			\langle m,n|  {\eR}(x-y) |m',n'\rangle
		&=\sum_{n'',m''=0}^\infty \langle n,m|  {R}^+(x_2|x_1,y_2) |m'',n''\rangle \langle m'',n''|  {R}^-(x_1,x_2|y_1) |m',n'\rangle
			\\&=\delta_{m+n,{m'}+{n'}}\frac{(x_2-x_1)_{|m|}}{(x_2-x_1)_{|{m'}|}} \sum_{m''=0}^{\infty}\frac{\left( y_2-x_2\right)_{|n|-|m''|}\left(x_2-y_1\right)_{|m''|}\left(y_1-x_1\right)_{|{m'}|-|m''|}}{\left(y_2-x_1\right)_{|m|+|n|-|m''|}}
			\\&\qquad\qquad\qquad\qquad\qquad\qquad\qquad\times \prod_{a=1}^{\M}\binom{m_{a}+n_{a}-m''_{a}}{m_{a}}\binom{m'_{a}}{m''_{a}}
		\end{split}
	\end{equation}
	where we chose $n''=m+n-m''$ to respect the conservation rules. The indices in the bra-vector on the right hand side of the first line of~\eqref{eq:evR} have changed place because of the action of the permutation. Then for $m+n-m''\geq 0$ we use
	\begin{equation}
		\begin{split}
			\binom{m+n-m''}{m}\binom{m'}{m''}
			&=\binom{m+n}{m}\frac{(-1)^{m''}}{m''!}\frac{({-n})_{m''}(-m')_{m''}}{(-m-n)_{m''}}
		\end{split}
	\end{equation}
	that follows from the relation
	\begin{align}
		(-m)_{n}  =(-1)^{n}\frac{m!}{(m-n)!}\,,
	\end{align}
	for $m,n \in \mathbb{N}$.
	Finally using
	\begin{align}
		(x)_{m-n}=(-1)^{n}\frac{(x)_{m}}{(1-x-m)_{n}}
	\end{align}
	one obtains~\eqref{eq:fullR}.

Moreover, from the fact that, for fixed $m'$ and $n'$, we have
         \begin{equation}
	\begin{split}
		&\sum_{m,n}\langle m,n|{\eR}(x-y)|m',n'\rangle
		\\=&
		\sum_{m,n}\sum_{m'',n''}\langle n,m|  {\eR}^+(x_2|x_1,y_2) |m'',n''\rangle \langle m'',n''|  {\eR}^-(x_1,x_2|y_1) |m',n'\rangle
		\\=&
														              \sum_{m'',n''}\langle m'',n''|  {\eR}^-(x_1,x_2|y_1) |m',n'\rangle
				\sum_{i,j}\langle i,j|{\eR}^{+}(x_{2}|x_{1},y_{2})|m'',n''\rangle
        \\=&1\,.
	\end{split}
\end{equation}
Above, in the next to last equality we have used~\eqref{permutation-elements} and performed a change of summation indices, while in the last equality we used equations~\eqref{stoch-RpRm}. The non-negativity of the elements of $\eR(x-y)$ follows from $\eR^{\pm}$, see Proposition \ref{Proposition-R-components}.
 \end{proof}
We further observe that the elements of the matrix $\eR(x-y)$ of equation \eqref{eq:fullR} may be written as
\begin{equation}
\begin{split}
   & \langle m,n|  {\eR}(x-y) |m',n'\rangle =\frac{(x_2-x_1)_{|m|} }{(x_2-x_1)_{|m'|}}\frac{(y_1-x_1)_{|m'|}(y_2-x_2)_{|n|}}{(y_2-x_1)_{|m|+|n|}} \prod_{a=1}^{\M}\binom{m_a+n_a}{m_a} \delta_{m_a+n_a,m'_a+n'_a}
    \\&\times 
    {F}_{2;1;\ldots;1}^{2;2;\ldots;2}\left(
\begin{matrix}
(x_{2}-y_{2}),(x_{1}-y_{2}+1-|m|-|n|);\;\;(-n_{1}),\ldots,(-n_{\M});\;\;(-m_{1}),\ldots,(-m_{\M})\\
(x_{2}-y_{2}+1-|n'|), (x_{1}-y_{1}+1-|m|);\;\;(-n_{1}-m_{1}),\ldots,(-n_{\M}-m_{\M})
\end{matrix};\; 1,\ldots, 1
 \right)\,.
\end{split}
\end{equation}

Above we used the Srivastava–Daoust (Kampé de Fériet) type hypergeometric function, see \cite{Srivastava1987NeumannExpansions}, defined as 
\begin{equation}\label{SD-HPGF}
\begin{split}
        F_{q_{0},q_{1}\ldots,q_{\M}}^{p_{0},p_{1},\ldots, p_{\M}}\left(\begin{matrix}
        \bm{a}_{0};\; \bm{a}_{1};\ldots; \bm{a}_{\M}\\
        \bm{b_{0}};\; \bm{b}_{1};\ldots;\bm{b}_{\M}
    \end{matrix};\; x_{1},\ldots,x_{\M}\right):= \sum_{m_{1},\ldots,m_{\M}=0}^{\infty}\frac{\prod_{k=1}^{p_{0}}(a_{0}^{k})_{|m|}}{\prod_{k=1}^{q_{0}}(b_{0}^{k})_{|m|}}\prod_{j=1}^{\M}\frac{\prod_{k=1}^{p_{j}}(a_{j}^{k})_{m_{j}}}{\prod_{k=1}^{q_{j}}(b_{j}^{k})_{m_{j}}}\frac{x^{m_{j}}}{m_{j}!}\,.
    \end{split}
\end{equation}
Here we have denoted 
\begin{equation}
    \bm{a}_{j}= (a_{1},\ldots,a_{p_{j}})\quad \text{and}\quad \bm{b}_{j}=(b_{1},\ldots,b_{q_{j}})\qquad \text{for}\quad j\in \{0,1,\ldots, \M\}\,,
\end{equation}
where $p_{j},q_{j}\in \mathbb{N}$. A similar formula can be obtained from the matrix elements written in \eqref{eq:fullR1}. 

When restricting to $\M=1$, one has that the Srivastava–Daoust hypergeometric function \eqref{SD-HPGF} reduces to the generalized hypergeometric function ${}_4F_3\left(\cdot
\right)$, see \cite{Mangazeev:2014gwa}. Namely, for $\M=1$, one has that 
\begin{equation}
    \begin{split}
        &\langle m,n|  {\eR}(x-y) |m',n'\rangle =\frac{(x_2-x_1)_{m} }{(x_2-x_1)_{m'}}\frac{(y_1-x_1)_{m'}(y_2-x_2)_{n}}{(y_2-x_1)_{m+n}} \binom{m+n}{m} \delta_{m+n,m'+n'}
    \\\times & 
    {}_4F_3\left(
\begin{matrix}
(x_{2}-y_{2});(x_{1}-y_{2}+1-m-n);\;(-n)\\
(x_{2}-y_{2}+1-n'); (x_{1}-y_{1}+1-m);\;(-n-m)
\end{matrix};\; 1
 \right)\,.
    \end{split}
\end{equation}

We now fix $\mu_{i}^{[1]}=\mu_{i}^{[2]}$ for $i=1,2$, that is, we assign the same representation labels to the spaces 1 and 2. By means of the logarithmic derivative of the discrete representation of the stochastic R-matrix we obtain the Hamiltonian density of the harmonic process with generator~\eqref{bulk-generatorDensity}. \begin{corollary}\label{corollary-bulk-elements}
The component form of the Hamiltonian density $\eH$, obtained via the logarithmic derivative of the R-matrix $\eR(x)$ with components~\eqref{eq:fullR}, is given by~\eqref{hbulk}. Namely, we have that 
\begin{equation}
\label{ham-discrete}
\begin{split}
 \eH|m',n'\rangle=\eH^{+}|m',n'\rangle+\eH^{-}|m',n'\rangle\,,
  \end{split}
\end{equation}
where
\begin{equation}
    \eH^{+}|m',n'\rangle=h_s(|n'|)|m',n'\rangle-\sum_{k_{1}=0}^{n_{1}'}\cdots \sum_{k_{N}=0}^{n_{N}'}\varphi_{s}(k,n') |m'+k,n'-k\rangle\,
\end{equation}
and \begin{equation}
    \eH^{-}|m',n'\rangle=h_s(|m'|)|m',n'\rangle-\sum_{k_{1}=0}^{m_{1}'}\cdots \sum_{k_{N}=0}^{m_{N}'}\varphi_{s}(k,m') |m'-k,n'+k\rangle\,.
\end{equation}
Above the $\varphi_{s}(k,m)$ is the transition rate written in \eqref{rates}. 

 \end{corollary}

\begin{proof}
 Using Proposition~\ref{Proposition-R-components} we may write

\begin{equation}\label{element-Mt}
		 				\begin{split}
		 					& \eR^+(x_2|y_1,y_2) |m',n'\rangle =  \sum_{k=0}^{n'}\frac{(x_2-y_1)_{|n'|-|k|} (y_2-x_2)_{|k|}}{(y_2-y_1)_{|n'|} }\left[\prod_{a=1}^{\M}\binom{n_a'}{n_a'-k_a}\right] |m'+k,n'-k\rangle\,,
		 				\end{split}
		 			\end{equation}

		 			\begin{equation}\label{element-Nt}
		 				\begin{split}
		 					& \eR^-(x_1,y_2|y_1) |m',n'\rangle =\sum_{k=0}^{m'} \frac{(y_2-y_1)_{|m'|-|k|} (y_1-x_1)_{|k'|}}{(y_2-x_1)_{|m'|} }\left[\prod_{a=1}^{\M}\binom{m_a'}{m_a'-k_a}\right] | m'-k,n'+k\rangle \,,
		 				\end{split}
		 			\end{equation}
Then we  find
\begin{equation}\label{element-Mr}
		 				\begin{split}
		 					\frac{\partial}{\partial x}  \eR^+(x+\mu_2|\mu_1,\mu_2)& |m',n'\rangle|_{x=0} =\eH^{+}|m',n'\rangle=\left[\psi(n'+2s)-\psi (2s)\right] |m',n'\rangle\\& -\sum_{k=0}^{n'}\mathbbm{1}_{\{|k|>0\}}\frac{\Gamma (|k|) \Gamma (|n'|-|k|+2s)}{\Gamma (|n'|+2s)}\left[\prod_{a=1}^{\M}\binom{n_a'}{n_a'-k_a}\right]  |m'+k,n'-k\rangle
		 					 \,,
		 				\end{split}
		 			\end{equation}

		 			\begin{equation}\label{element-Nr}
		 				\begin{split}
		 					\frac{\partial}{\partial x}  \eR^-(x+\mu_1,\mu_2|\mu_1) &|m',n'\rangle |_{x=0} =\eH^{-}|m',n'\rangle=\left[\psi(m'+2s)-\psi (2s)\right] |m',n'\rangle\\& -\sum_{k=0}^{m'}\mathbbm{1}_{\{|k|>0\}}\frac{\Gamma (|k|) \Gamma (|m'|-|k|+2s)}{\Gamma (|m'|+2s)}\left[\prod_{a=1}^{\M}\binom{m_a'}{m_a'-k_a}\right]  |m'-k,n'+k\rangle
		 					 \,,
		 				\end{split}
		 			\end{equation}
which concludes the proof.
                    
\end{proof}

\subsection{Comparison with existing models.}

We relate the $\eR$-operators introduced in Proposition~\ref{Proposition-R-components} (and in Corollary \ref{Corollary-R-components-total}) as well as the Hamiltonian defined in Corollary~\ref{corollary-bulk-elements} to previously known results.

    \paragraph{R-matrix as a rational limit.}
   Considering the $q$-Hahn weights $\Phi_q(\gamma\mid\beta;\lambda,\mu)$ introduced in equation~(19) of~\cite{Kuniba:2016fpi} and the matrices, 
     $\eR^{+}(x+\mu_{2}\mid y+\mu_{1},y+\mu_{2})$ and 
    $\eR^{-}(x+\mu_{1},y+\mu_{2}\mid y+\mu_{1})$,  given in~\eqref{element-M} and~\eqref{element-N}, we have that
    \begin{equation}
    \begin{split}
        &\lim_{q\to 1}\Phi_q(n\mid n';q^{2s+x-y},q^{2s})\delta_{m+n,m'+n'}
        = \langle m,n|\eR^{+}(x+\mu_{2}\mid y+\mu_{1},y+\mu_{2})|m',n'\rangle, \\
        &\lim_{q\to 1}\Phi_q(m\mid m';q^{2s},q^{2s+y-x})\delta_{m+n,m'+n'}
        = \langle m,n|\eR^{-}(x+\mu_{1},y+\mu_{2}\mid y+\mu_{1})|m',n'\rangle.
    \end{split}
    \end{equation}
    Above we denoted $2s=\mu_{2}-\mu_{1}>0$. These expressions naturally provide a multispecies generalization of the $q\to 1$ limit of the $q$-Hahn weights introduced in equation~(8) of~\cite{povolotsky2013integrability}.
    Furthermore, consider the choice of representation parameters
    \begin{equation}\label{label-choice-compact}
        \mu_{1}^{[1]}=\frac{I}{2}+\frac{1}{2},\quad 
        \mu_{2}^{[1]}=-\frac{I}{2}+\frac{1}{2},\qquad
        \mu_{1}^{[2]}=\frac{J}{2}+\frac{1}{2},\quad 
        \mu_{2}^{[2]}=-\frac{J}{2}+\frac{1}{2}
    \end{equation}
    with $I,J\in \mathbb{N}$.
    Then the component form of the $\eR$-matrix in~\eqref{eq:fullR1} coincides, up to a normalization factor $B_{I,J}$, with the $q\to 1$ limit of the $R$-matrix derived in~\cite[(3.16)--(3.17)]{Kuniba:2016fpi,Bosnjak:2016oze}. In particular \footnote{In the $q\to 1$ limit, the normalization factor reduces to
$        \lim_{q\to 1}B_{I,J}(q^{x})
        = \frac{\left(x-\frac{I+J}{2}\right)_{I+J+1}}
        {\left(x-\frac{I+J}{2}\right)_{I+1}\left(x-\frac{I+J}{2}\right)_{J+1}}
    $. },
    \begin{equation}\label{eq:fin-R} 
        \lim_{q\to 1}\left[R_{I,J}(q^{x})\right]_{m,n}^{m',n'}=\langle m,n|\eR(x)|m',n'\rangle\,,
    \end{equation}
    where the total number of particles per site depends on $I$ and $J$ and is $|m|,|m'|\leq I$ and $|n|,|n'|\leq J$\footnote{For $|m|$ or $|m'|> I$ and $|n|$ or $|n'|\leq J$ the elements of the R-matrix are not well defined.}.

    We remark that the Drinfeld twist becomes trivial in the rational limit.  The factorisation of the stochastic R-matrix has been noted in \cite[(5.7)]{Bosnjak:2016oze}.     
    
    \begin{remark}\label{Remark-factorisationFail}
The factorisation  \eqref{Facrtorized-R} of the R-matrix relies on an infinite-dimensional auxiliary space. While the full R-matrix admits a consistent truncation to finite-dimensional representations, this property does not extend to the factors $\eR^{\pm}$. More precisely, the finite-dimensional R-matrix fails to factorise unless one allows for infinite-dimensional intermediate states, i.e. the factorisation maps out of the configuration space.
    
 To see this, let consider $I, J\in \mathbb{N}$, fix $\mu_{1}^{[1]}-\mu_{2}^{[1]}=J$, $\mu_{1}^{[2]}-\mu_{2}^{[2]}=I$, and let us restrict to a configuration space where the total number of particles per site is $|m|,|m'|\leq I$ and $|n|,|n'|\leq J$. Then, using Corollary \ref{Corollary-R-components-total}, one obtains that 
\begin{equation}
    \begin{split}
    \eR(x-y)=    &\sum_{m'',{n}''}\langle {n},{m}|\eR^{-}(x_{1},y_{2}|y_{1})|{m}'',{n}''\rangle\langle {m}'',{n}''|\eR^{+}(x_{2}|y_{1},y_{2})|{m}',{n}'\rangle \,.
                  \end{split}
\end{equation}

The summation index  does not satisfy the particle number constraint and there are non-zero terms in the summation corresponding, which are being neglected when restricting $\eR^{\pm}$ to finite dimension. \end{remark}

   We fix the fundamental representation on both spaces in \eqref{eq:fin-R} by imposing $\mu_1^{[i]}=1$ and $\mu_2^{[i]}=0$ for $i=1,2$. Let $(|A\rangle)_{A=0}^{\M}$ denote the canonical basis of $\mathbb{C}^{\M+1}$, and let $\eR_{\square,\square}(x)$ denote the resulting R-matrix. Then we obtain that
\begin{equation}\label{R-double-fund}
\langle A,B|\eR_{\square,\square}(x)|C,D\rangle
=\frac{1}{(x+1)}\,
[\fundR(x)]^{C,D}_{A,B}\,,
\qquad \forall\, A,B,C,D\in \{0,1,\ldots,\M\},
\end{equation}
where $\fundR(x)=x+P$ denotes the fundamental R-matrix defined in \eqref{fundamental-R}.
 \paragraph{Hamiltonian as a rational limit.}
    
    Using Corollary~\ref{corollary-bulk-elements}, for $|n'|>|n|$ and for $|m'|>|m|$, the jump rates of the Hamiltonian~\eqref{rates} can be expressed as 
    \begin{equation}\label{rate+}
        \varphi_{s}(n'-n,n')\delta_{m+n,m'+n'}
        = -\lim_{x\to 0}
        \frac{\langle m,n|\eR^{+}(x+\mu_{2}\mid \mu_{1},\mu_{2})|m',n'\rangle}{x},
    \end{equation}
    and
    \begin{equation}\label{rate-}
        \varphi_{s}(m'-m,m')\delta_{m+n,m'+n'}
        =- \lim_{x\to 0}
        \frac{\langle m,n|\eR^{-}(x+\mu_{1},\mu_{2}\mid \mu_{1})|m',n'\rangle}{x}.
    \end{equation}
    Above we have denoted $2s=\mu_{2}-\mu_{1}>0$.   This identifies the process as the $q\to 1$ limit of a multispecies generalization of the $q$-Hahn zero-range process in~\cite{sasamoto1998one,barraquand2016q} (see also Remark \ref{Remark-MqHahn}).
    On the diagonal, we obtain 
\begin{equation}\label{rate-diag+}
        \psi(2s+|n'|)-\psi(2s)
        = -\lim_{x\to 0}
        \frac{\langle m',n'|\eR^{+}(x+\mu_{2}\mid \mu_{1},\mu_{2})|m',n'\rangle}{x},
    \end{equation}
    and 
    \begin{equation}\label{rate-diag-}
        \psi(2s+|m'|)-\psi(2s)
        = -\lim_{x\to 0}
        \frac{\langle m',n'|\eR^{-}(x+\mu_{2}\mid \mu_{1},\mu_{2})|m',n'\rangle}{x}\,.
    \end{equation}
          
	When restricting to finite-dimensional representations, namely by fixing $2s\in -\mathbb{N}$ and assuming that the particle number is bounded by $|2s|$, the factorization of the $R$-matrix $\eR(x)$ no longer holds (see Remark \ref{Remark-factorisationFail}) and therefore the equations \eqref{rate-diag+}-\eqref{rate-diag-} are not defined in this setup. However, we observe that for the fundamental representation of $gl(\M+1)$, namely for $2s = -1$ and upon restricting the space to $\mathbb{C}^{\M+1} \otimes \mathbb{C}^{\M+1}$, one obtain the Hamiltonian density of the multispecies stirring process \cite{vanicat2017exact, casini2024duality}, given by
\begin{equation}\label{Mstirring-hamiltonian}
    h := \frac{\partial}{\partial x} \log\!\left(\fundR(x)\right)\bigg|_{x=0}\,.
\end{equation}
Above, $\fundR(x)$ is the fundamental $R$-matrix \eqref{fundamental-R}. In particular, upon further restriction to $M=1$, one recovers the Hamiltonian of the simple symmetric exclusion process (SSEP); see the review \cite{crampe2014integrable}.

We remark that this is the only finite-dimensional representation for which a stochastic Hamiltonian is obtained via the logarithmic derivative of $\eR(x)$. For $2s\in -\mathbb{N}$ with $|2s|>1$, the Hamiltonian matrix has negative off-diagonal elements, breaking Markovianity. This can be verified by direct computation, see also \cite{2016CMaPh.343..651C}.

 \subsection{Integral representations for the R-matrix}\label{section-integral-reps}
Starting from the abstract algebraic form of the R-matrix in Theorem \ref{Theorem-main}, we now realize the algebra using differential and multiplication operators and represent the action of $gl(\M+1)$ on the space of polynomials following the ideas of \cite{Derkachov:1999pz}.
 \subsubsection{Integral form}\label{sec:intrepp}

We introduce the representation $\rep:\mathcal{A}_M\to End(\mathbb{C}[\hidden_1,\ldots,\hidden_M])$ of the Heisenberg albegra where  $\mathbb{C}[\hidden_1,\ldots,\hidden_M]$ denotes the infinite-dimensional space
    of polynomials in variables $\hidden_1,\ldots,\hidden_M$. We have
    \begin{equation}\label{rep-infiniteDim-Poly}
\rep(\oa_{a})=\partial_{\hidden_{a}}\,,\qquad \rep(\oad_{a})=\hidden_{a}\,.
	\end{equation}
	such that
	\begin{equation}\label{rep-infiniteDim-Poly2}
\rep(\OA)=\bm{\partial}_{\hidden}\,,\qquad \rep(\OAD)=\bm{\Hidden}\,.
	\end{equation}
	where  $\bm{\partial}_{\hidden}=(\partial_{\hidden_1},\ldots,\partial_{\hidden_\M})^{t}$ and $\bm{\Hidden}=(\hidden_{1},\ldots,\hidden_{\M})$.

 The generators of $gl(M+1)$ then take the form of first-order differential operators
 \begin{equation}\label{GL-continuous-reps}
		\rep(J_{BA})=\left(\begin{array}{cc}
			\mmu_1-\OZD\OZ&\OZD\left((\mmu_1-\mmu_2+1)\ID-\OZ\OZD\right)\\
			\OZ&(\mmu_2-1)\ID+\OZ\OZD
		\end{array}\right) _{AB}
	\end{equation}
    with the  highest-weight vector $v_0 = 1$.  This infinite-dimensional representation is reducible for $\mu_1-\mu_2\in\mathbb{N}$. This can be see by inspecting the upper right block of~\eqref{GL-continuous-reps}, i.e. we have that
    \begin{equation}
        \bm{\Hidden}\left((\mu_{1}-\mu_{2}+1)\ID-\bm{\partial}_{\hidden}\bm{\Hidden}\right)=\bm{\Hidden}(\mu_{1}-\mu_{2}-\bm{\Hidden}\bm{\partial}_{\hidden})\,.
    \end{equation}
    As a consequence, for all $a\in \{1,\ldots, \M\}$ we obtain that
    \begin{equation}
        \hidden_{a}(\mu_{1}-\mu_{2}-\bm{\Hidden}\bm{\partial}_{\hidden})\prod_{b=1}^{\M}\hidden_{b}^{n_{b}}=\left(\mu_{1}-\mu_{2}-\sum_{b=1}^{\M}n_{b}\right)\hidden_{a}^{n_{a}+1}\prod_{b=1\,b\neq a}^{\M}\hidden_{b}^{n_{b}}\,.
    \end{equation}
    Therefore, for $\mu_{1}-\mu_{2}\in \mathbb{N}$ the representation admits a non-trivial invariant subspace and thus is reducible.

    An element  in the oscillator algebra $\mathcal{O}\in\mathcal{A}_M$ in  representation $\pi$ is related to its representation $\rep$ through the intertwining relation
     \begin{equation}\label{intertw}
 \rep(\mathcal{O} )g_{m}(z)=\sum_{m'=0}^\infty g_{m'}(z)\langle m'| \,\pi(\mathcal{O} ) \,|m\rangle  \,,
 \end{equation}
 where $g_{m}(z)$ is a polynomial of the form $g_{m}(z)=\prod_{a=1}^{\M}(\hidden_a)^{m_a}$, which can be thought of as the basis of representation $\rep $. 
    
As for the Fock representation, these definitions naturally lift to the R-operator such that
 \begin{equation}\label{Rp-continuous}
			\ecR ^+(x_2|y_1,y_2)f(\hidden^{1},\hidden^{2}) =\kappa^+(x_2|y_1,y_2)e^{-\bm{\Hidden}^{[1]}\bm{\partial}_{\hidden}^{[2]}}\frac{\Gamma(\bm{\Hidden}^{[2]}\bm{\partial}_{\hidden}^{[2]}-y_1+x_2)}{\Gamma(\bm{\Hidden}^{[2]}\bm{\partial}_{\hidden}^{[2]}-y_1+y_2)}e^{\bm{\Hidden}^{[1]}\bm{\partial}_{\hidden}^{[2]}}f(\hidden^{1},\hidden^{2})
		\end{equation}
        and
    \begin{equation}\label{Rm-continuous}
		\begin{split}
\ecR ^{-}(x_{1},x_{2}|y_{1})f(\hidden^{1},\hidden^{2}
)&=\kappa^{-}(x_{1},x_2|y_{1})e^{-\bm{\Hidden}^{[2]}\bm{\partial}_{\hidden}^{[1]}}\frac{\Gamma\left(\bm{\Hidden}^{[1]}\bm{\partial}_{\hidden}^{[1]}+x_{2}-y_{1}\right)}{\Gamma\left(\bm{\Hidden}_{\hidden}^{[1]}\bm{\partial}_{\hidden}^{[1]}+x_{2}-x_{1}\right)}e^{\bm{\Hidden}^{[2]}\bm{\partial}_{\hidden}^{[1]}}f(\hidden^{1},\hidden^{2})
		\end{split}
	\end{equation}
  	where $[1]$ and $[2]$ denote the two spaces on which the R-operator acts and $f(\hidden^1,\hidden^2)\in \mathbb{C}[\hidden^1_1,\ldots,\hidden^1_M]\otimes \mathbb{C}[\hidden^2_1,\ldots,\hidden^2_M]$ a polynomial function. In the following we will suppress $\rep$ whenever it is apparent that we are acting on functions.

	The integral representation of the operators $\abrOSCR^\pm$ is obtained manipulating~\eqref{Rp-continuous} and~\eqref{Rm-continuous}.
	\begin{proposition}\label{proposition-intergrtalReps}
		The  operators $\ecR ^{+}(x_{2}|y_{1},y_{2})$ and $\ecR ^{-}(x_{1},x_{2}|y_{1})$ have the following integral representation:
        		\begin{equation}\label{redR2sl2int}
			\begin{split}
				\ecR ^+(x_2|y_1,y_2) f(\hidden^{1},\hidden^{2})&=\frac{\kappa_+(x_{2}|y_{1},y_{2})}{ \Gamma(y_2-x_2)} \int_0^1d\alpha\, \alpha^{x_2-y_1-1}(1-\alpha)^{y_2-x_2-1}  f(\hidden^{1},(1-\alpha)\hidden^{1}+\alpha \hidden^{2})\,,
			\end{split}
		\end{equation}
		and
		\begin{equation}\label{redR2sl2int2}
			\begin{split}
				\ecR ^-(x_1,x_2|y_1) f(\hidden^{1},\hidden^{2})&=\frac{\kappa_-(x_{1},x_2|y_{1})}{ \Gamma(y_1-x_1)}  \int_0^1d\alpha\,\alpha^{x_2-y_1-1}(1-\alpha)^{y_1-x_1-1}  f((1-\alpha)\hidden^{2}+\alpha \hidden^{1},\hidden^{2})\,,
			\end{split}
		\end{equation}
																							where $\hidden^{1}=(z_{1}^{1},\ldots,z_{\M}^{1})$ and $\hidden^{1}=(z_{1}^{2},\ldots,z_{\M}^{2})$.	\end{proposition}

\begin{proof}
  		                 
The integral representation follows by acting with~\eqref{Rp-continuous} and~\eqref{Rm-continuous} on functions when using the relations
	\begin{equation}\label{shift-formulas}
		\begin{split}
			&e^{\alpha\partial_{\hidden}}f(\hidden)=f(\hidden+\alpha),\qquad \alpha^{\hidden\partial_{\hidden}}f(z)=f(\alpha \hidden)\,,
		\end{split}
	\end{equation}
	and from the integral representation of the Beta-function~\eqref{eq:intrp}
that is used to express the fraction of Gamma functions as an integral,
see \cite{Derkachov:1999pz}.
\end{proof}

\begin{remark}\label{InterRe}
As a consequence of the intertwining relation~\eqref{intertw},
 the R-operators~\eqref{redR2sl2int} and~\eqref{redR2sl2int2} in representation $\rep$  are related  to  the one in representation $\pi$ with elements written in~\eqref{element-M} and~\eqref{element-N} via
 \begin{equation}
  \rep (\abrOSCR ^\pm) g_{m,n}(\hidden^{1},\hidden^{2})=\sum_{m',n'} g_{m',n'}(\hidden^{1},\hidden^{2})\langle m'|\otimes \langle n'|\,\pi(\abrOSCR^\pm) \,|m\rangle \otimes |n\rangle\,.
 \end{equation}
where $g_{m,n}(\hidden^{1},\hidden^{2})=\prod_{a=1}^{\M}(\hidden_a^{1})^{m_a}\prod_{b=1}^{\M}(\hidden_{b}^{2})^{n_{b}}$.

\end{remark}

We now fix $\mu_{i}^{[1]}=\mu_{i}^{[2]}$ for $i=1,2$, that is, the same representation label for both spaces $V_i$ and we choose the normalization $\kappa_{\pm}$ for the matrices $\ecR^{\pm}$ as \eqref{normali}. Then, by means of the logarithmic derivative, we obtain the Hamiltonian density acting on $\mathbb{C}[\hidden^1]\otimes \mathbb{C}[\hidden^2]$.
\begin{corollary}\label{cor:int}
 For the Hamiltonian density we obtain the integral representation
 \begin{equation}\label{eq:inthamm}
  \ecH f(\hidden^{1},\hidden^{2})=\ecH ^+f(\hidden^{1},\hidden^{2})+\ecH ^-f(\hidden^{1},\hidden^{2})
 \end{equation}
with
	\begin{equation}\label{redR2sl2int24}
			\begin{split}
				\ecH ^-f(\hidden^{1},\hidden^{2})&= \int_0^1d\alpha\,\frac{\alpha ^{2s-1}}{1-\alpha}\left[f(\hidden^{1},\hidden^{2})-f(\alpha \hidden^{1}+(1-\alpha )\hidden^{2},\hidden^{2})\right]\,.
			\end{split}
		\end{equation}
        and
\begin{equation}\label{redR2sl2int4}
			\begin{split}
				\ecH ^+ f(\hidden^{1},\hidden^{2})&= \int_0^1d\alpha\,\frac{\alpha ^{2s-1}}{1-\alpha }\left[f(\hidden^{1},\hidden^{2})-f(\hidden^{1},\alpha \hidden^{2}+(1-\alpha)\hidden^{1})\right]\,,
			\end{split}
		\end{equation}
\end{corollary}
\begin{proof}
Fixing the representation~\eqref{rep-infiniteDim-Poly} the Hamiltonian density is given by the sum of 
\begin{equation}			\ecH ^-f(\hidden^{1},\hidden^{2})=e^{-\bm{\Hidden}^{[2]}\bm{\partial}_{\hidden}^{[1]}}\left[\psi(\bm{\Hidden}^{[1]}\bm{\partial}_{\hidden}^{[1]}+2s)-\psi(2s)\right]e^{\bm{\Hidden}^{[2]}\bm{\partial}_{\hidden}^{[1]}}f(\hidden^{1},\hidden^{2})\,,
		\end{equation}
and
			\begin{equation}					\ecH ^+f(\hidden^{1},\hidden^{2}) = e^{-\bm{\Hidden}^{[1]}\bm{\partial}_{\hidden}^{[2]}}\left[\psi(\bm{\Hidden}^{[2]}\bm{\partial}_{\hidden}^{[2]}+2s)-\psi(2s)\right]e^{\bm{\Hidden}^{[1]}\bm{\partial}_{\hidden}^{[2]}}f(\hidden^{1},\hidden^{2})\,.
		\end{equation}
Then, we use the integral representation of the Digamma function~\eqref{eq:intrp} and proceed in analogy to Proposition~\ref{proposition-intergrtalReps}.
		\end{proof}
    We observe that the Hamiltonian operator~\eqref{eq:inthamm} is related to the bulk generator of the hidden parameter model~\eqref{generator-bulk-hidden} by
    \begin{equation}
    \ecL_{\ell,\ell+1}f(\hidden^{\ell},\hidden^{\ell+1})=-\ecH _{\ell,\ell+1}f(\hidden^{\ell},\hidden^{\ell+1})\,,
    \end{equation} for all $\ell\in\{1,\ldots, \N-1\}$.

\subsubsection{Dual integral form}\label{Heat-RDensity}
 We construct a second integral representation from the R-operator on the space of polynomials as done in the previous section.
For this purpose, we introduce the representation $\dualrep:\mathcal{A}_M\to End(\mathbb{C}[\heat_1,\ldots,\heat_M])$ with
\begin{equation}\label{rep-infiniteDim-Poly_dual}
\dualrep(\oa_{a})=\heat_a\,,\qquad \dualrep(\oad_{a})=-\partial_{\heat_{a}}\,.
	\end{equation}
	such that
	\begin{equation}\label{rep-infiniteDim-Poly2_dual}
\dualrep(\OA)=\bm{\heat}\,,\qquad \dualrep(\OAD)=-\bm{\partial}_\heat\,.
	\end{equation}
	where  $\bm{\partial}_\heat=(\partial_{\heat_1},\ldots,\partial_{\heat_\M})$ and $\bm{\Heat}=(\heat_{1},\ldots,\heat_{\M})^{t}$.

 The generators of $gl(M+1)$ are first-order differential operators in this representation, as
 \begin{equation}\label{bargen}
		\dualrep(J_{BA})=\left(\begin{array}{cc}
			\mmu_1+\OT\OTD&-\OT\left((\mmu_1-\mmu_2+1)\ID+\OTD\OT\right)\\
			\OTD&(\mmu_2-1)\ID-\OTD\OT
		\end{array}\right) _{AB}
	\end{equation}
    with lowest-weight vector $\bar v_0 = 1$. The Cartan eigenvalues on the lowest weight state are
    \begin{equation}
   \dualrep (J_{AA})\cdot 1=\begin{cases*}
                    \lambda_1\cdot 1 & if  $A =0$  \\
                     \lambda_2\cdot 1 & if $A=1,\ldots,M$
                 \end{cases*}\,,
    \end{equation}
with $\lambda_1=\mu_1+M$ and $\lambda_2=\mu_2-1$. This infinite dimensional representation is reducible for $\lambda_{2}-\lambda_{1}=\mu_2-\mu_1-M-1\in\mathbb{N}$. This can be seen by inspecting the upper right block of~\eqref{bargen}.
By using the commutation relations we have that
\begin{equation}
    \begin{split}
        -\OT\left((\mmu_1-\mmu_2+1)\ID+\OTD\OT\right)=\left(\lambda_{2}-\lambda_{1}-\OTD^{t}\OT^{t}\right)\OT\,,
    \end{split}
\end{equation}
after normal ordering.
As a consequence, for all $a\in \{1,\ldots,\M\}$, we obtain that
\begin{equation}
    \begin{split}
       (\lambda_{2}-\lambda_{1}-\OTD^{t}\OT^{t}) \partial_{\heat_{a}}\prod_{b=1}^{\M}\heat_{b}^{n_{b}}=n_{a}\left(\lambda_{2}-\lambda_{1}+1-\sum_{b=1}^{\M}n_{b}\right)\heat_{a}^{n_{a}-1}\prod_{b=1\,b\neq a}^{\M}\heat_{b}^{n_{b}}\,.
    \end{split}
\end{equation}
Above we denoted $\bm{\partial}_\heat^{t}=(\partial_{\heat_1},\ldots,\partial_{\heat_\M})^{t}$.
Therefore, for $\lambda_{2}-\lambda_{1}\in \mathbb{N}$, the representation admits a non-trivial invariant sub-space and then it is reducible.

 An element  in the oscillator algebra $\mathcal{O}\in\mathcal{A}_M$ in the representation $\pi$ is related to $\dualrep $ through the intertwining relation
   \begin{equation}\label{intertw-dual}
\dualrep (\mathcal{O} )g_{m}(\heat)=\sum_{m'=0}^\infty g_{m'}(\heat)\langle m'| \,Q^{-1}\pi(\mathcal{O} )^{t}Q \,|m\rangle  \,,
 \end{equation}
 where the upper-case $t$ stands for transposition and where 
 \begin{equation}\label{Q-matrix}
 Q:=\prod_{a=1}^{\M}\Gamma(\pi(\oad_{a})\pi(\oa_{a})+1)\,.
 \end{equation}
 Above $g_{m}(\heat)$ is a polynomial of the form $g_{m}(\heat)=\prod_{a=1}^{\M}(\heat_a)^{m_a}$, which can be thought of as the basis of representation $\dualrep $. 
\begin{proposition}
     The operators $\ecdR ^\pm=\dualrep (\abrOSCR^\pm)$ can be written as
\begin{equation}\label{Rp-continuous2}
			\ecdR ^+(x_2|y_1,y_2)f(\heat^{1},\heat^{2}) =\bar\kappa_+(x_2|y_1,y_2)e^{\OT^{[1]}\OTD^{[2]}}\frac{\Gamma(\OT^{[2]}\OTD^{[2]}+y_1-y_2+1)}{\Gamma(\OT^{[2]}\OTD^{[2]}+y_1-x_2+1)}e^{-\OT^{[1]}\OTD^{[2]}}f(\heat^{1},\heat^{2})
		\end{equation}
        and
    \begin{equation}\label{Rm-continuous2}
		\begin{split}
\ecdR ^{-}(x_{1},x_{2}|y_{1})f(\heat^{1},\heat^{2}
)&=\bar \kappa_-(x_{1},x_2|y_{1})e^{\OT^{[2]}\OTD^{[1]}}\frac{\Gamma\left(\OT^{[1]}\OTD^{[1]}-x_{2}+x_{1}+1\right)}{\Gamma\left(\OT^{[1]}\OTD^{[1]}-x_{2}+y_{1}+1\right)}e^{-\OT^{[2]}\OTD^{[1]}}f(\heat^{1},\heat^{2})
		\end{split}
	\end{equation}
where  $x_i=x-\mu_i^{[1]}$, $y_i=y-\mu_i^{[2]}$ and    $\bar \kappa_\pm$ denotes a normalisation.
\end{proposition}
\begin{proof}
 Applying the algebra automorphism $\dualrep $ to equations~\eqref{Rp-solution} and~\eqref{Rm-solution} for $\abrOSCR^{+}$ and $\abrOSCR^{-}$, yields the expressions
    \begin{equation}\label{Rp-continuous22}
			\ecdR ^+(x_2|y_1,y_2)f(\heat^{1},\heat^{2}) =\bar\kappa_+(x_2|y_1,y_2)e^{\OT^{[1]}\OTD^{[2]}}\frac{\Gamma(-\OT^{[2]}\OTD^{[2]}-y_1+x_2)}{\Gamma(-\OT^{[2]}\OTD^{[2]}-y_1+y_2)}e^{-\OT^{[1]}\OTD^{[2]}}f(\heat^{1},\heat^{2})
		\end{equation}
        and
    \begin{equation}\label{Rm-continuous22}
		\begin{split}
\ecdR ^{-}(x_{1},x_{2}|y_{1})f(\heat^{1},\heat^{2}
)&=\bar\kappa_-(x_{1},x_2|y_{1})e^{\OT^{[2]}\OTD^{[1]}}\frac{\Gamma\left(-\OT^{[1]}\OTD^{[1]}+x_{2}-y_{1}\right)}{\Gamma\left(-\OT^{[1]}\OTD^{[1]}+x_{2}-x_{1}\right)}e^{-\OT^{[2]}\OTD^{[1]}}f(\heat^{1},\heat^{2})
		\end{split}
	\end{equation}
Given that we are acting on the space of polynomials $\mathbb{C}[\heat_1,\ldots,\heat_M]\otimes \mathbb{C}[\heat_1,\ldots,\heat_M]$ it follows that the operators $\OT^{[1]}\OTD^{[1]}$ and $\OT^{[2]}\OTD^{[2]}$ have integer spectrum when applied to a basis vector as $\OT\OTD \heat_a^n=(M+n)\heat_a^n$ for any $a$. As a consequence we can use the reflection formula
	\begin{equation}
	 \Gamma(z)\Gamma(1-z)=\frac{\pi}{\sin(\pi z)}
	\end{equation}
to obtain the operators in~\eqref{Rp-continuous2} and~\eqref{Rm-continuous2}.
\end{proof}
We now fix the normalization constants as
 \begin{equation}\label{normali2}
		 				\bar\kappa_+(x_2|y_1,y_2)=\frac{\Gamma(M+y_1-x_2+1)}{\Gamma(M+y_1-y_2+1)}\,,\qquad
		 				\bar\kappa_-(x_1,x_2|y_1)=\frac{\Gamma(M-x_2+y_1+1)}{\Gamma(M-x_2+x_1+1)}\,.
		 			\end{equation}
\begin{proposition}
    The operators $\ecdR^{+}(x_2|y_1,y_2)$ and $\ecdR^{-}(x_1,x_2|y_1)$ have the following integral representation
   		\begin{equation}\label{redR2sl2int-dual}
			\begin{split}
				\ecdR ^+(x_2|y_1,y_2) f(\heat^{1},\heat^{2})&=                \frac{k_{+}(x_{2}|y_{1},y_{2})}{\Gamma(y_{2}-x_{2})}\int_{0}^{1}\alpha^{y_{1}-y_{2}+\M}(1-\alpha)^{y_{2}-x_{2}-1}f(\heat^1+(1-\alpha)\heat^{2},\alpha \heat^{2})d\alpha\,,
			\end{split}
		\end{equation}
		and
		\begin{equation}\label{redR2sl2int2-dual}
			\begin{split}
				\ecdR ^-(x_1,x_2|y_1) f(\heat^{1},\heat^{2})&=                \frac{k_{-}(x_{1},x_{2}|y_{1})}{\Gamma(y_{1}-x_{1})}\int_{0}^{1}\alpha^{x_{1}-x_{2}+\M}(1-\alpha)^{y_{1}-x_{1}-1}f(\alpha \heat^{1},\heat^{2}+(1-\alpha)\heat^{1})d\alpha\,.
			\end{split}
		\end{equation}
\end{proposition}
\begin{proof} We only prove in detail the $\ecdR^{+}$ case, since the other is analogous. By using~\eqref{shift-formulas} we write that 
    \begin{equation}
        e^{-\bm{\Heat}^{[2]}\bm{\partial}^{[1]}_{\heat}}f(\heat^{1},\heat^{2})=f(\heat^{1}-\heat^{2},\heat^{2})\,.
    \end{equation}
    We notice that
    \begin{equation}
    \frac{\Gamma(1+\bm{\partial}_{\heat}^{[2]}\bm{\Heat}^{[2]}+y_{1}-y_{2})}{\Gamma(1+\bm{\partial}_{\heat}^{[2]}\bm{\Heat}^{[2]}+y_{1}-x_{2})}=\frac{\Gamma(1+\M+(\bm{\Heat}^{[2]})^{t}(\bm{\partial}_{\heat}^{[2]})^{t}+y_{1}-y_{2})}{\Gamma(1+\M+(\bm{\Heat}^{[2]})^{t}(\bm{\partial}_{\heat}^{[2]})^{t}+y_{1}-x_{2})}\,.
    \end{equation}
    Then, by using the integral representation of the Beta function we write 
    \begin{equation}
        \begin{split}
         &\frac{\Gamma(1+\M+(\bm{\Heat}^{[2]})^{t}(\bm{\partial}_{\heat}^{[2]})^{t}+y_{1}-y_{2})}{\Gamma(1+\M+(\bm{\Heat}^{[2]})^{t}(\bm{\partial}_{\heat}^{[2]})^{t}+y_{1}-x_{2})}f(\heat^{1}-\heat^{2},\heat^{2})
         \\=&
         \frac{1}{\Gamma(y_{2}-x_{2})}\int_{0}^{1}\alpha^{\M+y_{1}-y_{2}+\sum_{a=1}^{\M}\heat_{a}^{[2]}\partial_{a}^{[2]}}(1-\alpha)^{y_{2}-x_{2}-1}f(\heat^{1}-\heat^{2},\heat^{2})d\alpha
         \\=&
          \frac{1}{\Gamma(y_{2}-x_{2})}\int_{0}^{1}\alpha^{\M+y_{1}-y_{2}}(1-\alpha)^{y_{2}-x_{2}-1}f(\heat^{1}-\alpha \heat^{2},\alpha \heat^{2})d\alpha
        \end{split}
    \end{equation}
    where in the last equality we have applied formulas~\eqref{shift-formulas}. Finally, we compute 
    \begin{equation}
        e^{\bm{\Heat}^{[1]}\bm{\partial}_{\heat}^{[2]}}f(\heat^{1}-\alpha \heat^{2},\alpha \heat^{2})=f(\heat^{1}+(1-\alpha)\heat^{2},\alpha \heat^{2})\,,
    \end{equation}
    concluding the proof. 
\end{proof}

\begin{remark}\label{InterRe-dual}
As a consequence of the intertwining relation~\eqref{intertw-dual},
 the R-operators~\eqref{redR2sl2int-dual} and~\eqref{redR2sl2int2-dual} in representation $\dualrep$  are related  to  the one in representation $\pi$ with elements written in~\eqref{element-M} and~\eqref{element-N} via
 \begin{equation}
  \dualrep (\abrOSCR ^\pm)g_{m,n}(\heat^{1},\heat^{2})=\sum_{m',n'} g_{m',n'}(\hidden^{1},\hidden^{2})\langle m'|\otimes \langle n'|\left(Q^{-1}\otimes Q^{-1}\right)\,\pi(\abrOSCR^\pm)^{t}\left(Q\otimes Q\right) \,|m\rangle \otimes |n\rangle\,,
 \end{equation}
where $g_{m,n}(\heat^{1},\heat^{2})=\prod_{a=1}^{\M}(\heat_a^{1})^{m_a}\prod_{b=1}^{\M}(\heat_{b}^{2})^{n_{b}}$.

\end{remark}

We now fix equal representation labels, namely $\lambda_{i}^{[1]}=\lambda_{i}^{[2]}$ sucht that  $\mu_{i}^{[1]}=\mu_{i}^{[2]}$  for $i=1,2$ and we take $\lambda_{1}-\lambda_{2}=\mu_1-\mu_2+M+1=2s>0$.
\begin{corollary}
    For the Hamiltonian density we have that 
    \begin{equation}
        \ecdH f(\heat^{1},\heat^{2})=\ecdH^{-}f(\heat^{1},\heat^{2})+\ecdH^{+}f(\heat^{1},\heat^{2})
    \end{equation}
    where
     \begin{equation}
        \ecdH^{-}f(\heat^{1},\heat^{2})=\int_{0}^{1}\frac{\alpha^{2s-1}}{(1-\alpha)}\left[f(\heat^{1},\heat^{2})-f(\alpha\heat^{1}, \heat^{2}+(1-\alpha)\heat^{1})\right]d\alpha
    \end{equation}
    and
    \begin{equation}
        \ecdH^{+}f(\heat^{1},\heat^{2})=\int_{0}^{1}\frac{\alpha^{2s-1}}{(1-\alpha)}\left[f(\heat^{1},\heat^{2})-f(\heat^{1}+(1-\alpha)\heat^{2},\alpha \heat^{2})\right]d\alpha\,.
    \end{equation}
\end{corollary}
\begin{proof}We only prove only the integral relation for $\ecdH^{+}$, since the one for $\ecdH^{-}$ is analogous. We compute the derivative in $x$ of $\ecdR^{+}(x+\mu_{1}|\mu_{1},\mu_{2})$ and we evaluate it at $x=0$. 
    \begin{equation}
        \begin{split}
            \ecdH^{+}f(\heat^{1},\heat^{2})=&\frac{\partial}{\partial x}\ecdR^{+}(x+\mu_{1}|\mu_{1},\mu_{2})\bigg|_{x=0}f(\heat^{1},\heat^{2})
                                          \\=&
           e^{\OT^{[1]}\OTD^{[2]}}\left[\psi((\OTD^{[2]})^t(\OT^{[2]})^t+\lambda_{1}-\lambda_{2})-\psi(\lambda_{1}-\lambda_{2})\right]e^{-\OT^{[1]}\OTD^{[2]}}f(\heat^{1},\heat^{2})
                \\=
            &e^{\OT^{[1]}\OTD^{[2]}}\left[\psi((\OTD^{[2]})^t(\OT^{[2]})^t+2s)-\psi(2s)\right]e^{-\OT^{[1]}\OTD^{[2]}}f(\heat^{1},\heat^{2})\,.
        \end{split}
    \end{equation}
    Here, we used $\mu_{1}=\lambda_{1}-\M$ and $\mu_{2}=\lambda_{2}+1$.
By using~\eqref{eq:intrp}, we compute 
\begin{equation}
\begin{split}
    &\left[\psi((\OTD^{[2]})^t(\OT^{[2]})^t+2s)-\psi(2s)\right]f(\heat^{1}-\heat^{2},\heat^{2})
    \\=&\int_{0}^{1}\frac{\alpha^{2s-1}}{(1-\alpha)}(1-\alpha^{\sum_{a=1}\heat_{a}^{[2]}\partial_{a}^{[2]}})f(\heat^{1}-\heat^{2},\heat^{2})d\alpha
    \\=&\int_{0}^{1}\frac{\alpha^{2s-1}}{(1-\alpha)}\left[f(\heat^{1}-\heat^{2},\heat^{2})-f(\heat^{1}-\alpha\heat^{2},\alpha \heat^{2})\right]d\alpha\,,
    \end{split}
\end{equation}
concluding the proof. 
\end{proof}
\section{The open spin chain}\label{sec:openchain}
In this section, we show that the stochastic Hamiltonian introduced in~\eqref{Hamiltonian-MHarmonic} admits a commuting family of operators that can be constructed within the framework of the quantum inverse scattering method using double-row transfer matrices~\cite{Sklyanin:1988yz}. For the extension to higher-rank cases we refer the reader to e.g.~\cite{Belliard:2010nhl}. In particular, the construction of transfer matrices with triangular boundary conditions in the $q$-deformed setting has been carried out recently in~\cite{2024JPhA...57x5201K}; see also~\cite{2019LMaPh.109.2049K}.

We proceed as follows. In Section~\ref{Section-BYBE}, we derive certain diagonal $K$-matrices by solving directly the boundary-Yang-Baxter equation (BYBE). By adapting the method introduced in~\cite{Frassek:2019vjt} to $gl(\M+1)$ we rederive a solution obtained by Tsuboi in \cite{Tsuboi:2018gfd}. The stochastic $K$-matrix that underlies the harmonic process is then obtained by means of a suitable similarity transformation. Together with the R-matrix obtained in Theorem~\ref{Theorem-main}, this allows to construct the transfer matrix for non-compact representations relevant here. In Section~\ref{Section-transfer-matrix}, we recover the Hamiltonian~$\eH$ defined in~\eqref{Hamiltonian-MHarmonic} for the multispecies harmonic process including the stochastic boundary terms from the logarithmic derivative of this transfer matrix, thereby establishing the integrability of the corresponding non-equilibrium Markov chain.
 
Throughout this section, we work in the Fock representation $\pi$ with highest-weight Dynkin labels given in~\eqref{HW-discreteReps}, and we fix
\begin{equation}\label{repsLabelFixing}
\mu_1 = -s + \tfrac{1}{2}, 
\qquad  
\mu_2 = s + \tfrac{1}{2}\,,
\end{equation}
cf.~\eqref{label-choice-compact}.
Equivalently, in terms of the parametrization~\eqref{CandS}, this corresponds to
\begin{equation}
    c = \frac{\mu_{1} + \mu_{2}}{2} = \frac{1}{2}\,,
    \qquad 
    s = \frac{\mu_{2} - \mu_{1}}{2} > 0\,.
\end{equation}
 This choice is advantageous for the analysis of the BYBE, as Proposition~\ref{prop:Lbar} implies that $L(x)L(-x) \propto \ID$ and  $L^t(x)L^t(-x-(M+1)) \propto \ID$, see Proposition \ref{crossL}.

\subsection{Boundary Yang-Baxter equations}\label{Section-BYBE}
Consider the fundamental R-matrix $\fundR(x)$ written in~\eqref{fundamental-R}. Let $q_{1},q_{2}\in \mathbb{R}$ with $q_{2}\neq 0$, then by a direct computation one verifies that the diagonal K-matrix 
\begin{equation}\label{K-fund}
    \Khatfund(x)=(q_{1}+xq_{2})e_{00}+(q_{1}-xq_{2})\sum_{a=1}^{\M}e_{aa}
\end{equation}  
solves the fundamental BYBE (see \cite{Sklyanin:1988yz,crampe2014integrable}) over $\mathbb{C}^{\M+1}\otimes\mathbb{C}^{\M+1}$
\begin{equation}\label{fundamental-BYBE}
    \fundR_{1,2}(x-y)K_{1}^{0}(x)\fundR_{1,2}(x+y)K_{2}^{0}(y)=K_{2}^{0}(y)\fundR_{1,2}(x+y)K_{1}^{0}(x)\fundR_{1,2}(x-y)\,.
\end{equation}
Above $K_{i}^{0}(x)$, with $i=1,2$, denotes the K-matrix $\Khatfund(x)$ acting non-trivially only on space $i$.
\begin{proposition}\label{Proposition-right-Kmatrix}
Let $q_{1},q_{2}\in \mathbb{R}$ with $q_{2}\neq 0$ be two parameters. The K-matrix  
\begin{equation}\label{K-hatShifted}
 \Khat(x)=\frac{\Gamma\left(-x+s+\frac{1}{2}-\frac{q_{1}}{q_{2}}\right)}{\Gamma\left(x+s+\frac{1}{2}-\frac{q_{1}}{q_{2}}\right)}\frac{\Gamma\left(x+\OAD\OA+s+\frac{1}{2}-\frac{q_{1}}{q_{2}}\right)}{\Gamma\left(-x+\OAD\OA+s+\frac{1}{2}-\frac{q_{1}}{q_{2}}\right)}
\end{equation}
 		solves the boundary Yang-Baxter equation
\begin{equation}\label{shifted-BYBE}
L(x-y)\Khat(x){L}(x+y)\Khatfund(y)=\Khatfund(y)L(x+y)\Khat(x){L}(x-y)\,.
\end{equation}

\end{proposition}
\begin{proof}
   Consider the K-matrix $\Khatfund(y)$ defined in~\eqref{K-fund} and substitute it into the BYBE~\eqref{shifted-BYBE}. This yields a matrix equation in the canonical basis of $\mathbb{C}^{\M+1}$, which decomposes into $(\M+1) \times (\M+1)$ component-wise relations. By matching the powers of the spectral parameter $y$, we obtain the following set of equations:
\begin{equation}\label{eq1-K}
    \begin{split}
        &\left[\Khat(x),\oad_{a}\oa_{b}\right]=0\qquad \forall \quad a,b\in\{1,\ldots,\M\}\,,
         \end{split}
        \end{equation} 
        \begin{equation}\label{eq2-K}
        \begin{split}
        &        \left((x+\mu_{1}+\OAD \OA)\Khat(x)h(\OAD,\OA)\right.
          +\left.h(\OAD,\OA)\Khat(x)((x+\mu_{2}-1)\ID+\OA\OAD)\right)
        =\frac{q_{1}}{q_2}\left[h(\OAD,\OA),\Khat(x)\right]
        \end{split}
        \end{equation}
        and 
        \begin{equation}\label{eq3-K}
        \begin{split}
        &\left(\OA \Khat(x)(x+\mu_{1}-\OAD\OA)+((x+\mu_{2}-1)\ID+\OA\OAD)\Khat(x)\OA \right)
        =\frac{q_{1}}{q_2}\left[\Khat(x),\OA\right]\,.
    \end{split}
\end{equation}
Above, for the sake of notation, we have defined $h(\OAD,\OA):= \OAD((\mu_{1}-\mu_{2}+1)\ID-\OA\OAD)$. 
Introducing the diagonal matrix
\begin{equation}
    Q:=q_{2}e_{00}-q_{2}\sum_{a=1}^{\M}e_{aa}\,,      \end{equation} equations~\eqref{eq2-K} and~\eqref{eq3-K} may be rewritten as a single equation
\begin{equation}
    \begin{split}
    2q_{1}\left[L(x),\Khat(x)\right]=\left[Q,L(x)\Khat(x)L(x)\right]\,.
    \end{split}
\end{equation}
 Multiplying both sides of the above equation on the left and on the right by $L(-x)$, and using Proposition~\ref{prop:Lbar} specialized at $c=\tfrac{1}{2}$, which implies that $L(x)L(-x)$ is proportional to the identity matrix, we obtain
\begin{equation}\label{eqK-final}
    \begin{split}
   \Khat(x)\left(2q_{1}L(-x)+L(x)QL(-x)\right)=\left(2q_{1}L(-x)+L(x)QL(-x)\right)\Khat(x)\,.
    \end{split}
\end{equation}
From equation~\eqref{eq1-K} we obtain that $\Khat(x)$ is a function of the spectral parameter $x$ and of the total number operator only, namely we write the ansatz
\begin{equation}
    \Khat(x)=f(x,\OAD\OA)\,,
\end{equation}
where $f(x,\OAD\OA)$ is a function to determine. Plugging this into equation~\eqref{eqK-final}, we obtain the recursion relation
\begin{equation}
    \frac{f(x,\OAD\OA+1)}{f(x,\OAD\OA)}=\frac{x+\OAD\OA+s+\frac{1}{2}-\frac{q_{1}}{q_{2}}}{-x+\OAD\OA+s+\frac{1}{2}-\frac{q_{1}}{q_{2}}}\,.
\end{equation}
This recursion relation is solved by
    \begin{equation}
    f(x,\OAD\OA)=C(x,\mu,q_{1},q_{2})\frac{\Gamma\left(x+\OAD\OA+s+\frac{1}{2}-\frac{q_{1}}{q_{2}}\right)}{\Gamma\left(-x+\OAD\OA+s+\frac{1}{2}-\frac{q_{1}}{q_{2}}\right)}
\end{equation}
where $C(x,\mu,q_{1},q_{2})$ is a normalization that drops in the Yang-Baxter equation. Choosing
\begin{equation}
  C(x,\mu,q_{1},q_{2})=  \frac{\Gamma\left(-x+s+\frac{1}{2}-\frac{q_{1}}{q_{2}}\right)}{\Gamma\left(x+s+\frac{1}{2}-\frac{q_{1}}{q_{2}}\right)}\,,
\end{equation}
one obtains~\eqref{K-hatShifted}.
\end{proof}
\begin{remark}
The normalization chosen for $\Khat$ in~\eqref{K-hatShifted} is motivated by the requirement of recovering the stochastic Hamiltonian from the transfer matrix, as will be made precise in Theorem~\ref{Thm-transfer-Ham}. We stress, however, that different normalizations leave the solution of the BYBE unaffected, as they cancel out identically within the equation.
\end{remark}

\begin{corollary}\label{Corollary-left-Kmatrix}
The K-matrix
 \begin{equation}\label{K0-def2}
\K(x)=\frac{\Gamma\left(x+\frac{\M}{2}+1
+s-\frac{q'_{1}}{q'_{2}}\right)}{\Gamma\left(-x-\frac{\M}{2}
+s-\frac{q'_{1}}{q'_{2}}\right)}\frac{\Gamma\left(-x-\frac{\M}{2}
+s-\frac{q'_{1}}{q'_{2}}+\OAD\OA\right)}{\Gamma\left(x+\frac{\M}{2}+1
+s-\frac{q'_{1}}{q'_{2}}+\OAD\OA\right)}\,.
 \end{equation}
 		solve the (dual-)boundary Yang-Baxter equation
\begin{equation}\label{D-BYBE}
       L(y-x)\K(x)L(-x-y-(\M+1))\Kfund(y)= \Kfund(y)L(-x-y-(\M+1))\K(x)L(y-x)\,,
    \end{equation}
    with
\begin{equation}\label{eq:Kfund2}
      \Kfund(x)=\left(q'_{1}-q'_{2}\left(x+\frac{\M+1}{2}\right)\right)e_{0,0}+\left(q'_{1}+q'_{2}\left(x+\frac{\M+1}{2}\right)\right)\sum_{a=1}^{\M} e_{aa}\,,
 \end{equation}
 where the shift in the Lax matrix is fixed by Proposition~\eqref{crossL}. 
  		\end{corollary}
\begin{proof} By inverting the boundary Yang-Baxter equation~\eqref{shifted-BYBE} and exploiting Proposition~\ref{prop:Lbar} in the case $c=\tfrac{1}{2}$, we find that the $K$-matrices that solve~\eqref{D-BYBE} are, up to an overall normalization, related to those given in~\eqref{K-fund} and~\eqref{K-hatShifted} through
     \begin{equation}
         \Kfund(x)=\left.\Khatfund\left(-x-\frac{\M+1}{2}\right)\right|_{q_{1}\to q'_{1},\,q_{2}\to q'_{2}}\,,\qquad    \K(x)=\left.\Khat\left(-x-\frac{\M+1}{2}\right)\right|_{q_{1}\to q'_{1},\,q_{2}\to q'_{2}}\,.
    \end{equation}
    This yields the K-matrix presented in~\eqref{K0-def2}.
\end{proof}

\begin{lemma}\label{lemmadd}
 The Lax matrix is invariant under the transformations
 \begin{equation}
 [L(x), d_\alpha\otimes D_\alpha]=0
 \end{equation}
 where
 \begin{equation}\label{dalpha}
  D_\alpha=\exp{\left(-\sum_{a=1}^{\M}J_{0,a}\right)}\exp{\left(-\sum_{a=1}^{\M}\alpha_{a}J_{a,0}\right)}
 \end{equation}
 and
 \begin{equation}\label{similarity-fundamental}
  d_\alpha=\exp{\left(-\sum_{ a=1}^{\M}e_{0,a}\right)}\exp{\left(-\sum_{a=1}^{\M}\alpha_{a}e_{a,0}\right)}
 \end{equation}
 where $\alpha_{a}\in\mathbb{R}$. Above $J_{A,B}$ stands for $\pi(J_{A,B})$, i.e. the algebra generator of $gl(\M+1)$ in the Fock representation, c.f.~\eqref{E-matrix}. 
\end{lemma}

\begin{proof}
When expanding  the RLL relation
\begin{equation}
		\fundR(x)(L(x+y)\otimes1)(1\otimes L(y))=(1\otimes L(y))(L(x+y)\otimes1)\fundR(x)
	\end{equation}
in the spectral parameter $x$ we obtain at linear order that for all $A,B\in \{0,1,\ldots,\M\}$
\begin{equation}
	\left[e_{AB}+J_{AB},L(y)\right]=0\,.
\end{equation}
Thus the Lax matrix commutes with any function of $e_{AB}+J_{AB}$. Writing
 \begin{equation}
  d_\alpha\otimes D_\alpha=\exp{\left(-\sum_{ a=1}^{\M}(e_{0,a}+J_{0,a})\right)}\exp{\left(-\sum_{a=1}^{\M}\alpha_{a}(e_{a,0}+J_{a,0})\right)}
 \end{equation}
 we see that the lemma above holds.
\end{proof}

\begin{corollary}\label{corollary-Full-K}
As a consequence of Lemma~\ref{lemmadd} it follows that the conjugated K-matrices
\begin{equation}\label{Kfulls}
   \KhatFull(x)=D_{\rho_{r}}\Khat(x)D_{\rho_{r}}^{-1}\,,\quad \Kfull(x)=D_{\rho_{l}}\K(x)D_{\rho_{l}}^{-1}
\end{equation}
\begin{equation}\label{Kfull-Fund}
  \KhatFundFull(x)=d_{\rho_{r}}\Khatfund(x)d_{\rho_{r}}^{-1}\,\quad \KFundFull(x)=d_{\rho_{l}} \Kfund(x)d_{\rho_{l}}^{-1}\,,\qquad 
\end{equation}
obey the same boundary Yang-Baxter equations as their diagonal counterparts, namely they solve
\begin{equation}
    L(x-y)\KhatFull(x){L}(x+y)\KhatFundFull(y)=\KhatFundFull(y)L(x+y)\KhatFull(x){L}(x-y)
\end{equation}
and 
\begin{equation}
    L(y-x)\Kfull(x)L(-x-y-(\M+1))\KFundFull(y)= \KFundFull(y)L(-x-y-(\M+1))\Kfull(x)L(y-x)\,.
\end{equation}
In particular the parameters $\rho_l$ and $\rho_{r}$ can be chosen arbitrarily.
\end{corollary}
\begin{proof}
 The proof is obtained by conjugating the boundary Yang-Baxter equations~\eqref{shifted-BYBE} and~\eqref{D-BYBE} and using Lemma~\ref{lemmadd}.
\end{proof}
\begin{remark}
 The K-matrix $\KhatFull(x)$ in~\eqref{Kfulls} for $\rho_r=0$ becomes triangular. Moreover, when choosing $\frac{q_{1}}{q_{2}}=\frac{1}{2}-J$, it coincides with the rational limit of the K-matrix found in \cite[(5.7)]{2024JPhA...57x5201K}.
\end{remark}
\begin{remark}
For $2s=-1$ and restricting to the vector space spanned by the canonical basis of $\mathbb{C}^{\M+1}$ denoted by $\{|A\rangle\}_{A=0,1,\ldots,\M}$, the components of the diagonal K-matrices $\Khat(x)$ and $\K(x)$ defined in \eqref{K-hatShifted} and \eqref{K0-def2} match, up to an overall scalar factor, with those of \eqref{K-fund} and \eqref{eq:Kfund2}, respectively. 
{ 
To recover the non-diagonal $K$-matrices associated with the multispecies stirring process \cite{vanicat2017exact,casini2024duality}, we consider \eqref{Kfulls} on the restricted subspace with  $\rho_{a}=-\beta_{a}$ for all $a\in \{1,\ldots, \M\}$.
} 
\end{remark}

\subsection{Transfer matrix and Hamiltonian}\label{Section-transfer-matrix}

In this section we derive the Hamiltonian of the boundary-driven multispecies harmonic process within the framework of the quantum inverse scattering method for integrable open spin chains \cite{Sklyanin:1988yz}.

To this end, we introduce the double-row transfer matrix $T(x)$, defined through an infinite-dimensional representation with labels $(\mu_{1},\mu_{2},\ldots,\mu_{2})$ in the quantum and auxiliary space
\begin{equation}
T(x)=\tr_{a}\!\left(
\Kfull_{a}(x)
\eR_{a,1}(x)\cdots \eR_{a,N}(x)
\KhatFull_{a}(x)
\eR_{N,a}(x)\cdots \eR_{1,a}(x)
\right).
\end{equation}
In addition, we consider the transfer matrix $T_{f}(x)$ associated with the fundamental representation of $gl(\mathcal{M}+1)$ in the auxiliary space\begin{equation}
T_{f}(x)=\tr_{f}\!\left(
\KFundFull(x)
L_{1}(x)\cdots L_{{N}}(x)
\KhatFundFull(x)
L_{{N}}(x)\cdots L_{1}(x)
\right).
\end{equation}
While the transfer matrix $T(x)$ yields the (stochastic) Hamiltonian via logarithmic differentiation, $T_{f}(x)$ plays a central role in the framework of the Bethe ansatz \cite{Faddeev:1996iy}. In order for these two families of transfer matrices to be simultaneously diagonalizable, it is necessary to impose their mutual commutativity
\begin{equation}
\left[T(x),T_{f}(y)\right]=0.
\end{equation}
This relation holds if the $R$-matrix $\eR(x)$ and the K-matrices $\hat{\mathcal{K}}(x)$, $\mathcal{K}(x)$ satisfy the $RLL$ relations~\eqref{RLL},~\eqref{shifted-BYBE}, and~\eqref{D-BYBE}. These equations were solved in Proposition~\ref{Proposition-R-components} and Corollary~\ref{corollary-Full-K}, respectively, yielding explicit expressions for the relevant $R$- and $K$-matrices. We are thus in a position to construct $T(x)$ explicitly and to derive the corresponding (stochastic) Hamiltonian by means of its logarithmic derivative.

As will be demonstrated below, the stochasticity of the boundary Hamiltonian requires a specific choice of the boundary parameters, namely we fix
\begin{equation}\label{paramter-choice_K}
\frac{q_{1}}{q_{2}}=\frac{1}{2}-s, 
\qquad 
\frac{q'_{1}}{q'_{2}}=s-\frac{\M}{2}\,.
\end{equation}
For this choice, the diagonal boundary matrices obtained in Proposition~\ref{Proposition-right-Kmatrix} and Corollary~\ref{Corollary-left-Kmatrix} read
\begin{equation}\label{Stochastic-Diagonal-Kmatrices}
\Khat(x)= 
\frac{\Gamma(-x+2s)}{\Gamma(x+2s)}
\frac{\Gamma(x+2s+\OFD\OF)}{\Gamma(-x+2s+\OFD\OF)},
\qquad 
\K(x)=
\frac{\Gamma(x+\M+1)}{\Gamma(-x)}
\frac{\Gamma(-x+\OFD\OF)}{\Gamma(x+\M+1+\OFD\OF)}.
\end{equation}

\begin{theorem}\label{Thm-transfer-Ham}
The Hamiltonian reads
\begin{equation}
  \eH = \mathcal{H}_\text{left} + \sum_{\ell=1}^{\N-1}\mathcal{H}_{\ell,\ell+1} + \mathcal{H}_\text{right}\,,
\end{equation}
where the bulk Hamiltonian density $\mathcal{H}_{\ell,\ell+1}$ is given in Corollary~\ref{Hamalg} and the boundary terms read
\begin{equation}\label{right-left-BD-hamiltonian}
 \mathcal{H}_{\text{left},\text{right}} = D_{\rho_{l,r}}\left(\psi(2s+\OFD\OF)-\psi(2s)\right)D_{\rho_{l,r}}^{-1}\,.
\end{equation}
Here $D_{\rho_{l,r}}$ is the similarity transformation \eqref{dalpha} and $\rho_{l},\rho_{r}$ are defined in \eqref{rhos}. This Hamiltonian is generated by the double-row transfer matrix
\begin{equation}
T(x)=\tr_{a}\Big(\Kfull_{a}(x)\eR_{a,1}(x)\cdots \eR_{a,N}(x)\KhatFull_{a}(x)\eR_{N,a}(x)\cdots \eR_{1,a}(x)\Big)\,,
\end{equation}
where $\tr_{a}$ denotes the trace over the auxiliary space $\mathcal{F}$ (cf.~\eqref{fock-space}). The $R$-matrix is given by $\eR(x)$ with components~\eqref{eq:fullR}, while the K-matrices are written in~\eqref{Kfulls}, 
with boundary parameters $\beta_{l,a}$ and $\beta_{r,a}$ introduced in Section \ref{sec:markov}. 
More precisely, the Hamiltonian is obtained from the logarithmic derivative of the transfer matrix as
\begin{equation}
 \eH = \frac{1}{2}\left(\frac{\partial}{\partial x}\ln T(x)\bigg|_{x=0} + h(\M)\ID \right)\,,
\end{equation}
where $h(\M)=\sum_{a=1}^{\M}\frac{1}{a}$ denotes the $\M$-th harmonic number.
\end{theorem}

\begin{proof}
We write the logarithmic derivative of the transfer matrix as
\begin{equation}\label{eq:logderiv}
\frac{\partial}{\partial x}\ln T(x)\big|_{x=0}=\frac{\tr_a \Kfull_a'(0)}{\tr_a \Kfull_a(0)}+2\frac{\tr_a \Kfull_a(0)\mathcal{H}_{a,1}}{\tr_a \Kfull_a(0)}+\frac{\KhatFull_\N'(0)}{\KhatFull_\N(0)}+2\sum_{\ell=1}^{\N-1}\frac{\partial}{\partial x}\ln \eR_{\ell,\ell+1}(x)\big|_{x=0}\,,
\end{equation}
and study the resulting terms below.
\paragraph{First term.}
 Using the Gaussian hypergeometric function ${}_2F_1(-x + n_2 + \cdots + n_{\M},1;x + M + 1 + n_2 + \cdots + n_M;1)$, one obtains
\begin{equation}\sum_{n_1=0}^\infty \frac{\Gamma(-x +n_1+\ldots+n_\M)}{\Gamma(x+\M+1+n_1+\ldots+n_\M)}=\frac{\Gamma (M+2 x) \Gamma (-x+n_2+\ldots+n_\M)}{\Gamma (M+2 x+1) \Gamma (x+M+n_2+\ldots+n_\M)}\,.
\end{equation}
As a consequence, we find that 
\begin{equation}
 \tr\Kfull(x)=\frac{\Gamma (x+M+ 1)  }{\Gamma (x+1) }\frac{\Gamma (2x+1)  }{\Gamma (2x+M+1) }\,,
\end{equation}
which in turn implies that
\begin{equation}
 \tr\Kfull(0)=1\,,\qquad \tr\Kfull'(0)=-\sum_{k=1}^M\frac{1}{k}
\end{equation}
Finally, by means of derivation we conclude that
\begin{equation}
 \frac{\tr \Kfull'(0)}{\tr \Kfull(0)}=-\sum_{k=1}^M\frac{1}{k}\,.
\end{equation}
\paragraph{Second term.}
Using the cyclicity of the trace  and the invariance of the Hamiltonian density, the  numerator of the second term can be written as
\begin{equation}
\begin{split}
\tr_a \Kfull_{a}(0)\eH_{a,1} &=\tr_a D_{\rho_l,a}\K_{a}(0)D^{-1}_{\rho_l,a}\eH_{a,1}\\&=\tr_a \K_{a}(0)D^{-1}_{\rho_l,a}\eH_{a,1}D_{\rho_l,a}\\&=D_{\rho_l,1}\tr_a \K_a(0)\eH_{a,1}D_{\rho_l,1}^{-1}
 \end{split}
\end{equation} 
We now note that $\mathcal{K}^0$ becomes a projector on the Fock vacuum for spectral parameter $x=0$, such that  
\begin{equation}\label{eq:projK}
\begin{split}
\langle n|
\tr_a \K_{a}(0)\mathcal{H}_{a,1}|n'\rangle&=\sum_{m_1=0}^{\infty}\cdots\sum_{m_\M=0}^\infty
\langle m,n| \K_{a}(0)\eH_{a,1}|m,n'\rangle\\&= \langle 0,n|  \eH_{a,1}|0,n'\rangle
\\&=\left[\psi(|n|+2s)-\psi (2s)\right] \prod_{a=1}^{\M} \delta_{n'_a,n_a}
 \end{split}
\end{equation}
Thus we can write
\begin{equation}
\begin{split}
\tr_a \K_{a}(0) \eH_{a,1} &=\psi(\OAD\OA+2s)-\psi (2s)
 \end{split}
\end{equation} and obtain
\begin{equation}\label{left-BD-Hamiltonian-from-K}
\begin{split}
2\frac{\tr_a \K_a(0)\eH_{a,1}  }{ \tr\K(0)}=2 D_{\rho_l}\left[\psi(\OAD\OA+2s)-\psi (2s)\right] D_{\rho_l}^{-1}\,,
 \end{split}
\end{equation} which corresponds to the left boundary Hamiltonian written in~\eqref{right-left-BD-hamiltonian}.
\paragraph{Third term.} To evaluate the third term we note that
\begin{align}
		   \KhatFull(0)=\ID\,.
	\end{align}
and
\begin{align}\label{right-BD-Hamiltonian-from-K}
		   \KhatFull'(0)= 2D_{\rho_{r}}\left(\psi(\OAD\OA+2s)-\psi(2s)\right )D_{\rho_{r}}^{-1}\,,
	\end{align}
    which corresponds to the right boundary Hamiltonian written in~\eqref{right-left-BD-hamiltonian}.
    \paragraph{Fourth term.} By Corollary~\ref{corollary-bulk-elements}, this term gives the bulk Hamiltonian. 
\end{proof}
We now show that the matrix components of the boundary Hamiltonians obtained in Theorem~\eqref{Thm-transfer-Ham} are given by~\eqref{hleft} and~\eqref{hright}, leading to the boundary generators~\eqref{left-boundary-generator} and~\eqref{right-boundary-generator}. 
\begin{proposition}\label{proposition-Boundary-elemnents}
 The matrix elements of the boundary Hamiltonian~\eqref{right-left-BD-hamiltonian} are
 \begin{equation}
     \langle n|\mathcal{H}_{\text{left},\text{right}} |m\rangle =\begin{cases}
					\psi(|m|+2s)-\psi(2s)+\log(1-|\beta_{l,r}|)& \text{ if }\,n_a=m_a \\
					-\varphi_s(m-n,m)& \text{ if }\,n_a\leq m_a \text{ and }\, |n|<|m|\\
					-\Gamma (|n|-|m|)\prod_{a=1}^{\M}\frac{\beta_{l,r,a}^{n_{a}-m_a}}{(n_{a}-m_a)!}
					&\text{ if }\, n_{a}\geq m_{a}\text{ and }\,|n|>|m|\;				\end{cases}
	 \end{equation}
\end{proposition}
\begin{proof}
The proof immediately follows from the subsequent Lemma~\ref{lem:b1} and Lemma~\ref{lem:b2}
\end{proof}
For the sake of notation, we will not write in the following Lemmas the subscripts $r$ and $l$, but simply $\beta$ and $\rho$. 
\begin{lemma}\label{lem:b1}
The operator
\begin{equation}\label{O}
\begin{split}
\mathcal{O}=
 e^{-\sum_{a=1}^{\M}\rho_{a}J_{a,0}}\left(\psi(\OAD\OA+2s)-\psi(2s)\right )e^{\sum_{a=1}^{\M}\rho_{a}J_{a,0}}
 \end{split}
\end{equation}
has matrix elements
\begin{equation}\label{eq:compb1}
\begin{split}
\langle n|\mathcal{O}|m\rangle=\langle n_1,\ldots,n_\M|\mathcal{O}|m_1,\ldots,m_\M\rangle=\begin{cases}
                                \psi(|m|+2s)-\psi(2s)& n_a=m_a \text{ for all } a\\
						0& n_a<m_a \text{ for any } a\\
                              -\Gamma (|n|-|m|)\prod_{a=1}^{\M}\frac{\rho_{a}^{n_{a}-m_a}}{(n_{a}-m_a)!}
						& \text{else}
                              \end{cases}
 \end{split}
\end{equation}
\end{lemma}

\begin{proof}
Using the relation~\eqref{mp1right} we can write the components as
\begin{equation}\label{eq:compb11}
\begin{split}
&\langle n|
 \mathcal{O}|m\rangle
  =\frac{\Gamma(|n|+2s)}{\Gamma(|m|+2s)}\left[\prod_{a=1}^{\M}(-1)^{m_a}\rho_{a}^{n_{a}-m_a}\right]\sum_{k=0}^\infty \left[\prod_{a=1}^{\M}\frac{(-1)^{k_{a}}}{(n_{a}-k_a)!(k_{a}-m_a)! }\right]\left(\psi(|k|+2s)-\psi(2s)\right )
 \end{split}
\end{equation}
We see that $\langle n|
 \mathcal{O}|m\rangle$ vanishes if $n_a<m_a$  for any $a$ and exclude this case in the following.
We now use the integral representation~\eqref{eq:intrp} such that
\begin{equation}
 \begin{split}
  &\sum_{k=0}^\infty \left[\prod_{a=1}^{\M}\frac{(-1)^{k_{a}}}{(n_{a}-k_a)!(k_{a}-m_a)! }\right]\left(\psi(|k|+2s)-\psi(2s)\right )\\&=\int_0^1 dt\frac{t^{2s-1}}{1-t}\sum_{k=0}^\infty \left[\prod_{a=1}^{\M}\frac{(-1)^{k_{a}}}{(n_{a}-k_a)!(k_{a}-m_a)! }-\prod_{a=1}^{\M}\frac{(-t)^{k_{a}}}{(n_{a}-k_a)!(k_{a}-m_a)! }\right]
 \end{split}
\end{equation}
Then using that
\begin{equation}
\sum_{i=0}^\infty \frac{(-1)^{i}}{(n-i)!(i-m)! }=(-1)^m\delta _{n,m}
\end{equation}
and
\begin{equation}
\sum_{i=0}^\infty \frac{(-t)^{i}}{(n-i)!(i-m)! }=\frac{(-t)^m (1-t)^{n-m}}{(n-m)!}
\end{equation}
we find
\begin{equation}
 \begin{split}
  &\sum_{k=0}^\infty \left[\prod_{a=1}^{\M}\frac{(-1)^{k_{a}}}{(n_{a}-k_a)!(k_{a}-m_a)! }\right]\left(\psi(|k|+2s)-\psi(2s)\right )\\&=(-1)^{|m|}\int_0^1 dt\frac{t^{2s-1}}{1-t} \left[\prod_{a=1}^{\M}  \delta _{n_a,m_a} -\prod_{a=1}^{\M}\frac{t^{m_a} (1-t)^{n_a-m_a}}{(n_a-m_a)!}\right]\\
  &=(-1)^{|m|}\int_0^1 dt\frac{t^{2s-1}}{1-t}\left[(1-t^{|m|}) \prod_{a=1}^{\M}\delta _{n_a,m_a}-\left[\prod_{a=1}^{\M}\frac{t^{m_a} (1-t)^{n_a-m_a}}{(n_a-m_a)!}\right]\left(1-\prod_{a=1}^{\M}\delta _{n_a,m_a}\right)\right]
 \end{split}
\end{equation}
Finally using the integral representations in~\eqref{eq:intrp} we find
\begin{equation}
 \int_0^1 dt\frac{t^{2s-1}(1-t^{|m|})}{1-t} =\psi(|m|+2s)-\psi(2s)
\end{equation}
and
\begin{equation}
 \int_0^1 dt\frac{t^{2s-1}}{1-t}  \left[\prod_{a=1}^{\M}t^{m_a} (1-t)^{n_a-m_a}\right]=\frac{ \Gamma (|n|-|m|) \Gamma (|m|+2 s)}{\Gamma (|n|+2 s)}
\end{equation}
so we obtain~\eqref{eq:compb1} from~\eqref{eq:compb11}.
\end{proof}

\begin{lemma}\label{lem:b2}
The operator
\begin{equation}
\begin{split}\label{O'}
\mathcal{O}'=
 \exp(-\sum_{a=1}^{\M}J_{0,a})\mathcal{O}\,\exp(\sum_{a=1}^{\M}J_{0,a})
 \end{split}
\end{equation}
has matrix elements

\begin{equation}
			\begin{split}
				\langle n|\mathcal{O}'|m\rangle=&\langle n_1,\ldots,n_\M|\mathcal{O}'|m_1,\ldots,m_\M\rangle\\=&\begin{cases}
					\psi(|m|+2s)-\psi(2s)+\log(1-|\beta|)& \text{ if }n_a=m_a \text{ for all } a\\
					-\varphi_s(m-n,m)& \text{ if }n_a\leq m_a \text{ and } |n|<|m|\\
					-\Gamma (|n|-|m|)\prod_{a=1}^{\M}\frac{\beta_{a}^{n_{a}-m_a}}{(n_{a}-m_a)!}
					&\text{ if } n_{a}\geq m_{a}\text{ and }|n|>|m|\;				\end{cases}
			\end{split}
		\end{equation}
\end{lemma}
\begin{proof}
Using,~\eqref{exponential-annhilation} we have
\begin{equation}\label{to-find-exact}
\begin{split}
\langle n|\mathcal{O}'|m\rangle
 &=\sum_{m'}(-1)^{|m'|-|n|}\left[\psi(|m'|+2s)-\psi(2s)\right]\prod_{a=1}^{\M}\binom{m'_{a}}{m'_{a}-n_a}\binom{m_{a}}{m_{a}-m'_a} \\
 &-\sum_{n'_a\leq m'_a:|n'|< |m'|} (-1)^{|m'|-|n|}   \Gamma (|m'|-|n'|)    \prod_{a=1}^{\M}\frac{\rho_{a}^{m'_{a}-n'_a}}{(m'_{a}-n'_a)!}\binom{m'_{a}}{m'_{a}-n_a} \binom{m_{a}}{m_{a}-n'_a}
 \end{split}
\end{equation}
	We consider separately the first and second terms of the equation above.
\paragraph{First term.}

	The first term is evaluated following Lemma~\ref{lem:b1}. We get 		\begin{equation}
			\begin{split}
				& \sum_{m'}(-1)^{|m'|-|n|}\left[\psi(|m'|+2s)-\psi(2s)\right]\prod_{a=1}^{\M}\binom{m'_{a}}{m'_{a}-n_a}\binom{m_{a}}{m_{a}-m'_a} \\
				& =(-1)^{|n|}\left[\prod_{a=1}^{\M}\frac{m_a!}{n_a!}\right]\sum_{m'}\left[\psi(|m'|+2s)-\psi(2s)\right]\prod_{a=1}^{\M}\frac{(-1)^{m_a'}}{(m'_a-n_a)!(m_a-m_a')!}\\
				&=\begin{cases}
					\psi(|m|+2s)-\psi(2s)& n_a=m_a \text{ for all } a\\
					0& m_a<n_a \text{ for any } a\\
					-\frac{ \Gamma (|m|-|n|) \Gamma (|n|+2 s)}{\Gamma (|m|+2 s)}\prod_{a=1}^\M\binom{m_a}{n_a}
					& \text{else}
				\end{cases}
							\end{split}
		\end{equation}

				\paragraph{Second term.}
		For the sake of notation, we do not write here the minus sing in front. We have to analyze $3$ possible cases, depending on the order relations between the vectors $m$ and $n$. 
		\begin{itemize}
			\item \textbf{Case $n\geq m$ and $|n|>|m|$.}
				We re-write this term as
		\begin{equation}
			\begin{split}
						&\sum_{n'_a\leq m'_a:|n'|< |m'|} (-1)^{|m'|-|n|}   \Gamma (|m'|-|n'|)    \prod_{a=1}^{\M}\frac{\rho_{a}^{m'_{a}-n'_a}}{(m'_{a}-n'_a)!}\binom{m'_{a}}{n_a} \binom{m_{a}}{n'_a}
			\\=&
			\sum_{n'=n}^{\infty}
						\sum_{m'=0}^{m}
						(-1)^{|n'|-|n|}\Gamma(|n'|-|m'|)\left(\prod_{a=1}^{\M}\frac{\binom{m_{a}}{m'_{a}}\binom{n'_{a}}{n_{a}}\rho_{a}^{n'_{a}-m'_{a}}}{\Gamma(n'_{a}-m'_{a}+1)}\right)\mathbbm{1}_{\{n'\geq m'\}}\mathbbm{1}_{\{|n'|> |m'|\}}\,.
			\end{split}
		\end{equation}
	Here, in the last line we have changed the name of the summation indices by $n'\to m'$ and $m'\to n'$ for notation convenience.  Moreover, for the sake of notation we will often write $\sum_{n=n'}^{\infty}$ as shortcut that means $\sum_{n'_{1}=n_{1}}^{\infty}\cdots\sum_{n'_{\M}=n_{\M}}^{\infty}$. 

	We recall (see for instance see equation (104) of \cite{lauricella1893sulle}) that the Lauricella's function of type $D$ with $N$ variables   
	\begin{align}\label{definition-lauricella}
		F_D^{(N)}\bigl(a; b_j; c_j; x_j\bigr)
		= \sum_{m=0}^{\infty}
			\frac{(a)_{m_1 + \cdots + m_N} \prod_{j=1}^N (b_j)_{m_j}}
		{\prod_{j=1}^N (c_j)_{m_j} \; m_j!}
		\prod_{j=1}^N x_j^{m_j}\,
	\end{align}
	satisfies the property 
	\begin{align}\label{pfaff-Lau}
		F_D^{(N)}(a; b_j; c_j; x_j)
		= \left(1 - |x|\right)^{-a} 
		F_D^{(N)} \left(
		a;  c_ j- b_j; c_j;\,
		\frac{x_j}{|x| - 1}	\right)
	\end{align}
	
	By performing a the change of summation indices $n'_{a}=k_{a}+n_{a}$, we obtain 
	\begin{equation}\label{F-intermediate}
		\begin{split}
		&F_{D}^{(\M)}(|n|-|m'|;n_{j}+1;n_{j}-m'_{j}+1;-\rho_{j})
		\\=&
		\sum_{k=0}^{\infty}
				\frac{(|n|-|m'|)_{|k|}\prod_{a=1}^{\M}(n_{a}+1)_{k_{a}}}{\prod_{a=1}^{\M}(n_{a}-m'_{a}+1)_{k_{a}}k_{a}!}(-1)^{|k|}\prod_{a=1}^{\M}\rho_{a}^{k_{a}}
							\\=&	\sum_{n'=n}^{\infty}
				\frac{\Gamma(|n'|-|m'|)}{\Gamma(|n|-|m'|)}(-1)^{|n'|-|n|}\prod_{a=1}^{\M}\frac{\Gamma(n_{a}-m'_{a}+1)}{\Gamma(n'_{a}-m'_{a}+1)}\binom{n'_{a}}{n_{a}}\rho_{a}^{n'_{a}-n_{a}}
	\end{split}
		\end{equation}

	By using first~\eqref{F-intermediate} and then~\eqref{pfaff-Lau}, we obtain 
	\begin{equation}
		\begin{split}
		&
		\sum_{n'=n}^{\infty}\sum_{m'=0}^{m}
				(-1)^{|n'|-|n|}\Gamma(|n'|-|m'|)
			\left(\prod_{a=1}^{\M}\frac{\binom{m_{a}}{m'_{a}}\binom{n'_{a}}{n_{a}}\rho_{a}^{n'_{a}-m'_{a}}}{\Gamma(n'_{a}-m'_{a}+1)}\right)\mathbbm{1}_{\{n'\geq m'\}}\mathbbm{1}_{\{|n'|> |m'|\}}
						\\=&
	\sum_{m'=0}^{m}
	F_{D}^{(\M)}(|n|-|m'|;-m'_{j};n_{j}-m'_{j}+1;\beta_{j})
	\\&\qquad \cdot
		\Gamma(|n|-|m'|)(1-|\beta|)^{|n|-|m'|}\prod_{a=1}^{\M}	\frac{\binom{m_{a}}{m'_{a}}\rho_{a}^{n_{a}-m'_{a}}}{\Gamma(n_{a}-m'_{a}+1)}\,.\end{split}
	\end{equation}
	By using~\eqref{definition-lauricella}, by performing the change of summation indices $k_{a}=m'_{a}-n'_{a}$ for all $a\in \{1,\ldots,\M\}$ and by recalling that
	\begin{equation}
		(-\alpha)_{k}
			=\frac{\alpha!}{(\alpha-k)!}(-1)^{k}\,,
	\end{equation} we have that 
	\begin{equation}
		\begin{split}
		&	\sum_{m'=0}^{m}
				F_{D}^{(\M)}(|n|-|m'|;-m'_{j};n_{j}-m'_{j}+1;\beta_{j})
		\\&\qquad \cdot
		\Gamma(|n|-|m'|)(1-|\beta|)^{|n|-|m'|}\prod_{a=1}^{\M}	\frac{\binom{m_{a}}{m'_{a}}\rho_{a}^{n_{a}-m'_{a}}}{\Gamma(n_{a}-m'_{a}+1)}																\\=&
		\left(\prod_{a=1}^{\M}m_{a}!\right)
		\sum_{k=0}^{m}
			\sum_{m'=k}^{m}	(-1)^{|m'|-|k|}\Gamma(|n|-|k|)
			\\&\qquad \qquad\qquad\qquad \quad  \cdot\prod_{a=1}^{\M}\frac{\beta_{a}^{n_{a}-k_{a}}}{(m_{a}-m'_{a})!k_{a}!(m'_{a}-k_{a})!\Gamma(n_{a}-k_{a}+1)}
		\\=&
		\sum_{k=0}^{m}
				\Gamma(|n|-|k|)\left(\prod_{a=1}^{\M}\frac{\binom{m_{a}}{k_{a}}\beta_{a}^{n_{a}-k_{a}}}{\Gamma(n_{a}-k_{a}+1)}\right)
			\left(\sum_{m'=k}^{m}
				(-1)^{|m'|-|k|}\prod_{a=1}^{\M}\binom{m_{a}-k_{a}}{m'_{a}-k_{a}}\right)\,.
\end{split}
	\end{equation}
	 By the binomial theorem, we obtain that
\begin{equation}
	\sum_{m'_{a}=k_{a}}^{m_{a}}(-1)^{m'_{a}-k_{a}}\binom{m_{a}-k_{a}}{m'_{a}-k_{a}}=\sum_{m'_{a}=0}^{m_{a}-k_{a}}(-1)^{m'_{a}}\binom{m_{a}-k_{a}}{m'_{a}}=\mathbbm{1}_{\{m_{a}=k_{a}\}}\,.
\end{equation}
Therefore, we finally obtain that 
\begin{equation}
	\begin{split}
				&\sum_{k=0}^{m}
				\Gamma(|n|-|k|)\left(\prod_{a=1}^{\M}\frac{\binom{m_{a}}{k_{a}}\beta_{a}^{n_{a}-k_{a}}}{\Gamma(n_{a}-k_{a}+1)}\right)
				\left(\sum_{m'=k}^{m}
				(-1)^{|m'|-|k|}\prod_{a=1}^{\M}\binom{m_{a}-k_{a}}{m'_{a}-k_{a}}\right)\,.
		\\=&
		\Gamma(|n|-|m|)\left(\prod_{a=1}^{\M}\frac{\beta_{a}^{n_{a}-m_{a}}}{\Gamma(n_{a}-m_{a}+1)}\right)\,.
	\end{split}
\end{equation} 

\item \textbf{Case $m=n$.}
  We recall two identities that we use in the following:
 \begin{align}\label{Beta-single-variable-property}
 	\sum_{r=0}^{\ell}(-1)^{r}\binom{\ell}{r}\frac{\Gamma(r+b)}{\Gamma(r+c)}=\frac{B(\ell+c-b,b)}{\Gamma(c-b)}
 \end{align}
 and, for $p,k\in \mathbb{N}_{0}$,
 \begin{align}\label{beta-ration-property}
 	\lim_{q\to p}\frac{B(p-q,k+1)}{\Gamma(-q)}=(-1)^{p}p!\,.
 \end{align}
 We write the second term of~\eqref{to-find-exact} after having performed a change of summation variables $ k_{a}=n'_{a}-m'_{a}$ for all $a\in \{1,\ldots,\M\}$:
\begin{equation}
	\begin{split}
		&\sum_{n'=m}^{\infty}
				\sum_{m'=0}^{m}(-1)^{|n'|-|n|}\Gamma(|n'|-|m'|)
				\\&\qquad\qquad\cdot \left(\prod_{a=1}^{\M}\binom{m_{a}}{m'_{a}}\binom{n'_{a}}{n_{a}}\frac{\rho_{a}^{n'_{a}-m'_{a}}}{\Gamma(n'_{a}-m'_{a}+1)}\right)\mathbbm{1}_{\{n'\geq m'\}}\mathbbm{1}_{\{|n'|>|m'|\}}\mathbbm{1}_{\{m=n\}}
												\\=&
		\sum_{k=0}^{\infty}
			\mathbbm{1}_{\{|k|>0\}}(-1)^{|k|-|n|}\Gamma(|k|)\left(\prod_{a=1}^{\M}\frac{\rho_{a}^{k_{a}}}{\Gamma(k_{a}+1)n_{a}!}\right)
		\\&\qquad \cdot \left(\sum_{m'=0}^{m}
			\prod_{a=1}^{\M}(-1)^{m'_{a}}\binom{m_{a}}{m'_{a}}\frac{\Gamma(k_{a}+m'_{a}+1)}{\Gamma(k_{a}+m'_{a}-n_{a}+1)}\right)\mathbbm{1}_{\{m=n\}}
		\end{split}
	\end{equation}
	By using equality~\eqref{Beta-single-variable-property} (with $b=k_{a}+1$, $c=k_a-q_a+1$) and using~\eqref{beta-ration-property}, we have that 
	\begin{equation}
	\begin{split}
		&\sum_{k=0}^{\infty}
				\mathbbm{1}_{\{|k|>0\}}(-1)^{|k|-|n|}\Gamma(|k|)\left(\prod_{a=1}^{\M}\frac{\rho_{a}^{k_{a}}}{\Gamma(k_{a}+1)n_{a}!}\right)
		\\&\qquad \cdot \left(\sum_{m'=0}^{m}
				\prod_{a=1}^{\M}(-1)^{m'_{a}}\binom{m_{a}}{m'_{a}}\frac{\Gamma(k_{a}+m'_{a}+1)}{\Gamma(k_{a}+m'_{a}-n_{a}+1)}\right)\mathbbm{1}_{\{m=n\}}
	\\=&
	\sum_{k=0}^{\infty}
			\mathbbm{1}_{\{|k|>0\}}(-1)^{|k|-|n|}\Gamma(|k|)\left(\prod_{a=1}^{\M}\frac{\rho_{a}^{k_{a}}}{\Gamma(k_{a}+1)n_{a}!}\right) \left(\prod_{a=1}^{\M}\lim_{q_{a}\to n_{a}}\frac{B(m_{a}-q_{a},k_{a}+1)}{\Gamma(-q_{a})}\right)
		\\=&
		\sum_{k=0}^{\infty}\mathbbm{1}_{\{|k|>0\}}(-1)^{|k|}\Gamma(|k|)\left(\prod_{a=1}^{\M}\frac{\rho_{a}^{k_{a}}}{\Gamma(k_{a}+1)}\right)
		\\=&
		-\log\left(1+\frac{|\beta|}{1-|\beta|}\right)
		\\=&
		\log(1-|\beta|)\,.
	\end{split}
\end{equation}		
The identity can then be extended to $|\beta|<1$ by analytic continuation.
\item \textbf{Remaining cases.}
The remaining cases for the second addend of~\eqref{to-find-exact} are given by the situations in which 
\begin{equation}
	\exists\, a\in \{1,\ldots \M\}\;:\; n_{a}<m_{a}\,.
\end{equation}
Without loss of generality, assume that $n_1<m_1$. After performing the change of summation variable $k_1=n'_1-m'_1$, we obtain
\begin{equation}
	\begin{split}
		&
		\sum_{n'=n}^{\infty}
		\sum_{m'=0}^{m}
		(-1)^{|n'|-|n|}
		\Gamma(|n'|-|m'|)
		\left(
		\prod_{a=1}^{\M}
		\binom{m_a}{m'_a}
		\binom{n'_a}{n_a}
		\frac{\rho_a^{n'_a-m'_a}}
		{\Gamma(n'_a-m'_a+1)}
		\right)
		\mathbbm{1}_{\{n'\geq m'\}}
		\mathbbm{1}_{\{|n'|>|m'|\}}
		\\
		&=
		\sum_{m'_2=0}^{m_2}
		\sum_{n'_2=n_2}^{\infty}
		\cdots
		\sum_{m'_{\M}=0}^{m_{\M}}
		\sum_{n'_{\M}=n_{\M}}^{\infty}
		\left(
		\prod_{a=2}^{\M}
		\binom{m_a}{m'_a}
		\binom{n'_a}{n_a}
		\frac{
			\rho_a^{n'_a-m'_a}
			(-1)^{n'_a-n_a}
		}{
			\Gamma(n'_a-m'_a+1)
		}
		\right)
		\\
		&\;\cdot
		\left(
		\prod_{a=2}^{\M}
		\mathbbm{1}_{\{n'_a\geq m'_a\}}
		\right)
					\sum_{k_1=0}^{\infty}
		\mathbbm{1}_{\left\{
			k_1+\sum_{a=2}^{\M}(n'_a-m'_a)>0
			\right\}}
		(-1)^{k_1-n_1}
		\frac{
			\Gamma\left(
			k_1+\sum_{a=2}^{\M}(n'_a-m'_a)
			\right)
		}{
			\Gamma(k_1+1)
		}
		\rho_1^{k_1}
		\frac{1}{n_1!}
		\\
		&\;\cdot
		\sum_{m'_1=0}^{m_1}
		(-1)^{m'_1}
		\binom{m_1}{m'_1}
		\frac{\Gamma(k_1+m'_1+1)}
		{\Gamma(k_1+m'_1-n_1+1)}.
	\end{split}
\end{equation}
For every $k_1\in\mathbb{N}_0$, using
\eqref{Beta-single-variable-property} with
$\ell=m_1$, $b=k_1+1$, and $c=k_1-q+1$, we obtain
\begin{equation}
	\begin{aligned}
		&
		\sum_{m'_1=0}^{m_1}
		(-1)^{m'_1}
		\binom{m_1}{m'_1}
		\frac{\Gamma(k_1+m'_1+1)}
		{\Gamma(k_1+m'_1-n_1+1)}
		=
		\lim_{q\to n_1}
		\frac{B(m_1-q,k_1+1)}
		{\Gamma(-q)}
		=
		0,
	\end{aligned}
\end{equation}
since $m_1>n_1$. Hence, the whole expression vanishes.
\end{itemize}
\end{proof}

\begin{remark}
 As seen in this section, the derivation of the Hamiltonian for the Fock representation $\pi$ uses that the K-matrix becomes a projector on the Fock vaccum, cf.~\eqref{eq:projK}. Given the intertwining relations \eqref{intertw} and \eqref{intertw-dual}, one may expect that a similar mechanism may exists for the integral representations $\rep$ and $\dualrep$ which would allow to prove integrability of the dual models presented in Section~\ref{sec:hidden} and Section~\ref{sec:heat} directly. To this aim one would need to take the trace over an appropriate basis of polynomials. For our purpose it is sufficient to focus on the transfer matrix construction of the Hamiltonian in representation $\pi$. The dual models with boundaries will be constructed using duality functions in Section~\ref{section-duality}.
\end{remark}

\subsection{Isospectral models with simpler boundary conditions}
Following \cite{exact-harmonic}, we can define two Hamiltonians that are isospectral to the stochastic Hamiltonian of Theorem \eqref{Thm-transfer-Ham}. One of the Hamiltonians $\mathcal{H}'$ is directly related to an absorbing dual process that is discussed in Section~\ref{sec:abs}. The other Hamiltonian $\mathcal{H}''$ has been used in \cite{exact-harmonic} to obtain the steady state for the singlespecies case. The local transformations that relate the Hamiltoians are
\begin{equation}\label{similarities}
    S'=\exp{\left(-\sum_{a=1}^{\M}J_{0,a}\right)}\,,\qquad S''_{\alpha}=\exp{\left(-\sum_{a=1}^{\M}\alpha_{a}J_{a,0}\right)}
\end{equation}
such that
\begin{equation}
  D_\alpha=S'S''_\alpha\,,
 \end{equation}
 cf.~\eqref{dalpha}.
\begin{proposition}\label{Proposition-Hp}
The Hamiltonian of Theorem \eqref{Thm-transfer-Ham} is isospectral to the Hamiltonian
\begin{equation}\label{prop-Ham}
  \eH' = \mathcal{H}_\text{left}' + \sum_{\ell=1}^{\N-1}\mathcal{H}_{\ell,\ell+1} + \mathcal{H}_\text{right}'\,,
\end{equation}
where the bulk Hamiltonian density $\mathcal{H}_{\ell,\ell+1}$ is given in Corollary~\ref{Hamalg}. The boundary terms are triangular and read
\begin{equation}\label{right-left-BD-hamiltonian-I}
 \mathcal{H}'_{\text{left},\text{right}} = S''_{\rho_{l,r}}\left(\psi(2s+\OFD\OF)-\psi(2s)\right)(S''_{\rho_{l,r}})^{-1}\,,
\end{equation}
\end{proposition}
\begin{proof}
    Given that the Hamiltonian density is invariant under the transformation 
    \begin{equation}
      (  (S')^{-1}\otimes (S')^{-1})\mathcal{H}_{\ell,\ell+1}(S'\otimes S')=0
    \end{equation}
    the similarity transformation only acts non-trivially on the boundaries and we immediately find 
    \begin{equation}
      (  (S')^{-1}\otimes\ldots\otimes (S')^{-1})\mathcal{H}(S'\otimes\ldots\otimes S')=\mathcal{H}'\,.
    \end{equation}

\end{proof}

\begin{proposition}\label{Proposition-Hs}
The Hamiltonian of Proposition \eqref{Proposition-Hp} is isospectral to the Hamiltonian
\begin{equation}
  \eH'' = \mathcal{H}_\text{left}'' + \sum_{\ell=1}^{\N-1}\mathcal{H}_{\ell,\ell+1} + \mathcal{H}_\text{right}''\,,
\end{equation}
where the bulk Hamiltonian density $\mathcal{H}_{\ell,\ell+1}$ is given in Corollary~\ref{Hamalg} and the boundary terms are diagonal on the right and triangular on the left. They read
\begin{equation}\label{right-left-BD-hamiltonian-II}
 \mathcal{H}''_{\text{left}} = S''_{\rho_{l}-\rho_{r}}\left(\psi(2s+\OFD\OF)-\psi(2s)\right)( S''_{\rho_{l}-\rho_{r}})^{-1}\,,\quad  \mathcal{H}''_{\text{right}} = \psi(2s+\OFD\OF)-\psi(2s)\,,
\end{equation}
\end{proposition}
\begin{proof}
   As before we use the invariance of the Hamiltonian density and find 
    \begin{equation}
      (  (S''_{\rho_{r}})^{-1}\otimes\ldots\otimes (S''_{\rho_{r}})^{-1})\mathcal{H}'(S''_{\rho_{r}}\otimes\ldots\otimes S''_{\rho_{r}})=\mathcal{H}''\,.
    \end{equation}    
\end{proof}
The transformation above allows to obtain the steady state in closed form, see \cite{Frassek:2019imp,Frassek:2020omo}. The problem of finding  the non-equilibrium steady state  $|\Psi\rangle$  of the boundary-driven Hamiltonian $\eH$, with
\begin{equation}
    \eH|\Psi\rangle =0\,,
\end{equation}
is reduced to finding the ground state $|\Psi''\rangle$ with $\eH''|\Psi''\rangle =0$ 
By Propositions~\ref{Proposition-Hp} and~\ref{Proposition-Hs}, the steady state and the ground states $|\Psi\rangle$ are related through $|\Psi'\rangle$ via
\begin{equation}
    |\Psi'\rangle= \big((S')^{-1}\otimes \cdots \otimes (S')^{-1}\big)|\Psi\rangle,
    \qquad
    \text{where}\qquad{
    |\Psi''\rangle =\big((S''_{\rho_r})^{-1}\otimes \cdots \otimes (S''_{\rho_r})^{-1}\big)|\Psi'\rangle \,.}
\end{equation}
As in the monospecies case, the  problem of determining $|\Psi''\rangle$ is significantly simpler than the original one as the eigenstate $|\Psi''\rangle$ factorises. We will address the study of the steady state in a follow up article \cite{toap}.

\section{Stochastic dualities}\label{section-duality}
In this section, we establish several duality relations for the models introduced in Section~\ref{sec:model}. Exploiting the $gl(\M+1)$ symmetry of the bulk Hamiltonian, we first show that the multispecies harmonic process is dual to an absorbing process, which exhibits the same bulk dynamics but where the boundary reservoirs are replaced by absorbing sites. This duality provides a natural framework for characterizing the non-equilibrium steady state of the boundary-driven model.
More precisely, it allows to express
the moments of the stationary measure in terms of the absorption probabilities of dual particles.
Furthermore, as a consequence of the intertwining relation~\eqref{intertw} between discrete and continuous representations of $gl(\M+1)$, we prove that the multispecies harmonic process with absorbing boundaries is dual both to the hidden-parameter model and to the open-boundary heat conduction model, each equipped with its corresponding duality function.

Before proving our duality result, we briefly recall the notion of duality for Markov processes~\cite{dualityBook}. We denote by $\mathbb{E}_{x}[\cdot]$ the expectation with respect to the law of the Markov process $(X(t))_{t\geq 0}$ initialized at configuration $x$, and by $\mathbb{E}_{y}[\cdot]$ the expectation with respect to the law of the Markov process $(Y(t))_{t\geq 0}$ initialized at configuration $y$.
\begin{definition}
Let $(X(t))_{t\geq 0}$ be a Markov process on a state space $V_{X}$, with generator $\mathcal{L}_X$, and let $(Y(t))_{t\geq 0}$ be a Markov process on a state space $V_{Y}$, with generator ${\mathcal{L}_Y}$. We say that these two processes are in duality relation with duality function $D:V_{X}\times V_{Y}\to \mathbb{R}$ if, for all $t\geq 0$,
\begin{equation}\label{duality-defintion-via-expectation}
	\mathbb{E}_{x}[D(X(t),y)]=\mathbb{E}_{y}[D(x,Y(t))]\qquad \forall x\in V_{X},\;y\in V_{Y}\,.
\end{equation}
\end{definition}
In our set-up, this definition can equivalently be formulated as a relation between the corresponding Markov generators (see \cite{dualityBook}):
\begin{equation}\label{general-duality-relation-generators}
	\left(\mathcal{L}_{X}D(\cdot,y)\right)(x)=\left(\mathcal{L}_{Y}D(x,\cdot)\right)(y)\,.
\end{equation}
In particular, when the state spaces of the processes are discrete, the generators can be represented as matrices. Therefore, by using relation~\eqref{gen-ham1}, equation~\eqref{general-duality-relation-generators} can be rephrased as
\begin{equation}\label{hamiltonian-duality-general}
	H_{X}^{t}D=D{H}_{Y}\,,
\end{equation}
where $D$ denotes the duality matrix.

\subsection{Absorbing duality}\label{sec:abs}
For the rest of this section, we fix the representation labels as in equation~\eqref{repsLabelFixing}, namely
\begin{equation}
   \mu_1 = -s + \frac{1}{2}, \qquad \mu_2 = s + \frac{1}{2}.
\end{equation}

Depending on convenience, the algebraic expressions for Hamiltonian operators may be written either in terms of the oscillator representation $\pi(\oa_a)$, $\pi(\oad_a)$ in Fock space, or using the $gl(\M+1)$ generators $(\pi(J_{A,B}))_{A,B=0,\ldots,\M}$. For brevity, we denote $\pi(J_{A,B})$ simply as $\eJ_{A,B}$, see~\eqref{E-matrix}.  
To establish the duality, the following identities will be useful:
\begin{equation}\label{good-relation-forFockGLN}
    -\eJ_{0,0} + \mu_2 = -\eJ_{0,0} + \frac{1}{2} + s = 2s + \OAD \OA, 
    \quad 
    \eJ_{a,0} (\eJ_{0,0} - \mu_2)^{-1} = \eJ_{a,0} \left(\eJ_{0,0} - \frac{1}{2} - s\right)^{-1} = \oad_a.
\end{equation}

In order to prove the duality for the bulk process, the reversible measure is required. 
We recall that the probability mass function of the $\text{Negative-Multinomial}(2s,p_{0},\beta_{1},\ldots,\beta_{\M})$ distribution is given by
\begin{equation}\label{mass-negBin}
    \nu^{\,\beta}(m_{1},\ldots,m_{\M})
    =
    \Gamma(|m|+2s)\frac{p_{0}^{2s}}{\Gamma(2s)}
    \prod_{a=1}^{\M}\frac{\beta_{a}^{m_{a}}}{m_{a}!}\,.
\end{equation}
\begin{lemma}\label{lemma-reversibleMeasure-M-harmonic-boundaries}
	The boundary driven multispecies harmonic process with generator~\eqref{boundary-driven-generator} admits a reversible product measure given by
	\begin{equation}\label{reversibleMeasure-M-harmonic-boundary}
		\nu^{\rm{rev},\beta}=\bigotimes_{\ell\in V}\nu_{\ell}^{\,\beta}\qquad \nu^{\,\beta}_{\ell}\sim \textnormal{Negative-Multinomial}(2s,p_{0},\beta_{1},\ldots,\beta_{\M})\;:\;p_{0}+|\beta|=1
	\end{equation} 
	if and only if $\beta_{l,a}=\beta_{r,a}=\beta_{a}$ for all $a\in\{1,\ldots,\M\}$. 
\end{lemma} 
\begin{proof}We impose the detailed balance condition on the bulk and on the boundary.
In the bulk, for any bond $(\ell, \ell+1)$ and for any transition from $(m^{\ell},m^{\ell+1})$ to $({m}^{\ell}-k,m^{\ell+1}+k)$, with $0\le k\le m^{\ell}$,
the detailed balance reads
\begin{align}
	&\nu_{\ell}^{\,\beta}(m^{\ell})\,\nu_{\ell+1}^{\,\beta}(m^{\ell+1})\,\varphi_{s}(k,m^{\ell})=	\nu_{\ell}^{\,\beta}(m^{\ell}-k)\,\nu_{\ell+1}^{\,\beta}(m^{\ell+1}+k)\,\varphi_{s}(k,m+k)\,.
\end{align}
By replacing the transition rates~\eqref{rates} and the probability mass function~\eqref{mass-negBin} one shows that the relation above holds. Analogously, one obtains the verification of the detailed balance for any transition  $(m^{\ell},m^{\ell+1})$ to $({m}^{\ell}+k,m^{\ell+1}-k)$, with $0\le k\le m^{\ell+1}$.

On the left boundary, for any transition that injects $k$ particles at site $1$, the detailed balance reads
\begin{align}
	\nu_{1}^{\,\beta}(m^{1})\Gamma(|k|)\prod_{a=1}^{\M}\frac{\beta_{l,a}^{k_{a}}}{k_{a}!}=\nu_{1}^{\,\beta}(m^{1}+k)\varphi(k,m^{1}+k)\,.
\end{align}
Substituting the transition rate~\eqref{rates} and the mass function~\eqref{mass-negBin} one shows that the above relation holds if and only if $\beta_{a}=\beta_{l,a}$ for all $a\in \{1,\ldots, \M\}$. Analogous computations allow to verify the detailed balance for the removal of $k\in\{1,\ldots, m^{1}\}$ particles from site $1$. Finally, a similar argument on the right boundary leads to $\beta_{r,a}=\beta_{a}$ for all $a\in \{1,\ldots,\M\}$, concluding the proof.
\end{proof}

\begin{theorem}[Absorbing duality]\label{Thm-duality}
	The boundary driven multispecies harmonic process $(\bm{m}(t))_{t\geq 0}$ with generator $\mathcal{L}$ in~\eqref{boundary-driven-generator} is dual to the multispecies harmonic process with absorbing boundaries $(\bm{\xi}(t))_{t\geq 0}$ with generator $\widetilde{\mathcal{L}}$ in~\eqref{abs-Harmonic-Generator}, with duality function $\mathfrak{F}:\Omega\times \widetilde{\Omega}\to \mathbb{R}$ given by
	\begin{equation}\label{duality-funciton}
		\mathfrak{F}(\bm{m},\bm{\xi})=\left(\prod_{a=1}^{\M}\rho_{l,a}^{\xi_{a}^{0}}\right)\left(\prod_{\ell=1}^{\N}\frac{\Gamma(2s)}{\Gamma(2s+|\xi^{\ell}|)}\prod_{a=1}^{\M}\frac{\Gamma(m_{a}^{\ell}+1)}{\Gamma(m_{a}^{\ell}-\xi_{a}^{\ell}+1)}\right)\left(\prod_{a=1}^{\M}\rho_{r,a}^{\xi_{a}^{N+1}}\right)\,.
	\end{equation}
    Above, $\rho_{l,r}$ are the boundary densities introduced in \eqref{rhos}. 
\end{theorem}
\begin{proof}
	It is convenient to reformulate duality in terms of the Hamiltonian operator, which are related to the generators via~\eqref{gen-ham1}. To this aim, to each configuration $\bm{\xi}\in \widetilde{\Omega}$ we associate orthonormal vectors
	\begin{equation}
		|\bm{\xi}\rangle =\bigotimes_{\ell=0}^{\N+1}|\xi_{1}^{\ell},\ldots,\xi_{\M}^{\ell}\rangle\,.
	\end{equation}
	with 
	\begin{equation}
		\langle n^{\ell}|\xi^{\ell}\rangle=\langle n_{1}^{\ell},\ldots,n_{\M}^{\ell}|\xi_{1}^{\ell},\ldots,\xi_{\M}^{\ell}\rangle=\prod_{a=1}^{\M}\delta_{n_{a}^{\ell},\xi_{a}^{\ell}}, \qquad \qquad \ell\in \{0,1,\ldots,\N,\N+1\}\,.
	\end{equation}
	The Hamiltonian of the dual process reads
	\begin{equation}\label{Hamiltonian-ABS-Harmonic}
		\widetilde{\eH}=\widetilde{\eH}_{\text{left}}+\sum_{\ell=1}^{\N-1}\eH_{\ell,\ell+1}+\widetilde{\eH}_{\text{right}}
	\end{equation}
	where $\eH_{\ell,\ell+1}$ is defined in~\eqref{bulk-density-Hamiltonian}. The dual boundary Hamiltonian  $\widetilde{\eH}_{\text{left}}$ 
		has the following expression in terms of the $gl(\M+1)$ algebra generators in the Fock representation~\eqref{E-matrix} 
		\begin{equation}\label{hleftt}
		\begin{split}
			&\widetilde{\eH}_{\text{left}}=\exp{\left(-\sum_{a=1}^{\M}\eJ_{a,0}^{[0]}\left(\eJ_{0,0}^{[0]}-\left(1/2+s\right)\right)^{-1}\eJ_{0,a}^{[1]}\right)}\left(\psi(-\eJ_{0,0}^{[1]}+s+\frac{1}{2})-\psi(2s)\right)
			\\&\;\qquad \quad\cdot
			\exp{\left(\sum_{a=1}^{\M}\eJ_{a,0}^{[0]}\left(\eJ_{0,0}^{[0]}-\left(1/2+s\right)\right)^{-1}\eJ_{0,a}^{[1]}\right)}\,.
		\end{split}
	\end{equation}
    Here, we have denoted $\eJ_{A,B}^{[\ell]}$ the algebra generator acting as $\eJ_{A,B}$ at site $\ell$ and trivially on all the other sites.
    We observe that, by using~\eqref{good-relation-forFockGLN}, the dual left Hamiltonian 
				  ~\eqref{hleftt} coincides with $\eH^{+}_{0,1}$, i.e. the bulk Hamiltonian density given in Corollary~\ref{corollary-bulk-elements} acting on sites $0$ and $1$.
    Similarly, the right boundary reads
    \begin{equation}\label{right-dual-hamiltonian}
    \begin{split}
        &\widetilde{\eH}_{\text{right}}=\exp{\left(-\sum_{a=1}^{\M}\eJ_{a,0}^{[\N+1]}\left(\eJ_{0,0}^{[\N+1]}-\left(1/2+s\right)\right)^{-1}\eJ_{0,a}^{[\N]}\right)}\left(\psi(-\eJ_{0,0}^{[\N]}+s+\frac{1}{2})-\psi(2s)\right)
			\\&\;\qquad \quad\cdot
			\exp{\left(\sum_{a=1}^{\M}\eJ_{a,0}^{[\N+1]}\left(\eJ_{0,0}^{[\N+1]}-\left(1/2+s\right)\right)^{-1}\eJ_{0,a}^{[\N]}\right)}\,.
            \end{split}
    \end{equation}
    Analogously, the dual right Hamiltonian~\eqref{right-dual-hamiltonian} coincides with $\eH^{-}_{N,N+1}$, i.e. the bulk Hamiltonian density given in Corollary~\ref{corollary-bulk-elements} acting on sites $\N$ and $\N+1$. 
    
	We define the duality matrix
	\begin{equation}\label{duality-matrix}
		\mathfrak{F}=\mathcal{I}_{0}\left(\prod_{\ell=1}^{\N}\mathfrak{D}_{\ell}\right)\mathcal{I}_{\N+1}\,,
	\end{equation}
	where the product over sites is interpreted as a tensor product along the chain $\{1,\ldots, \N\}$. Here, we have introduced the single-site duality matrix 
	\begin{equation}\label{duality-matrix-x}
		\mathfrak{D}_{\ell}=\mathfrak{d}_{\ell}\exp{\left(-\sum_{a=1}^{\M}\eJ_{a,0}^{[\ell]}\right)}\,,
	\end{equation}
	and the so-called cheap duality matrix\footnote{Its name comes from the fact that it is obtained "for free" from reversibility, see \cite{dualityBook}.}
	\begin{equation}\label{cheap-duality}
		\mathfrak{d}_{\ell}:=\sum_{\xi^{\ell}}\frac{\Gamma(2s)}{\Gamma\left(2s+|\xi^{\ell}|\right)}\left(\prod_{a=1}^{\M}\Gamma(\xi_{a}^{\ell}+1)\right)|\xi^{\ell}\rangle \langle \xi^{\ell}|\,.
	\end{equation} 
	Observe that the cheap duality matrix~\eqref{cheap-duality} is diagonal with elements
    \begin{equation}
        \langle \xi^{\ell}|\mathfrak{d}_{\ell}|\xi^{\ell}\rangle = (\M+1)^{|\xi^{\ell}|}\nu_{\ell}^{1/(\M+1)}\,.
    \end{equation}
                        	Furthermore, we have introduced the intertwiner
	\begin{equation}\label{intertwining}
		\mathcal{I}_{0}=\sum_{n_{1}^{0},\ldots,n_{\M}^{0}=0}^{\infty}\left(\prod_{a=1}^{\M}\rho_{l,a}^{n_{a}^{0}}\right)\langle n_{1}^{0},\ldots,n_{\M}^{0}|\,,
	\end{equation}
	and similarly $\mathcal{I}_{\N+1}$.
The elements of this duality matrix, for $\ell\in\{1,\ldots, \N\}$, read 
	\begin{align}\label{element-duality-matrix}
		\langle m^{\ell}|\mathfrak{D}_{\ell}|\xi^{\ell}\rangle=
				\frac{\Gamma(2s)}{\Gamma(2s+|\xi^{\ell}|)}\left(\prod_{a=1}^{\M}\frac{\Gamma(m_{a}^{\ell}+1)}{\Gamma(m_{a}^{\ell}-\xi_{a}^{\ell}+1)}\right)\,,
	\end{align}
	while on sites $\ell=0,$ and $\ell=\N+1$ we have, respectively  
	\begin{equation}
		        \mathcal{I}_{0}|\xi^{0}\rangle=\prod_{a=1}^{\M}\rho_{l,a}^{\xi_{a}^{0}}\quad\text{and}\quad  \mathcal{I}_{\N+1}|\xi^{\N+1}\rangle =\prod_{a=1}^{\M}\rho_{r,a}^{\xi_{a}^{\M+1}}\,.
	\end{equation}
    To show duality, we shall prove separately the \textit{bulk duality}
		\begin{equation}\label{bulk-duality}
			\eH_{\ell,\ell+1}^{t}\mathfrak{D}_{\ell}\mathfrak{D}_{\ell+1}=\mathfrak{D}_{\ell}\mathfrak{D}_{\ell+1}\eH_{\ell,\ell+1}\, \qquad \ell\in \{1,\ldots,\N-1\} 
		\end{equation} 
		and the  \textit{boundary duality} 
		\begin{equation}\label{boundary-duality}
			\eH_{\text{left}}^{t}\mathcal{I}_{0}\mathfrak{D}_{1}=\mathcal{I}_{0}\mathfrak{D}_{1} \widetilde{\eH}_{\text{left}}\,,
		\end{equation} 
        		\begin{equation}\label{boundary-duality-right}
			\eH_{\text{right}}^{t}\mathcal{I}_{\N+1}\mathfrak{D}_{\N}=\mathcal{I}_{\N+1}\mathfrak{D}_{\N} \widetilde{\eH}_{\text{right}}\,.
		\end{equation}
	\paragraph{Bulk duality.}
We introduce the diagonal matrix, c.f.~\eqref{cheap-duality},
\begin{equation}
    \Xi=\frac{p_{0}^{2s}}{\Gamma(2s)}\sum_{m\in \mathbb{N}_{0}^{\M}}\Gamma(|m^{\ell}|+2s)\left[\prod_{a=1}^{\M}\frac{\beta_{a}^{m_{a}}}{m_{a}!}\right]|m\rangle \langle m|\,.
\end{equation}
Then, the detailed balance condition of  Lemma~\ref{lemma-reversibleMeasure-M-harmonic-boundaries}, can be written as 
\begin{equation}\label{reversibility-H}
    \eH_{\ell,\ell+1}\Xi_{\ell}\Xi_{\ell+1}=\Xi_{\ell}\Xi_{\ell+1}\eH_{\ell,\ell+1}^{t}\,.
\end{equation}
	From the $gl(\M+1)$ invariance of the R-matrix  $\eR(x)$, one has that
	\begin{equation}\label{symmetry-exp-H}
		\left[\eH_{\ell,\ell+1},\exp{\left(-\eJ_{a,0}^{[\ell]}\right)}\exp{\left(-\eJ_{a,0}^{[\ell+1]}\right)}\right]=0\qquad \forall a\in \{1,\ldots,\M\}\,.
	\end{equation}
	Therefore the bulk duality relation~\eqref{bulk-duality} is equivalent to the relation
	\begin{equation}\label{bulk-duality2}
			\eH_{\ell,\ell+1}^{t}\mathfrak{d}_{\ell}\mathfrak{d}_{\ell+1}=\mathfrak{d}_{\ell}\mathfrak{d}_{\ell+1}\eH_{\ell,\ell+1}\,,
		\end{equation}
        which is implied by~\eqref{reversibility-H}, by choosing  $\beta_{a}=\frac{1}{\M}$ for all $a$, $p_{0}=\frac{1}{\M}$ and by taking the inverse of $\Xi$.

	\paragraph{Boundary duality.} We shall prove only~\eqref{boundary-duality} since the proof of~\eqref{boundary-duality-right} is analogous. 
	To prove the boundary duality relation~\eqref{boundary-duality} we need two intermediate lemmas. 
   \begin{lemma}\label{Lemma-transposition-relation}
                Let $(\eJ_{A,B})_{A,B\in \{0,1,\ldots,\M\}}$ be the $gl(\M+1)$ algebra generators~\eqref{GL-oscRealization}, then the following relations hold
\begin{equation}\label{transposition-relation}
		\begin{cases}
			\mathfrak{d}\eJ_{B,A}\mathfrak{d}^{-1}=-\left(\eJ_{A,B}\right)^{t}\qquad \text{if}&\qquad A=0 \quad\dot \vee \quad B=0\\
			\mathfrak{d}\eJ_{B,A}\mathfrak{d}^{-1}=\left(\eJ_{A,B}\right)^{t}\qquad\;\;\; \text{if}&\qquad \text{otherwise}
		\end{cases}\,.
	\end{equation}
            \end{lemma}
            \begin{proof}
                By a direct computation one shows the relations
\begin{equation}\label{transposition-relation2}
		\begin{cases}
			\langle m|\mathfrak{d}\eJ_{B,A}\mathfrak{d}^{-1}|n\rangle =-\langle n|\eJ_{A,B}|m\rangle \qquad \text{if}&\qquad A=0 \quad\dot \vee \quad B=0\\
			\langle m|\mathfrak{d}\eJ_{B,A}\mathfrak{d}^{-1}|n\rangle =\langle n|\eJ_{A,B}|m\rangle\qquad\;\;\; \text{if}&\qquad \text{otherwise}
		\end{cases}\,.
	\end{equation}
    Then, the lemma follows. 
            \end{proof}    
    Let consider a representation of $gl(\M+1)$ with algebra generators $(\rep(J_{A,B}))_{A,B\in \{0,1,\ldots, \M\}}$ defined in~\eqref{GL-continuous-reps} and acting the space of the polynomials of $\M$ variables $\prod_{a=1}^{\M}\rho_{a}^{m_{a}}$, with $m_{a}\in \mathbb{N}_{0}$.
	\begin{lemma}\label{Lemma-intertwining}
		For all $A,B\in \{0,1,\ldots,\M\}$ the following intertwining relation holds between the algebra generators $\rep(J_{A,B})$  of equation~\eqref{GL-continuous-reps} and $\eJ_{A,B}$ of equation~\eqref{E-matrix}:
		\begin{align}
			\rep(J_{A,B})\mathcal{I}=&\mathcal{I}\eJ_{A,B}\,.
		\end{align}
		Here, $\mathcal{I}$	is the intertwiner introduced in~\eqref{intertwining}, namely
		\begin{align}
			\mathcal{I}=\sum_{m_{1},\ldots,m_{N}=0}^{\infty}\rho_{1}^{m_{1}}\cdots \rho_{\M}^{m_{\M}}\langle m_{1},\ldots,m_{\M}|\,.
		\end{align}
	\end{lemma}
	\begin{proof}
 Using~\eqref{right-action}, the proof consists of a straightforward computation. 
    																																																																															\end{proof}  
    		         
	Consider the algebraic expression of the boundary Hamiltonian~\eqref{right-left-BD-hamiltonian}. 
        Using Lemma~\ref{Lemma-transposition-relation}, we obtain
    \begin{equation}\label{duality-sim1}
			\begin{split}
				\eH_{\text{left}}^{t}=&
				\mathfrak{D}_{1}\exp{\left(-\sum_{a=1}^{\M}\rho_{l,a}\eJ_{0,a}^{[1]}\right)}\left(\psi(-\eJ_{0,0}^{[1]}+s+\frac{1}{2})-\psi(2s)\right)
				\\&\;\cdot
				\exp{\left(\sum_{a=1}^{\M}\rho_{l,a}\eJ_{0,a}^{[1]}\right)}\mathfrak{D}_{1}^{-1}\,,
			\end{split}
		\end{equation}
		where $\mathfrak{D}_1$ is defined in~\eqref{duality-matrix-x}. 								Using the $gl(\M+1)$ generators in the representation $\rep(\cdot)$, see~\eqref{GL-continuous-reps}, we have that
	\begin{equation}\label{rho-continuousReps}
	\rho_{l,a} =\rep(J_{a,0}^{[0]})\left(\rep(J_{0,0}^{[0]})-\frac{1}{2}-s\right)^{-1}
	\end{equation}
	Therefore, inserting~\eqref{rho-continuousReps} into~\eqref{duality-sim1} and using Lemma~\ref{Lemma-intertwining} we obtain
 									\begin{equation}
		\begin{split}
			\eH_{\text{left}}^{t}\mathcal{I}_{0}\mathfrak{D}_{1}=&
            \mathfrak{D}_{1}\exp{\left(-\sum_{a=1}^{\M}\rep(J_{a,0}^{[0]})\left(\rep(J_{0,0}^{[0]})-\frac{1}{2}-s\right)^{-1}\eJ_{0,a}^{[1]}\right)}\left(\psi(-\eJ_{0,0}^{[1]}+s+\frac{1}{2})-\psi(2s)\right)
			\\&\;\cdot
			\exp{\left(\sum_{a=1}^{\M}\rep(J_{a,0}^{[0]})\left(\rep(J_{0,0}^{[0]})-\frac{1}{2}-s\right)^{-1}\eJ_{0,a}^{[1]}\right)}\mathcal{I}_{0}			\\=&
            \mathcal{I}_{0}\mathfrak{D}_{1}\exp{\left(-\sum_{a=1}^{\M}\eJ_{a,0}^{[0]}\left(\eJ_{0,0}^{[0]}-\frac{1}{2}-s\right)^{-1}\eJ_{0,a}^{[1]}\right)}\left(\psi(-\eJ_{0,0}^{[1]}+s+\frac{1}{2})-\psi(2s)\right)
			\\&\;\cdot
			\exp{\left(\sum_{a=1}^{\M}\eJ_{a,0}^{[0]}\left(\eJ_{0,0}^{[0]}-\frac{1}{2}-s\right)^{-1}\eJ_{0,a}^{[1]}\right)}			\\=&
			\mathcal{I}_{0}\mathfrak{D}_{1}\widetilde{\eH}_{\text{left}}\,.
		\end{split}
	\end{equation}
	This concludes the proof.
\end{proof}
\begin{corollary}[Consequences of absorbing duality]\label{Corollary-consequenceDuality}
Given $\bm{\xi}\in \Omega$, consider the scaled  factorial moments of order $(|\xi_{1}|,\ldots,|\xi_{\M}|)$ defined as 
	\begin{equation}\label{scaled-factorial-mom}
		G(|\xi_{1}|,\ldots, |\xi_{\M}|):=\sum_{\bm{m}\in \Omega}\mu(\bm{m})\prod_{\ell=1}^{\N}\frac{(2s)!}{(2s+|\xi^{\ell}|)!}\prod_{a=1}^{\M}m_{a}^{\ell}(m_{a}^{\ell}-1)\cdots (m_{a}^{\ell}-\xi_{a}^{\ell}+1)
				\,.
	\end{equation} 
    where $\mu:\Omega\to [0,1]$ is the  probability mass function of the non-equilibrium steady state of the boundary driven multispecies harmonic process.
Then, \begin{equation}
G(|\xi_{1}|,\ldots,|\xi_{\M}|)=\sum_{j_{1}=0}^{|\xi_{a}|}\cdots\sum_{j_{\M}=0}^{|\xi_{\M}|}\left(\prod_{a=1}^{\M}\rho_{l,a}^{j_{a}}\rho_{r,a}^{|\xi_{a}|-j_{a}}\right)\mathcal{P}_{\hat{\bm{\xi}}}\left(j_{1},\ldots,j_{\M}\right)
\end{equation}
where  $\mathcal{P}_{\hat{\bm{\xi}}}\left(j_{1},\ldots,j_{\M}\right)$ denotes the absorption probability of dual particles, namely 
\begin{equation}
	\mathcal{P}_{\hat{\bm{\xi}}}\left(j_{1},\ldots,j_{\M}\right):=\mathbb{P}\left(\bm{\xi}(\infty)=(j,0,\ldots, 0,|\bm{\xi}|-j)\bigg| \bm{\xi}(0)=\hat{\bm{\xi}}\right)\,.
\end{equation}
Above we denoted  \begin{equation}\label{dualConfig-moments}
	\hat{\bm{\xi}}:=\begin{cases}
		\hat{\xi}_{a}^{\ell}=\xi_{a}^{\ell}\quad &\text{if}\quad \ell\in \{1,\ldots,\N\}\\			\hat{\xi}_{a}^{\ell}=0\quad &\text{else}	\end{cases},\qquad \forall a\in \{1,\ldots,\M\}
\end{equation}
and $(j,0,\ldots, 0,|\bm{\xi}|-j)$ denotes a dual configuration where we placed $j$ dual particles at site $0$, $|\bm{\xi}|-j$ dual particles at site $\N+1$ and no dual particles elsewhere. 
\end{corollary}
\begin{proof}
Consider the initial dual configuration $\bm{\hat{\xi}}$ written in \eqref{dualConfig-moments}. By ergodicity we have that 
\begin{equation}
    G(|\xi_{1}|,\ldots,|\xi_{\M}|)=\lim_{t\to \infty}\mathbb{E}[\mathfrak{F}(\bm{m}(t),\bm{\hat{\xi}})]\,.
\end{equation}
The rest of the proof is the multispecies counterpart of the proof of Proposition 2.6 of \cite{exact-harmonic}, then it is omitted.
			\end{proof}

We now relate the scaled factorial moments $G(|\xi_{1}|,\ldots,|\xi_{\M}|)$ of the non-equilibrium distribution to the ground state $|\Psi'\rangle$ of the matrix $\eH'$ introduced in Proposition~\ref{Proposition-Hp}. To this end, we observe that the similarity transformation $S'$, defined in \eqref{similarities}, satisfies
\begin{equation}
    (S')^{-t}=\exp{\left(\sum_{a=1}^{\M}J_{0,a}^{t}\right)}\,.
\end{equation}
By Lemma~\ref{Lemma-transposition-relation}, together with the identity $\mathfrak{d}^{t}=\mathfrak{d}$, it follows that
\begin{equation}
    (S')^{-t}\mathfrak{d}^{t}=\mathfrak{d}\, \exp{\left(-\sum_{a=1}^{\M}J_{a,0}\right)}=\mathfrak{D}\,.
\end{equation}
It is convenient to introduce the compact notation
\begin{equation}
    \bm{S}':=S'\otimes \cdots \otimes S',
    \qquad
    \bm{\mathfrak{d}}(\bm{S}')^{-1}:= \mathfrak{d}(S')^{-1}\otimes \cdots \otimes \mathfrak{d}(S')^{-1}\,,
\end{equation}
which allows us to express the duality matrix \eqref{duality-matrix} in terms of the transformation $S'$ as
\begin{equation}
    \prod_{\ell=1}^{N}\mathfrak{D}_{\ell}=\big(\bm{\mathfrak{d}}(\bm{S}')^{-1}\big)^{t}\,.
\end{equation}

Let us now initialize the dual process with the configuration $\bm{\hat{\xi}}$ introduced in \eqref{dualConfig-moments}, and set the probability mass function of the non-equilibrium stationary distribution $\mu(\bm{\eta})=\langle \bm{m}|\Psi\rangle$. Then, by Corollary~\ref{Corollary-consequenceDuality}, the scaled factorial moments take the form
\begin{equation}
\begin{split}
    G(|\xi_{1}|,\ldots,|\xi_{\M}|)
    &=\sum_{\bm{\eta}\in \Omega}\langle \bm{\eta}|\prod_{\ell=1}^{N}\mathfrak{D}_{\ell}|\bm{\hat{\xi}}\rangle\, \langle \bm{\eta}|\Psi\rangle \\[2pt]
    &=\sum_{\bm{\eta}\in \Omega}\langle\bm{\hat{\xi}}|\bm{\mathfrak{d}}(\bm{S}')^{-1}|\bm{\eta}\rangle\, \langle \bm{\eta}|\Psi\rangle \\[2pt]
    &= \langle \bm{\hat{\xi}} |\bm{\mathfrak{d}}|\Psi^{'}\rangle\,,
\end{split}
\end{equation}
which establishes the correspondence between the scaled factorial moments and the ground state $|\Psi'\rangle$.
\subsection{Duality for the hidden parameter model}\label{subsection-hiddenDuality}

 We begin this section by writing the algebraic expression of the hidden parameter model. Consider the representation of $gl(\M+1)$ on the polynomial space defined in~\eqref{rep-infiniteDim-Poly}, with algebra generators $\bm{\Hidden}=(\hidden_{1},\ldots,\hidden_{\M})$ and $\bm{\partial}_{\hidden}=(\partial_{\hidden_1},\ldots,\partial_{\hidden_\M})^{t}$. Let $f:\mathbb{R}^{\M\times \N}\to \mathbb{R}$ be a polynomial function, we introduce the Hamiltonian operator
 \begin{equation}\label{algebraic-Hidden}
 		\ecH f(\bm{\hidden}) =\ecH _{\text{left}}f(\bm{\hidden})+\sum_{\ell=1}^{\N-1}\ecH_{\ell,\ell+1}f(\bm{\hidden})+\ecH _{\text{right}}f(\bm{\hidden})\,,
 \end{equation}
  where $\ecH_{\ell,\ell+1}$ coincides with~\eqref{eq:inthamm}, while the boundary operators read
 \begin{equation}\label{hleftt4-alg}
 	\begin{split}
 		&\ecH _{\text{left}}f(\bm{\hidden})=\exp{\left(-\rho_{l}\bm{\partial}^{[1]}_{\hidden}\right)}\left(\psi(\bm{\Hidden}^{[1]}\bm{\partial}_{\hidden}^{[1]}+2s)-\psi(2s)\right)
 		\exp{\left(\rho_{l}\bm{\partial}^{[1]}_{\hidden}\right)}f(\bm{\hidden})
          	\end{split}
 \end{equation}
 and 
 \begin{equation} 	\begin{split}
 		&\ecH _{\text{right}}f(\bm{\hidden})=\exp{\left(-\rho_{r}\bm{\partial}^{[\N]}_{\hidden}\right)}\left(\psi(\bm{\Hidden}^{[\N]}\bm{\partial}^{[\N]}_{\hidden}+2s)-\psi(2s)\right)
 		\exp{\left(\rho_{r}\bm{\partial}^{[\N]}_{\hidden}\right)}f(\bm{\hidden})\,.
          	\end{split}
 \end{equation}
 Above we denoted $\rho_{\ell}\bm{\partial}_{\hidden}=\sum_{a=1}^{\M}\rho_{l,a}\partial_{\hidden_a}$. 
  \begin{lemma}
     Let $\ecL$ be the generator~\eqref{generator-Hidden} of the hidden parameter model. Then, when acting on polynomial functions $f:\mathbb{R}_{+}^{\N\times \M}\to \mathbb{R}$, we have that 
    \begin{equation}
        \ecH f(\bm{\hidden})=-\ecL f(\bm{\hidden})\,.
    \end{equation}
 \end{lemma}
 \begin{proof}
     For the bulk part we refer to Corollary~\ref{cor:int}. For the boundary the proof directly follows from~\eqref{shift-formulas}. 
 \end{proof}

\begin{proposition}\label{Proposition_Hidden-ABS}
The hidden parameter model $(\bm{\hidden}(t))_{t\geq 0}$ with generator~\eqref{generator-Hidden} is dual to the absorbing harmonic process $(\bm{\xi}(t))_{t\geq 0}$ with generator~\eqref{abs-Harmonic-Generator}, with the following duality function:
\begin{equation}\label{duality-function-hidden}
	\mathbb{D}(\bm{\hidden},\bm{\xi})=\left(\prod_{a=1}^{\M}(\rho_{l,a})^{\xi_{a}^{0}}\right)\left(\prod_{\ell=1}^{\N}\prod_{a=1}^{\M}(\hidden_{a}^{\ell})^{\xi_{a}^{\ell}}\right)\left(\prod_{a=1}^{\M}(\rho_{r,a})^{\xi_{a}^{\N+1}}\right)\,.
\end{equation}
 \end{proposition}
\begin{proof} 
We aim to show that 
\begin{equation}
   \left( \ecL \mathbb{D}(\cdot,\bm{\xi})\right)(\bm{\hidden})=\left(\eL\mathbb{D}(\bm{\hidden},\cdot)\right)(\bm{\xi})\qquad \forall \bm{\hidden}\in\mathbb{R}^{\M\times \N},\,\bm{\xi}\in \widetilde{\Omega}\,.
   \end{equation}
By using~\eqref{gen-ham1}, we reformulate the above duality relation via the corresponding Hamiltonian operators
\begin{equation}\label{duality-hidden-algebraic}
    \left( \ecH\mathbb{D}(\cdot,\bm{\xi})\right)(\bm{\hidden})=\sum_{\bm{\xi}'\in \widetilde{\Omega}}\langle \bm{\xi}'|\mathbb{D}(\bm{\hidden},\bm{\xi}')\widetilde{\eH}|\bm{\xi}\rangle\qquad \forall\,\bm{\hidden}\in \mathbb{R}^{\M\times \N},\,\bm{\xi}\in \widetilde{\Omega}\,,
\end{equation}
where $\widetilde{\eH}$ is the Hamiltonian of the absorbing dual process~\eqref{Hamiltonian-ABS-Harmonic} and $\ecH$ is the operator defined in~\eqref{algebraic-Hidden}. 
The duality relation~\eqref{duality-hidden-algebraic} then follows form equation~\eqref{intertw} and Lemma~\ref{Lemma-intertwining}.

\end{proof}
  \begin{corollary}
  The Hamiltonian $\ecH$ of the hidden parameter model and the Hamiltonian $\eH^{t}$ of the boundary driven harmonic model satisfy the intertwining relation
	\begin{equation}
			\begin{split}
				\hmH \left(\sum_{\bm{\xi}} \mathbb{D}(\bm{\hidden},\bm{\xi})\langle \bm{\xi}|\mathfrak{D}^{-1}\right) =\left(\sum_{\bm{\xi}} \mathbb{D}(\bm{\hidden},\bm{\xi})\langle \bm{\xi}|\mathfrak{D}^{-1}\right) \eH^t\qquad \forall\, \bm{\hidden}\in \mathbb{R}^{\M\times \N}\,.
			\end{split}
		\end{equation}
        Above, on the left-hand-side $\ecH$ acts on the variable $\bm{\hidden}$. 
\end{corollary}
\begin{proof}
Let $\mathfrak{D}= \prod_{\ell=1}^{L}\mathfrak{D}_{\ell}$, where $\mathfrak{D}_{\ell}$ is defined in~\eqref{duality-matrix-x}. Then, working as in the proof of Theorem~\ref{Thm-duality} and using~\eqref{good-relation-forFockGLN}, we have
 \begin{equation}\label{similarity-hidden}
 \begin{split}
     	\mathfrak{D}^{-1}\eH^t \mathfrak{D}=&\exp{\left(-\rho_{l}\OA^{[1]}
                  \right)}\left(\psi(\OAD^{[1]}\OA^{[1]}+2s)-\psi(2s)\right)
 		\exp{\left(\rho_{l}\OA^{[1]}
                  \right)}
         \\&+\sum_{\ell=1}^{\N-1}\eH_{\ell,\ell+1}
         \\&+
         \exp{\left(-\rho_{r}\OA^{[\N]}
                  \right)}\left(\psi(\OAD^{[\N]}\OA^{[\N]}+2s)-\psi(2s)\right)
 		\exp{\left(\rho_{r}\OA^{[\N]}
                  \right)}\,.
 \end{split}
 \end{equation}
Using the~\eqref{duality-sim1}, Lemma~\ref{Lemma-intertwining} and the duality relation~\eqref{duality-hidden-algebraic}, we have
\begin{equation}
    \begin{split}
        \left(\ecH \mathbb{D}(\cdot,\bm{\xi})\right)(\bm{\hidden})=&\sum_{\bm{\xi}'\in\widetilde{\Omega}}\langle \bm{\xi}'|\mathcal{D}^{-1}\eH^{t}\mathcal{D}|\bm{\xi}\rangle\mathbb{D}(\bm{\hidden},\bm{\xi}')\,,    \end{split}
\end{equation}
which implies 
\begin{equation}
    \begin{split}
          \ecH\left(\sum_{\bm{\xi}\in \widetilde{\Omega}}\langle\bm{\xi}|\mathbb{D}(\bm{\hidden},\bm{\xi})\mathcal{D}^{-1}\right)=&\left(\sum_{\bm{\xi}'\in\widetilde{\Omega}}\langle \bm{\xi}'|\mathcal{D}^{-1}\mathbb{D}(\bm{\hidden},\bm{\xi}')\right)\eH^{t}\,.
    \end{split}
\end{equation}
\end{proof}
\subsection{Duality for the heat conduction model}
First we introduce the algebraic form for the Hamiltonian of the boundary driven integrable heat conduction model. Consider the $gl(\M+1)$ representation on the polynomial space defined in~\eqref{bargen}, with algebra generators $\bm{\Heat}=(\heat_{1},\ldots,\heat_{\M})^{t}$ and $\bm{\partial}_{\heat}=(\partial_{\heat_1},\ldots,\partial_{\heat_\M})$.  Let $f:\mathbb{R}^{\M\times \N}\to \mathbb{R}$ be a polynomial function, we introduce the Hamiltonian operator 
    \begin{equation}\label{Heat-Hamiltonian}
	\ecdH f(\bm{\heat})=\ecdH_{\text{left}}f(\bm{\heat})+\sum_{\ell=1}^{\N-1}\ecdH_{\ell,\ell+1}f(\bm{\heat})+\ecdH_{\text{right}}f(\bm{\heat})\,,
\end{equation}
where the bulk part decomposes in
\begin{equation}
    \ecdH_{\ell,\ell+1}f(\bm{\heat})=\ecdH_{\ell,\ell+1}^{-}f(\bm{\heat})+\ecdH_{\ell,\ell+1}^{+}f(\bm{\heat})
\end{equation}
 \begin{equation}
  \begin{split}
	&\ecdH_{\ell,\ell+1}^{-}f(\bm{\heat})=\exp{\left(\bm{\partial}^{[\ell+1]}_{\heat}\bm{\Heat}^{[\ell]}\right)}
       \left(\psi(\bm{\Heat}^{[\ell]}\bm{\partial}^{[\ell]}_{\heat}+2s)-\psi(2s)\right)\exp{\left(-\bm{\partial}^{[\ell+1]}_{\heat}\bm{\Heat}^{[1]}\right)}f(\bm{\heat})\,,
      \end{split}
\end{equation}
and 
  \begin{equation}
  \begin{split}
	&\ecdH_{\ell,\ell+1}^{+}f(\bm{\heat})=\exp{\left(\bm{\partial}^{[\ell]}_{\heat}\bm{\Heat}^{[\ell+1]}\right)}
        \left(\psi(\bm{\Heat}^{[\ell+1]}\bm{\partial}^{[\ell+1]}_{\heat}+2s)-\psi(2s)\right)\exp{\left(-\bm{\partial}^{[\ell]}_{\heat}\bm{\Heat}^{[\ell+1]}\right)}f(\bm{\heat})
        \end{split}
\end{equation}
Moreover, the left and right boundary read, respectively,
\begin{equation}
\begin{split}
	&\ecdH_{\text{left}}f(\bm{\heat})
    \\=&
      \exp{\left(-    \left(2s+\bm{\Heat}^{[1]}\bm{\partial}^{[1]}_{\heat}
        \right)
    \rho_{l}\bm{\partial}^{[1]}_{\heat}\right)}
    		\left(\psi(\bm{\Heat}^{[1]}\bm{\partial}^{[1]}_{\heat}+2s)
        -\psi(2s)\right)
	    \exp{\left(    \left(2s+\bm{\Heat}^{[1]}\bm{\partial}^{[1]}_{\heat}
        \right)
    \rho_{l}\bm{\partial}^{[1]}_{\heat}\right)}f(\bm{\heat})\,,
	    \end{split}
\end{equation}
\begin{equation}
\begin{split}
	&\ecdH_{\text{right}}f(\bm{\heat})
    \\=&
      \exp{\left(-    \left(2s+\bm{\Heat}^{[\N]}\bm{\partial}^{[\N]}_{\heat}
        \right)
    \rho_{r}\bm{\partial}^{[\N]}_{\heat}\right)}
    		\left(\psi(\bm{\Heat}^{[\N]}\bm{\partial}^{[\N]}_{\heat}+2s)
        -\psi(2s)\right)
	    \exp{\left(    \left(2s+\bm{\Heat}^{[\N]}\bm{\partial}^{[\N]}_{\heat}
        \right)
    \rho_{r}\bm{\partial}^{[\N]}_{\heat}\right)}f(\bm{\heat})\,.
	    \end{split}
\end{equation}
Above, we denoted $\rho_{l}\bm{\partial}_{\heat}=\sum_{a=1}^{\M}\rho_{l,a}\partial_{\heat_{a}}$. 
\begin{lemma}
  Let $\ecdL$ be the generator~\eqref{heatGenerato} of the boundary driven integrable heat conduction model. Then, when acting on functions $f:\mathbb{R}_{+}^{\N\times \M}\to \mathbb{R}$, we have that 
   \begin{equation}
       \ecdH f(\bm{\heat})=-\ecdL f(\bm{\heat})\,.
   \end{equation}
\end{lemma}
\begin{proof}
    We work in complete analogy with \cite{Franceschini:2022vmr}. Using the integral representation, it is straightforward to verify that in the bulk \begin{equation}
        \ecdH_{\ell,\ell+1}f(\heat^{\ell},\heat^{\ell+1})=-\ecdL_{\ell,\ell+1}f(\heat^{\ell},\heat^{\ell+1})\,.
    \end{equation}    On the left boundary we have that
		\begin{equation}
			\ecdH_\text{left}f(\heat^{1})=\left(\psi\left(\sum_{a=1}^{\M}\heat_{a}\partial_{\heat_{a}}+2s\right)-\psi(2s)\right)f(\heat^{1})+\log \left(1-\sum_{a=1}^\M \rho_{l,a}\partial_{\heat_{a}}\right)f(\heat^{1})\,.
		\end{equation}
												The integral representation then becomes
		\begin{equation}\label{integral-reps-leftHeat}
			\ecdH_\text{left}f(\heat^{1})=\int_0^1 \frac{\alpha^{2s-1}}{1-\alpha}(f(\heat^{1})-f(\alpha \heat^{1})) dt+\int_{0}^\infty\frac{d\alpha}{\alpha}e^{-\alpha}\left(f(\heat^{1})- f(\heat^{1}+\alpha\rho_{l})\right)\,.
		\end{equation}
		where we used
		\begin{equation}
			\left[ \psi(\bm{\heat}\bm{\partial}_{\heat}+2s)-\psi(2s)\right]f(\heat)=\int_0^1 d\alpha \frac{\alpha^{2s-1}}{1-\alpha}(f(\heat)-f(\alpha \heat)) 
		\end{equation}
		and
		\begin{equation}\label{log-Mercator}
			\log \left(1-\sum_{a=1}^\M \rho_{l,a}\partial_{\heat_{a}}\right)f(\heat)=\int_{0}^\infty\frac{d\alpha}{\alpha}e^{-\alpha}\left(f(\heat)- f(\heat+\alpha{\rho}_{l})\right)
		\end{equation}
		Equation~\eqref{log-Mercator} follows from the Taylor series expansion.
								 The equality on the right boundary can be proven similarly.
\end{proof}
\begin{proposition}\label{Proposition_ABS-Heat_Duality}
	The boundary driven integrable heat conduction model $(\bm{\heat}(t))_{t\geq 0}$ is dual to the absorbing harmonic process $(\bm{\xi}(t))_{t\geq 0}$ with the following duality function:
	\begin{equation}
		\mathscr{D}(\bm{\heat},\bm{\xi})=\left(\prod_{a=1}^{\M}\rho_{l,a}^{\xi_{a}^{0}}\right)\left(\prod_{\ell=1}^{\N}\frac{\Gamma(2s)}{\Gamma(2s+|\xi^{\ell}|)}\prod_{a=1}^{\M}(\heat_{a}^{\ell})^{\xi_{a}^{\ell}}\right)\left(\prod_{a=1}^{\M}\rho_{r,a}^{\xi_{a}^{\N+1}}\right)\,.
	\end{equation}
				\end{proposition}

\begin{proof}
We aim to show that 
\begin{equation}
\left(\ecdL\mathscr{D}(\cdot,\bm{\xi})\right)(\bm{\heat})   = \left(\widetilde{\eL}\mathscr{D}(\bm{\heat},\cdot)\right)(\bm{\xi})\qquad \forall\, \bm{\heat}\in\mathbb{R}^{\M\times \N},\,\bm{\xi}\in \widetilde{\Omega}\,.
\end{equation}
Using~\eqref{gen-ham1}, we reformulate the duality above as a relation between Hamiltonian operators as
\begin{equation}
	\begin{split}
&\left(\ecdH\mathscr{D}(\cdot,\bm{\xi})\right)(\bm{\heat})=
\sum_{\tilde{\bm{\xi}}}\langle \tilde{\bm{\xi}}|\mathscr{D}(\bm{\heat},\tilde{\bm{\xi}})\widetilde{\eH}|\bm{\xi}\rangle\qquad \forall\, \bm{\heat}\in\mathbb{R}^{\M\times \N},\,\bm{\xi}\in \widetilde{\Omega}\,.
\end{split}
\end{equation}
Consider the bulk duality matrix and the cheap duality matrix defined in~\eqref{duality-matrix-x} and in~\eqref{cheap-duality} respectively and consider the diagonal matrix $Q$ defined in~\eqref{Q-matrix}. We introduce
\begin{equation}
    \mathcal{D}_{\ell}:= \mathfrak{D}_{\ell}\mathfrak{d}_{\ell}^{-1}Q_{\ell}=\exp{\left((J_{0,a}^{[\ell]})^{t}\right)}Q_{\ell}\,.
\end{equation}
The last equality follows from~\eqref{transposition-relation}. Then, taking into account the whole lattice, we define
 \begin{equation}\label{similarity-heat}
\mathcal{D}:=\prod_{\ell=1}^{\N}\mathcal{D}_{\ell}\,.
\end{equation} 
We further introduce the matrix $\bar\eH$ 
\begin{equation}\label{similarity-relation-bar}
	\bar{\eH}:=\mathcal{D}^{-1}\eH^{t} \mathcal{D}\,,
\end{equation}
which decomposes into bulk and boundary matrices as 
\begin{equation}\label{bar-hamiltonian}
	\bar{\eH}=\bar{\eH}_{\text{left}}+\sum_{\ell=1}^{\N-1}\bar{\eH}_{\ell,\ell+1}+\bar{\eH}_{\text{right}}\,.
\end{equation}
Using the $gl(M+1)$ symmetry of the bulk Hamiltonian and the transposition relation for the oscillators~\eqref{transpositoin-oscillator}, we have that, on the bulk, 
\begin{equation}
    \bar{\eH}_{\ell,\ell+1}=\bar{\eH}^{-}_{\ell,\ell+1}+\bar{\eH}^{+}_{\ell,\ell+1}
\end{equation}
where 
\begin{equation}\label{Rm-solutionDss}
\begin{split}
    \bar{\eH}^{-}=       e^{\OA^{[2]}\OAD^{[1]}}\left[\psi(\OAD^{[1]}\OA^{[1]}+2s)-\psi(2s)\right]e^{-\OA^{[2]}\OAD^{[1]}}
    \end{split}
\end{equation}
and 
		\begin{equation}\label{Rp-solutionDss}
\bar{\eH}^{+} = e^{\OA^{[1]}\OAD^{[2]}}\left[\psi(\OAD^{[2]}\OA^{[2]}+2s)-\psi(2s)\right]e^{-\OA^{[1]}\OAD^{[2]}}\,.
\end{equation}
On the boundaries one obtains
	\begin{equation}
	\begin{split}
	&\bar{\eH}_{\text{left}}=							\\&
        \exp{\left(-\sum_{a=1}^{\M}\rho_{l,a}  \left(2s+\OAD^{[1]}\OA^{[1]}\right)\oa_{a}^{[1]}\right)}
				\left(\psi(\OAD^{[1]}\OA^{[1]}+2s)-\psi(2s)\right)
							\exp{\left(\sum_{a=1}^{\M}\rho_{l,a}  \left(2s+\OAD^{[1]}\OA^{[1]}\right)\oa_{a}\right)}\,.
				 						 						 						 					\end{split}
\end{equation}
and 
	\begin{equation}
	\begin{split}
	&\bar{\eH}_{\text{right}}=\\
		&
        \exp{\left(-\sum_{a=1}^{\M}\rho_{l,a}  \left(2s+\OAD^{[\N]}\OA^{[\N]}\right)\oa_{a}^{[\N]}\right)}		\left(\psi(\OAD^{[\N]}\OA^{[\N]}+2s)-\psi(2s)\right)
		\exp{\left(\sum_{a=1}^{\M}\rho_{l,a}  \left(2s+\OAD^{[\N]}\OA^{[\N]}\right)\oa_{a}^{[\N]}\right)}\,.
	\end{split}
\end{equation}
Using the intertwining relations~\eqref{intertw} one has that 
\begin{equation}
	\left(\ecdH \mathbb{D}(\cdot,\bm{\xi})\right)(\bm{\heat})=\sum_{\bm{\xi}'}\langle \bm{\xi}'|\bar{\eH}|\bm{\xi}\rangle \mathbb{D}(\bm{\heat},\bm{\xi}) \,.
\end{equation}
where (c.f.\eqref{duality-function-hidden})
 \begin{equation}
	\mathbb{D}(\bm{\heat},\bm{\xi})=\left(\prod_{a=1}^{\M}(\rho_{l,a})^{\xi_{a}^{0}}\right)\left(\prod_{\ell=1}^{\N}\prod_{a=1}^{\M}(\heat_{a}^{\ell})^{\xi_{a}^{\ell}}\right)\left(\prod_{a=1}^{\M}(\rho_{r,a})^{\xi_{a}^{\N+1}}\right)\,.
\end{equation}
Using the similarity relations~\eqref{similarity-relation-bar}, equation~\eqref{duality-sim1} and Lemma~\eqref{Lemma-intertwining} the equation above may be rewritten as
\begin{equation}\label{HbarHtilde}
	\left(\ecdH\mathbb{D}(\cdot,\bm{\xi})\right)(\bm{\heat})=\sum_{\bm{\xi}'}\langle \bm{\xi}'|\mathbb{D}(\bm{\heat},\bm{\xi})\left(\prod_{\ell=1}^{\N}Q_{\ell}^{-1}\mathfrak{d}_{\ell}\right)\widetilde{\eH}\left(\prod_{\ell=1}^{\N}\mathfrak{d}_{\ell}^{-1}Q_{\ell}\right)|\bm{\xi}\rangle \,.
\end{equation}
Finally, for all $\ell\in \{1,\ldots,\N\}$, we have that 
\begin{equation}
	\mathfrak{d}_{\ell}^{-1}Q_{\ell}=\frac{\Gamma(2s+\OAD^{\ell}\OA^{\ell})}{\Gamma(2s)}\,,
    \end{equation} then we may re-write equation~\eqref{HbarHtilde} as 
\begin{equation}
	\begin{split}
&\left(\ecdH\mathscr{D}(\cdot,\bm{\xi})\right)(\bm{\heat})=
\sum_{\tilde{\bm{\xi}}}\langle \tilde{\bm{\xi}}|\mathscr{D}(\bm{\heat},\tilde{\bm{\xi}})\widetilde{\eH}|\bm{\xi}\rangle
\end{split}\,,
\end{equation}
concluding the proof. 
\end{proof}

\begin{corollary}
 The Hamiltonian $\ecdH$ of the boundary-driven heat conduction model and the Hamiltonian $\eH$ of the boundary driven harmonic model satisfy the following intertwining relation
	\begin{equation}
	 	\begin{split}
	 		\ecdH \left(\sum_{\bm{\xi}} \mathbb{D}(\bm{\heat},\bm{\xi})\langle \bm{\xi}|\mathcal{D}^{-1}\right) =\left(\sum_{\bm{\xi}} \mathbb{D}(\bm{\heat},\bm{\xi})\langle \bm{\xi}|\mathcal{D}^{-1}\right) \eH^t\qquad  \forall \,\bm{\heat}\in \mathbb{R}^{\N\times \M}.
	 	\end{split}
	 \end{equation}
    Above, on the left-hand-side $\ecdH$ acts on the variable $\bm{\heat}$. 

\end{corollary}
\begin{proof}
    The proof is a direct consequence of Proposition~\ref{Proposition_ABS-Heat_Duality}. 
					     \end{proof}

\vspace{1cm}
\appendix
\noindent{\huge \textbf{Appendix}}
\section{Useful formulas for the Fock representation of the Lie algebra}
In this section we collect some formulas that are useful for calculations. 
\subsection{Action of the \textit{gl(M+1)} algebra generators}
Consider the $gl(\M+1)$ representation with algebra generators~\eqref{E-matrix}. For the sake of notation we denote $\eJ_{A,B}=\pi(J_{A,B})$. Using the action of the oscillators on the Fock state vector written in equation~\eqref{discrete-representation} we have that 
\begin{equation}
\begin{split}
	&\eJ_{0,0}|m\rangle=(\mu_{1}-|m|)|m\rangle\,, \\
    &\eJ_{a,b}|m\rangle =(\mu_2\delta_{ab}+m_{b})|m_{1},\ldots,m_{a}+1,\ldots,m_{b}-1,\ldots,m_{\M}\rangle\,, \\
	&\eJ_{a,0}|m\rangle =(\mu_1-\mu_2-|m|)|m_{1},\ldots,m_{a}+1,\ldots, m_{N}\rangle \,,	\\
	&\eJ_{0,a}|m\rangle = m_{a}|m_{1},\ldots,m_{a}-1,\ldots,m_{\M}\rangle\,, \\
          \end{split}
\end{equation}
where $a,b\in \{1,\ldots,\M\}$.

The action of the oscillators $\oad_{a}$ and $\oa_{a}$ on bra-vectors read
\begin{equation}\label{action-oscillator-bra}
\begin{split}
    &\langle m|\oa_{a}=(n_{a}+1)\langle m_{1},\ldots,m_{a}+1,\ldots,m_{\M} |\,,\\
    &\langle m|\oad_{a}=\langle m_{1},\ldots, m_{a}-1,\ldots, m_{\M}|\,,
    \end{split}
\end{equation}
Therefore, we compute the action of the $J_{A,B}$ on bra-vectors $\langle m|$, obtaining, for $a,b\in \{1,2,\ldots,\M\}$
\begin{equation}\label{right-action}
    \begin{split}
        &\langle m|\eJ_{0,0}=\left(\mu_1-|m|\right)\langle m|\,,\\
        &\langle m|\eJ_{a,b}=\left(\mu_2+m_{a}\right)\langle m|\\
        &\langle m|\eJ_{a,0}=(\mu_1-\mu_2-|m|+1)\langle m_{1},\ldots,m_{a}-1,\ldots,m_{\M}|\;,\\
        &\langle m|\eJ_{0,a}=(m_{a}+1)\langle m_{1},\ldots,m_{a}+1,\ldots,m_{\M}|\,,\\
        &\langle m|\eJ_{a,b}=(m_{b}+1)\langle m_{1},\ldots,m_{a}-1,\ldots,m_{b}+1,\ldots,m_{\M}|
    \end{split}
\end{equation}
																																										        Consider the diagonal matrix $Q$ defined in~\eqref{Q-matrix}, then the following relation hold for all $a\in \{1,\ldots,\M\}$:
        	\begin{equation}\label{transpositoin-oscillator}
	Q\oa_{a}Q^{-1}=\oad_{a}^{t}\,,\qquad  Q\oad_{a}Q^{-1}=\oa_{a}^{t}\,.
\end{equation}
Above $\oad^{t}$ and $\oa^{t}$ denote the transposed of the matrices $\oad$ and $\oa$ respectively. 
 									\subsection{Component of matrix-exponentials}
Let consider the algebra generators~\eqref{E-matrix} in the Fock representation. We report here some formulas concerning the components of some exponential of matrices.
We have that 
\begin{align}
	\exp{\left(\sum_{a=1}^{\M}\alpha_{a}J_{0,a}\right)}|m\rangle=&
	\sum_{k_{1}=0}^{m_{1}}\cdots\sum_{k_{\M}=0}^{m_{\M}}\mathbbm{1}_{\{|k|>0\}}\left[\prod_{a=1}^{\M}\alpha_{a}^{k_{a}}\binom{m_{a}}{k_{a}}\right]|m-k\rangle\,,
\end{align}
such that
\begin{align}\label{exponential-annhilation}
	\langle n|\exp{\left(\sum_{a=1}^{\M}\alpha_{a}J_{0,a}\right)}&|m\rangle=\prod_{a=1}^{\M}\alpha_{a}^{m_a-n_a}\binom{m_{a}}{m_{a}-n_a} \,.
\end{align}
Moreover, we have that 
\begin{align}
	\exp{\left(-\sum_{a=1}^{\M}\alpha_{a}J_{a,0}\right)}|m\rangle=
	\sum_{k_{1}=0}^{\infty}\cdots\sum_{k_{\M}=0}^{\infty}\left[\prod_{a=1}^{\M}\frac{\alpha_{a}^{k_{a}}}{k_{a}! }\right]\frac{\Gamma(|m|+|k|+\lambda_2-\lambda_1)}{\Gamma(|m|+\lambda_2-\lambda_1)}|m+k\rangle\,,
\end{align}
such that
\begin{align}\label{mp1right}
	\langle n|\exp{\left(-\sum_{a=1}^{\M}\alpha_{a}J_{a,0}\right)}|m\rangle=\left[\prod_{a=1}^{\M}\frac{\alpha_{a}^{n_{a}-m_a}}{(n_{a}-m_a)! }\right]\frac{\Gamma(|n|+\lambda_2-\lambda_1)}{\Gamma(|m|+\lambda_2-\lambda_1)}\,.
 \end{align}

\bibliography{refs}

@article{Bosnjak:2016oze,
    author = "Bosnjak, Gary and Mangazeev, Vladimir V.",
    title = "{Construction of $R$-matrices for symmetric tensor representations related to ${U}_{q}(\widehat{{sl}_n})$}",
    eprint = "1607.07968",
    archivePrefix = "arXiv",
    primaryClass = "math-ph",
    doi = "10.1088/1751-8113/49/49/495204",
    journal = "J. Phys. A",
    volume = "49",
    number = "49",
    pages = "495204",
    year = "2016"
}

@article{Lipatov:1993yb,
    author = "Lipatov, L. N.",
    title = "{Asymptotic behavior of multicolor QCD at high energies in connection with exactly solvable spin models}",
    eprint = "hep-th/9311037",
    archivePrefix = "arXiv",
    reportNumber = "DFPD-93-TH-70",
    journal = "JETP Lett.",
    volume = "59",
    pages = "596--599",
    year = "1994"
}

@article{8203134,
    author = "Corwin, Ivan",
    title = "{The q-Hahn Boson Process and q-Hahn TASEP}",
    eprint = "1401.3321",
    archivePrefix = "arXiv",
    primaryClass = "math.PR",
    doi = "10.1093/imrn/rnu094",
    journal = "Int. Math. Res. Not.",
    volume = "2015",
    number = "14",
    pages = "5577--5603",
    year = "2015"
}

@article{Faddeev:1994zg,
    author = "Faddeev, L. D. and Korchemsky, G. P.",
    title = "{High-energy QCD as a completely integrable model}",
    eprint = "hep-th/9404173",
    archivePrefix = "arXiv",
    reportNumber = "ITP-SB-94-14",
    doi = "10.1016/0370-2693(94)01363-H",
    journal = "Phys. Lett. B",
    volume = "342",
    pages = "311--322",
    year = "1995"
}

@article{Braun:1999te,
    author = "Braun, Vladimir M. and Derkachov, Sergey E. and Korchemsky, G. P. and Manashov, A. N.",
    title = "{Baryon distribution amplitudes in QCD}",
    eprint = "hep-ph/9902375",
    archivePrefix = "arXiv",
    reportNumber = "LPT-ORSAY-98-83, NORDITA-99-11-HE, SPBU-IP-99-04",
    doi = "10.1016/S0550-3213(99)00265-5",
    journal = "Nucl. Phys. B",
    volume = "553",
    pages = "355--426",
    year = "1999"
}

@incollection{schutz2000exactly,
    author = "Schütz, Gunter M.",
    title = "{Exactly Solvable Models for Many-Body Systems Far from Equilibrium}",
    editor = "Domb, C. and Lebowitz, J. L.",
    series = "Phase Transitions and Critical Phenomena",
    volume = "19",
    publisher = "Academic Press",
    pages = "1--251",
    year = "2001",
    issn = "1062-7901",
    doi = "10.1016/S1062-7901(01)80015-X"
}

@article{Frassek:2022fjs,
    author = "Frassek, Rouven",
    title = "{Integrable boundaries for the q-Hahn process}",
    eprint = "2205.10512",
    archivePrefix = "arXiv",
    primaryClass = "math-ph",
    doi = "10.1088/1751-8121/ac901b",
    journal = "J. Phys. A",
    volume = "55",
    number = "40",
    pages = "404008",
    year = "2022"
}

@article{2018arXiv180801866B,
    author = "Borodin, Alexei and Wheeler, Michael",
    title = "{Coloured stochastic vertex models and their spectral theory}",
    eprint = "1808.01866",
    archivePrefix = "arXiv",
    primaryClass = "math.PR",
    doi = "10.24033/ast.1180",
    journal = "Ast\'erisque",
    number = "437",
    year = "2022"
}

@article{crampe2016matrix,
    author = "Cramp\'e, Nicolas and Evans, M. R. and Mallick, K. and Ragoucy, E. and Vanicat, M.",
    title = "{Matrix product solution to a 2-species TASEP with open integrable boundaries}",
    eprint = "1606.08148",
    archivePrefix = "arXiv",
    doi = "10.1088/1751-8113/49/47/475001",
    journal = "J. Phys. A",
    volume = "49",
    number = "47",
    pages = "475001",
    year = "2016"
}

@article{crampe2016integrable,
    author = "Cramp\'e, Nicolas and Finn, C. and Ragoucy, E. and Vanicat, M.",
    title = "{Integrable boundary conditions for multi-species ASEP}",
    eprint = "1606.01018",
    archivePrefix = "arXiv",
    primaryClass = "math-ph",
    doi = "10.1088/1751-8113/49/37/375201",
    journal = "J. Phys. A",
    volume = "49",
    number = "37",
    pages = "375201",
    year = "2016"
}

@article{Mallick,
    author = "Golinelli, Olivier and Mallick, Kirone",
    title = "{The asymmetric simple exclusion process: an integrable model for non-equilibrium statistical mechanics}",
    eprint = "cond-mat/0611701",
    archivePrefix = "arXiv",
    primaryClass = "cond-mat.stat-mech",
    doi = "10.1088/0305-4470/39/41/S03",
    journal = "J. Phys. A",
    volume = "39",
    number = "41",
    pages = "12679",
    year = "2006"
}

@article{blythe2007nonequilibrium,
    author = "Blythe, R. A. and Evans, M. R.",
    title = "{Nonequilibrium steady states of matrix-product form: a solver's guide}",
    eprint = "0706.1678",
    archivePrefix = "arXiv",
    primaryClass = "cond-mat.stat-mech",
    doi = "10.1088/1751-8113/40/46/R01",
    journal = "J. Phys. A",
    volume = "40",
    number = "46",
    pages = "R333--R441",
    year = "2007"
}

@article{Derrida,
    author = "Derrida, Bernard",
    title = "{Non-equilibrium steady states: fluctuations and large deviations of the density and of the current}",
    eprint = "cond-mat/0703762",
    archivePrefix = "arXiv",
    primaryClass = "cond-mat.stat-mech",
    doi = "10.1088/1742-5468/2007/07/P07023",
    journal = "J. Stat. Mech.",
    volume = "2007",
    pages = "P07023",
    year = "2007"
}

@article{NJMacKay_1991,
    author = "MacKay, N. J.",
    title = "{Rational R-matrices in irreducible representations}",
    doi = "10.1088/0305-4470/24/17/018",
    journal = "J. Phys. A",
    volume = "24",
    number = "17",
    pages = "4017",
    year = "1991"
}

@article{Derkachov:2005hw,
    author = "Derkachov, Sergey E.",
    title = "{Factorization of the R-matrix. I.}",
    eprint = "math/0503396",
    archivePrefix = "arXiv",
    primaryClass = "math.QA",
    doi = "10.1007/s10958-007-0164-8",
    journal = "J. Math. Sci.",
    number = "143",
    pages = "2773--2790",
    year = "2007"
}

@article{Bytsko:2001uh,
    author = "Bytsko, Andrei G.",
    title = "{On integrable Hamiltonians for higher spin XXZ chain}",
    eprint = "hep-th/0112163",
    archivePrefix = "arXiv",
    doi = "10.1063/1.1591054",
    journal = "J. Math. Phys.",
    volume = "44",
    pages = "3698--3717",
    year = "2003"
}

@article{borodin2014duality,
    author = "Borodin, Alexei and Corwin, Ivan and Sasamoto, Tomohiro",
    title = "{From duality to determinants for q-TASEP and ASEP}",
    eprint = "1207.5035",
    archivePrefix = "arXiv",
    primaryClass = "math.PR",
    doi = "10.1214/13-AOP868",
    journal = "Ann. Probab.",
    volume = "42",
    number = "6",
    pages = "2314--2382",
    year = "2014"
}

@article{Mangazeev:2014gwa,
    author = "Mangazeev, Vladimir V.",
    title = "{On the Yang-Baxter equation for the six-vertex model}",
    eprint = "1401.6494",
    archivePrefix = "arXiv",
    primaryClass = "math-ph",
    doi = "10.1016/j.nuclphysb.2014.02.019",
    journal = "Nucl. Phys. B",
    volume = "882",
    pages = "70--96",
    year = "2014"
}

@article{2019NuPhB.94514665M,
    author = "Mangazeev, Vladimir V. and Lu, Xilin",
    title = "{Boundary matrices for the higher spin six vertex model}",
    eprint = "1903.00274",
    archivePrefix = "arXiv",
    primaryClass = "math-ph",
    doi = "10.1016/j.nuclphysb.2019.114665",
    journal = "Nucl. Phys. B",
    volume = "945",
    pages = "114665",
    year = "2019"
}

@article{Franceschini:2022vmr,
    author = "Franceschini, Chiara and Frassek, Rouven and Giardin\`a, Cristian",
    title = "{Integrable heat conduction model}",
    eprint = "2210.13627",
    archivePrefix = "arXiv",
    primaryClass = "cond-mat.stat-mech",
    doi = "10.1063/5.0138013",
    journal = "J. Math. Phys.",
    volume = "64",
    number = "4",
    pages = "043304",
    year = "2023"
}

@article{deGier:2005zz,
    author = "de Gier, Jan and Essler, Fabian H. L.",
    title = "{Bethe Ansatz Solution of the Asymmetric Exclusion Process with Open Boundaries}",
    eprint = "cond-mat/0508707",
    archivePrefix = "arXiv",
    primaryClass = "cond-mat.stat-mech",
    doi = "10.1103/PhysRevLett.95.240601",
    journal = "Phys. Rev. Lett.",
    volume = "95",
    pages = "240601",
    year = "2005"
}

@article{Srivastava1987NeumannExpansions,
    author = "Srivastava, H. M.",
    title = "{Neumann expansions for a certain class of generalised multiple hypergeometric series arising in physical and quantum chemical applications}",
    doi = "10.1088/0305-4470/20/4/020",
    journal = "J. Phys. A",
    volume = "20",
    number = "4",
    pages = "847--858",
    year = "1987"
}

@article{SPITZER1970246,
    author = "Spitzer, Frank",
    title = "{Interaction of Markov processes}",
    doi = "10.1016/0001-8708(70)90034-4",
    journal = "Adv. Math.",
    volume = "5",
    number = "2",
    pages = "246--290",
    year = "1970"
}

@article{2017CMaPh.352..773G,
    author = "Garbali, Alexandr and de Gier, Jan and Wheeler, Michael",
    title = "{A New Generalisation of Macdonald Polynomials}",
    eprint = "1605.07200",
    archivePrefix = "arXiv",
    primaryClass = "math-ph",
    doi = "10.1007/s00220-016-2818-1",
    journal = "Commun. Math. Phys.",
    volume = "352",
    number = "2",
    pages = "773--804",
    year = "2017"
}

@article{Frassek:2019vjt,
    author = "Frassek, Rouven and Giardin\`a, Cristian and Kurchan, Jorge",
    title = "{Non-compact quantum spin chains as integrable stochastic particle processes}",
    eprint = "1904.01048",
    archivePrefix = "arXiv",
    primaryClass = "math-ph",
    doi = "10.1007/s10955-019-02375-4",
    journal = "J. Stat. Phys.",
    volume = "180",
    pages = "135--171",
    year = "2020"
}

@phdthesis{Bosnjak,
    author = "Bo\v{s}njak, Gary",
    title = "{On solutions to the Yang--Baxter equation related to $sl(n)$}",
    eprint = "1412.3339",
    archivePrefix = "arXiv",
    primaryClass = "hep-th",
    school = "Australian National University",
    doi = "10.25911/5d70f1fd50878",
    year = "2017"
}

@article{Frassek:2019imp,
    author = "Frassek, Rouven",
    title = "{Eigenstates of triangularisable open XXX spin chains and closed-form solutions for the steady state of the open SSEP}",
    eprint = "1910.13163",
    archivePrefix = "arXiv",
    primaryClass = "math-ph",
    doi = "10.1088/1742-5468/ab7af3",
    journal = "J. Stat. Mech.",
    volume = "2005",
    pages = "053104",
    year = "2020"
}

@article{Frassek:2020omo,
    author = "Frassek, Rouven and Giardina, Cristian and Kurchan, Jorge",
    title = "{Duality and hidden equilibrium in transport models}",
    eprint = "2004.12796",
    archivePrefix = "arXiv",
    primaryClass = "cond-mat.stat-mech",
    doi = "10.21468/SciPostPhys.9.4.054",
    journal = "SciPost Phys.",
    volume = "9",
    pages = "054",
    year = "2020"
}

@article{toap,
    author = "Frassek, Rouven and Giardin\`a, Cristian and Redig, Frank and van Tol, Berend",
    title = "{work in progress}"
}

@article{2019LMaPh.109.2049K,
    author = "Kuniba, Atsuo and Okado, Masato and Yoneyama, Akihito",
    title = "{Matrix product solution to the reflection equation associated with a coideal subalgebra of $U_q(A^{(1)}_{n-1})$}",
    eprint = "1812.03767",
    archivePrefix = "arXiv",
    primaryClass = "math-ph",
    doi = "10.1007/s11005-019-01175-x",
    journal = "Lett. Math. Phys.",
    volume = "109",
    number = "9",
    pages = "2049--2067",
    year = "2019"
}

@article{2025JSP...192...21G,
    author = "Giardin\`a, Cristian and Redig, Frank and van Tol, Berend",
    title = "{Intertwining and propagation of mixtures for generalized KMP models and harmonic models}",
    eprint = "2406.01160",
    archivePrefix = "arXiv",
    primaryClass = "math.PR",
    doi = "10.1007/s10955-025-03393-1",
    journal = "J. Stat. Phys.",
    volume = "192",
    number = "2",
    pages = "21",
    year = "2025"
}

@incollection{molev2003yangians,
    author = "Molev, Alexander I.",
    title = "{Yangians and their applications}",
    booktitle = "Handbook of Algebra",
    volume = "3",
    publisher = "Elsevier",
    pages = "907--959",
    year = "2003",
    doi = "10.1016/S1570-7954(03)80058-9"
}

@inproceedings{Faddeev:1996iy,
    author = "Faddeev, L. D.",
    title = "{How algebraic Bethe ansatz works for integrable model}",
    eprint = "hep-th/9605187",
    archivePrefix = "arXiv",
    booktitle = "{Les Houches School of Physics: Astrophysical Sources of Gravitational Radiation}",
    pages = "149--219",
    year = "1996"
}

@article{2009JPhA...42p5004P,
    author = "Prolhac, S. and Evans, M. R. and Mallick, K.",
    title = "{The matrix product solution of the multispecies partially asymmetric exclusion process}",
    eprint = "0812.3293",
    archivePrefix = "arXiv",
    primaryClass = "cond-mat.stat-mech",
    doi = "10.1088/1751-8113/42/16/165004",
    journal = "J. Phys. A",
    volume = "42",
    number = "16",
    pages = "165004",
    year = "2009"
}

@article{Frassek:2019isa,
    author = "Frassek, Rouven",
    title = "{The non-compact XXZ spin chain as stochastic particle process}",
    eprint = "1904.02191",
    archivePrefix = "arXiv",
    primaryClass = "math-ph",
    doi = "10.1088/1751-8121/ab2fb1",
    journal = "J. Phys. A",
    volume = "52",
    number = "33",
    pages = "335202",
    year = "2019"
}

@article{Alcaraz:1992zc,
    author = "Alcaraz, Francisco C. and Droz, Michel and Henkel, Malte and Rittenberg, Vladimir",
    title = "{Reaction - diffusion processes, critical dynamics and quantum chains}",
    eprint = "hep-th/9302112",
    archivePrefix = "arXiv",
    reportNumber = "UGVA-DPT-1992-12-799",
    doi = "10.1006/aphy.1994.1026",
    journal = "Annals Phys.",
    volume = "230",
    pages = "250--302",
    year = "1994"
}

@article{Derkachov:2006fw,
    author = "Derkachov, Sergey E. and Manashov, Alexander N.",
    title = "{R-matrix and Baxter Q-operators for the noncompact SL(N,C) invariant spin chain}",
    eprint = "nlin/0612003",
    archivePrefix = "arXiv",
    primaryClass = "nlin.SI",
    doi = "10.3842/SIGMA.2006.084",
    journal = "SIGMA",
    volume = "2",
    pages = "084",
    year = "2006"
}

@article{Tsuboi:2018gfd,
    author = "Tsuboi, Zengo",
    title = "{On diagonal solutions of the reflection equation}",
    eprint = "1811.10407",
    archivePrefix = "arXiv",
    primaryClass = "math-ph",
    doi = "10.1088/1751-8121/ab0b6d",
    journal = "J. Phys. A",
    volume = "52",
    number = "15",
    pages = "155201",
    year = "2019"
}

@article{lauricella1893sulle,
    author = "Lauricella, Giuseppe",
    title = "{Sulle funzioni ipergeometriche a pi\`u variabili}",
    doi = "10.1007/BF03015097",
    journal = "Rend. Circ. Mat. Palermo",
    volume = "7",
    number = "Suppl 1",
    pages = "111--158",
    year = "1893"
}

@article{Derkachov:1999pz,
    author = "Derkachov, S. E.",
    title = "{Baxter's Q-operator for the homogeneous XXX spin chain}",
    eprint = "solv-int/9902015",
    archivePrefix = "arXiv",
    doi = "10.1088/0305-4470/32/28/309",
    journal = "J. Phys. A",
    volume = "32",
    pages = "5299--5316",
    year = "1999"
}

@article{Sklyanin:1988yz,
    author = "Sklyanin, E. K.",
    title = "{Boundary Conditions for Integrable Quantum Systems}",
    reportNumber = "E-11-86",
    doi = "10.1088/0305-4470/21/10/015",
    journal = "J. Phys. A",
    volume = "21",
    pages = "2375--2389",
    year = "1988"
}

@article{Belliard:2010nhl,
    author = "Belliard, S. and Ragoucy, E.",
    title = "{Nested Bethe ansatz for y(gl(n)) open spin chains with diagonal boundary conditions}",
    eprint = "1001.1314",
    archivePrefix = "arXiv",
    primaryClass = "math-ph",
    doi = "10.1134/S1547477111030058",
    journal = "Phys. Part. Nucl. Lett.",
    volume = "8",
    pages = "218--227",
    year = "2011"
}

@article{schutz1997duality,
    author = "Schütz, Gunter M.",
    title = "{Duality relations for asymmetric exclusion processes}",
    doi = "10.1023/A:1004943003729",
    journal = "J. Stat. Phys.",
    volume = "86",
    number = "5",
    pages = "1265--1287",
    year = "1997"
}

@article{sasamoto1998one,
    author = "Sasamoto, Tomohiro and Wadati, Miki",
    title = "{One-dimensional asymmetric diffusion model without exclusion}",
    doi = "10.1103/PhysRevE.58.4181",
    journal = "Phys. Rev. E",
    volume = "58",
    number = "4",
    pages = "4181--4190",
    year = "1998"
}

@article{kipnis1982heat,
    author = "Kipnis, Claude and Marchioro, Carlo and Presutti, Errico",
    title = "{Heat flow in an exactly solvable model}",
    doi = "10.1007/BF01012963",
    journal = "J. Stat. Phys.",
    volume = "27",
    number = "1",
    pages = "65--74",
    year = "1982"
}

@article{vanicat2017exact,
    author = "Vanicat, Matthieu",
    title = "{Exact solution to integrable open multi-species SSEP and macroscopic fluctuation theory}",
    eprint = "1610.08388",
    archivePrefix = "arXiv",
    primaryClass = "cond-mat.stat-mech",
    doi = "10.1007/s10955-016-1705-7",
    journal = "J. Stat. Phys.",
    volume = "166",
    number = "5",
    pages = "1129--1150",
    year = "2017"
}

@article{casini2024duality,
    author = "Casini, Francesco and Frassek, Rouven and Giardin\`a, Cristian",
    title = "{Duality for the multispecies stirring process with open boundaries}",
    eprint = "2312.15532",
    archivePrefix = "arXiv",
    doi = "10.1088/1751-8121/ad5c5f",
    journal = "J. Phys. A",
    volume = "57",
    number = "29",
    pages = "295001",
    year = "2024"
}

@book{demasi2006mathematical,
    author = "De Masi, Anna and Presutti, Errico",
    title = "{Mathematical Methods for Hydrodynamic Limits}",
    series = "Lecture Notes in Mathematics",
    volume = "1500",
    publisher = "Springer",
    year = "2006",
    isbn = "978-3-540-33983-5"
}

@book{Humphreys1972,
    author = "Humphreys, James E.",
    title = "{Introduction to Lie Algebras and Representation Theory}",
    series = "Graduate Texts in Mathematics",
    volume = "9",
    publisher = "Springer-Verlag",
    address = "New York",
    year = "1972",
    isbn = "978-0-387-90052-0",
    doi = "10.1007/978-1-4612-6398-2"
}

@article{crampe2014integrable,
    author = "Cramp\'e, Nicolas and Ragoucy, Eric and Vanicat, Matthieu",
    title = "{Integrable approach to simple exclusion processes with boundaries. Review and progress}",
    eprint = "1408.5357",
    archivePrefix = "arXiv",
    doi = "10.1088/1742-5468/2014/11/P11032",
    journal = "J. Stat. Mech.",
    volume = "2014",
    number = "11",
    pages = "P11032",
    year = "2014"
}

@article{2024JPhA...57x5201K,
    author = "Kolyaskin, Dmitry and Mangazeev, Vladimir V.",
    title = "{Triangular solutions to the reflection equation for ${U}_{q}(\widehat{{sl}_n})$}",
    eprint = "2402.05442",
    archivePrefix = "arXiv",
    primaryClass = "math-ph",
    doi = "10.1088/1751-8121/ad4d2f",
    journal = "J. Phys. A",
    volume = "57",
    number = "24",
    pages = "245201",
    year = "2024"
}

@article{Kuniba:2016fpi,
    author = "Kuniba, A. and Mangazeev, V. V. and Maruyama, S. and Okado, M.",
    title = "{Stochastic R matrix for $U_q(A_n^{(1)})$}",
    eprint = "1604.08304",
    archivePrefix = "arXiv",
    primaryClass = "math.QA",
    doi = "10.1016/j.nuclphysb.2016.09.016",
    journal = "Nucl. Phys. B",
    volume = "913",
    pages = "248--277",
    year = "2016"
}

@article{Gwa-Spohn,
    author = "Gwa, Leh-Hun and Spohn, Herbert",
    title = "{Six-vertex model, roughened surfaces, and an asymmetric spin Hamiltonian}",
    doi = "10.1103/PhysRevLett.68.725",
    journal = "Phys. Rev. Lett.",
    volume = "68",
    number = "6",
    pages = "725--728",
    year = "1992"
}

@article{exact-harmonic,
    author = "Frassek, Rouven and Giardin\`a, Cristian",
    title = "{Exact solution of an integrable non-equilibrium particle system}",
    eprint = "2107.01720",
    archivePrefix = "arXiv",
    primaryClass = "math-ph",
    doi = "10.1063/5.0086715",
    journal = "J. Math. Phys.",
    volume = "63",
    number = "10",
    pages = "103301",
    year = "2022"
}

@article{Spohn_1983,
    author = "Spohn, H.",
    title = "{Long range correlations for stochastic lattice gases in a non-equilibrium steady state}",
    doi = "10.1088/0305-4470/16/18/029",
    journal = "J. Phys. A",
    volume = "16",
    number = "18",
    pages = "4275",
    year = "1983"
}

@article{2024JSP...191..150D,
    author = "de Masi, Anna and Ferrari, Pablo A. and Gabrielli, Davide",
    title = "{Hidden Temperature in the KMP Model}",
    eprint = "2310.01672",
    archivePrefix = "arXiv",
    primaryClass = "math.PR",
    doi = "10.1007/s10955-024-03363-z",
    journal = "J. Stat. Phys.",
    volume = "191",
    number = "11",
    pages = "150",
    year = "2024"
}

@book{dualityBook,
    author = "Giardin\`a, Cristian and Redig, Frank",
    title = "{Duality for Markov processes: a Lie algebraic approach}",
    series = "Grundlehren der mathematischen Wissenschaften",
    volume = "365",
    publisher = "Springer",
    address = "Cham",
    year = "2025",
    isbn = "978-3-032-04098-5",
    doi = "10.1007/978-3-032-04099-2"
}

@article{Antonenko:2024bgw,
    author = "Antonenko, P. V. and Belousov, N. M. and Derkachov, S. \`E. and Khoroshkin, S. M.",
    title = "{Reflection operator and hypergeometry I: SL(2,R) spin chain}",
    eprint = "2406.19862",
    archivePrefix = "arXiv",
    primaryClass = "math-ph",
    journal = "Zap. Nauchn. Semin.",
    volume = "532",
    pages = "5--46",
    year = "2024"
}

@article{2023JPhA...56A4001S,
    author = "Schütz, G. M.",
    title = "{A reverse duality for the ASEP with open boundaries}",
    eprint = "2211.02844",
    archivePrefix = "arXiv",
    primaryClass = "math.PR",
    doi = "10.1088/1751-8121/acda6a",
    journal = "J. Phys. A",
    volume = "56",
    number = "27",
    pages = "274001",
    year = "2023"
}

@article{2016arXiv160105770B,
    author = "Borodin, Alexei and Petrov, Leonid",
    title = "{Higher spin six vertex model and symmetric rational functions}",
    eprint = "1601.05770",
    archivePrefix = "arXiv",
    primaryClass = "math.PR",
    doi = "10.1007/s00029-016-0301-7",
    journal = "Selecta Math.",
    volume = "24",
    pages = "751--874",
    year = "2018"
}

@article{Derkachov:2008aq,
    author = "Derkachov, Sergey E. and Manashov, Alexander N.",
    title = "{Factorization of R-matrix and Baxter Q-operators for generic sl(N) spin chains}",
    eprint = "0809.2050",
    archivePrefix = "arXiv",
    primaryClass = "nlin.SI",
    doi = "10.1088/1751-8113/42/7/075204",
    journal = "J. Phys. A",
    volume = "42",
    pages = "075204",
    year = "2009"
}

@article{2011arXiv1111.4408B,
    author = "Borodin, Alexei and Corwin, Ivan",
    title = "{Macdonald processes}",
    eprint = "1111.4408",
    archivePrefix = "arXiv",
    primaryClass = "math.PR",
    doi = "10.1007/s00440-013-0482-3",
    journal = "Probab. Theory Rel. Fields",
    volume = "158",
    number = "1-2",
    pages = "225--400",
    year = "2014"
}

@article{2016CMaPh.343..651C,
    author = "Corwin, Ivan and Petrov, Leonid",
    title = "{Stochastic Higher Spin Vertex Models on the Line}",
    eprint = "1502.07374",
    archivePrefix = "arXiv",
    primaryClass = "math.PR",
    doi = "10.1007/s00220-015-2479-5",
    journal = "Commun. Math. Phys.",
    volume = "343",
    number = "2",
    pages = "651--700",
    year = "2016"
}

@article{barraquand2016q,
    author = "Barraquand, Guillaume and Corwin, Ivan",
    title = "{The q-Hahn asymmetric exclusion process}",
    eprint = "1501.03445",
    archivePrefix = "arXiv",
    doi = "10.1214/15-AAP1148",
    journal = "Ann. Appl. Probab.",
    volume = "26",
    number = "4",
    pages = "2304--2356",
    year = "2016"
}

@article{povolotsky2013integrability,
    author = "Povolotsky, Alexander M.",
    title = "{On the integrability of zero-range chipping models with factorized steady states}",
    eprint = "1308.3250",
    archivePrefix = "arXiv",
    primaryClass = "math-ph",
    doi = "10.1088/1751-8113/46/46/465205",
    journal = "J. Phys. A",
    volume = "46",
    number = "46",
    pages = "465205",
    year = "2013"
}

@book{liggett1985interacting,
    author = "Liggett, Thomas M.",
    title = "{Interacting Particle Systems}",
    series = "Grundlehren der mathematischen Wissenschaften",
    volume = "276",
    publisher = "Springer",
    year = "1985",
    isbn = "978-0-387-96069-9"
}

@article{borodin2016stochastic,
    author = "Borodin, Alexei and Corwin, Ivan and Gorin, Vadim",
    title = "{Stochastic six-vertex model}",
    eprint = "1407.6729",
    archivePrefix = "arXiv",
    primaryClass = "math.PR",
    doi = "10.1215/00127094-3166843",
    journal = "Duke Math. J.",
    volume = "165",
    number = "3",
    pages = "563--624",
    year = "2016"
}

@article{2014JPhA...47C5202L,
    author = "Lazarescu, Alexandre and Pasquier, Vincent",
    title = "{Bethe Ansatz and Q-operator for the open ASEP}",
    eprint = "1403.6963",
    archivePrefix = "arXiv",
    primaryClass = "math-ph",
    doi = "10.1088/1751-8113/47/29/295202",
    journal = "J. Phys. A",
    volume = "47",
    number = "29",
    pages = "295202",
    year = "2014"
}
\bibliographystyle{utphys2}

\end{document}